\documentclass[10pt]{article}
\def\commentflag{3}
\def\ct{3}
\usepackage{titlesec}
\usepackage{amssymb,amsmath,amsfonts,amsthm,enumerate,color,mathrsfs,hyperref,amssymb,extarrows,cite}
\usepackage[shortlabels]{enumitem}
\usepackage[toc,title,titletoc,header]{appendix}

\usepackage{graphicx}
\usepackage{indentfirst}
\usepackage{multicol}
\usepackage{booktabs} 
\usepackage{float} 
\usepackage{subfigure}

\newtheorem{theorem}{Theorem}[section]
\newtheorem{lemma}[theorem]{Lemma}
\newtheorem{proposition}[theorem]{Proposition}
\newtheorem{corollary}[theorem]{Corollary}

\numberwithin{equation}{section}

\theoremstyle{remark}
\newtheorem{remark}[theorem]{Remark}

\theoremstyle{definition}
\newtheorem{definition}[theorem]{Definition}

\usepackage{enumitem}
\setenumerate[1]{itemsep=0pt,partopsep=0pt,parsep=\parskip,topsep=5pt}
\setitemize[1]{itemsep=0pt,partopsep=0pt,parsep=\parskip,topsep=5pt}
\setdescription{itemsep=0pt,partopsep=0pt,parsep=\parskip,topsep=5pt}

\newcommand{\norm}[1]{\left\lVert #1 \right\rVert}

\newcommand{\mb}[1]{{\color{black} #1}}
\newcommand{\mbb}[1]{{\color{black} #1}}

\def \sss{\scriptscriptstyle}

\def\q {\quad}

\def \l{\langle}
\def \r{\rangle}

\def\bb{\begin{equation}
  \left\{\ 
   \begin{aligned} }
\def\ee{   \end{aligned}
  \right.
  \end{equation}}

\def\mm{ \left[
 \begin{matrix}}
\def\nn{\end{matrix} \right] } 

\def\p{\partial}
\def \dd{\cdot}
\def \t{\times}

\def\n{\nu}
\def \w {\widetilde}
\def \h{\hat}

\def\d{\delta}
\def \vp{\varphi}
\def \na{\nabla}
\def \ep{\varepsilon}
\def \lad{\lambda}

\def\re{\mathfrak{Re}}
\def\im{\mathfrak{Im}}

\def \curl{\nabla \times}

\def \ccurl{{\rm  curl}}
\def \ddiv {{\rm div}}

\def \td{\mathcal{T}_D^\ww}
\def \tbb{\mathcal{T}_D}

\def \kd{\mathcal{K}_D^\ww}
\def \wg{\w{\gamma}_n^{-1}}

\def \R{\mathbb{R}}
\def \C{\mathbb{C}}

\def \S{\mathbb{S}}
\def \AA{\mathbb{A}}
\def \P{\mathbb{P}}
\def \xx {\mathbb{X}}

\def \I{\mathbb{I}}

\def \TT{\mathcal{T}}
\def \W{\mathcal{W}}
\def \np{\mathcal{K}}
\def \ss{\mathcal{S}}
\def \M{\mathcal{M}}
\def \A{\mathcal{A}}
\def \L{\mathcal{L}}

\def \rr{\mathcal{R}}

\def \H{{\rm {\bf H}}}

\def \ww{\omega}
\def \wt{\omega_\tau}

\def \di{{\rm{\bf d}}}
\def \e{{\rm{\bf e}}}

\def \F{\mathcal{F}}

\def \ynm {Y_n^m}
\def \unm {U_n^m}
\def \vnm {V_n^m}
\def \ete {E^{TE}_{n,m}}
\def \etm {E^{TM}_{n,m}}
\def \eth {E^h_{n,m}}
\def \wete {\w{E}^{TE}_{n,m}}
\def \wetm {\w{E}^{TM}_{n,m}}
\def \anm {\alpha_{n,m}}
\def \bnm {\beta_{n,m}}

\def \gnm {\gamma_{n,m}}
\def \enm {\eta_{n,m}}
\def \hh{\mathcal{H}}
\def \jj{\mathcal{J}}

\def \hzz {{\bf H}_0(\ddiv 0, D)}
\def \hbb {H_{00}^{-1/2}(\p D)}

\def \hz {{\bf H}(\ddiv 0, D)}
\def \pd {\mathbb{P}_{{\rm d}}} 
\def \pw {\mathbb{P}_{{\rm w}}}
\def \poo {\P_{\lad_0}}

\def \ed {{\rm {\bf p}}}
\def \md {{\rm {\bf m}}}
\def \eq {{\rm {\bf Q}}}
\def \mq {{\rm {\bf M}}}
\def \eo {{\rm {\bf O}}}

\def \ei {{\rm {\bf E}}_0^i}

\def \ap {{\rm {\bf A(\vp)}}}

\newcommand{\tb}[1]{\mathcal{T}_{D,#1}}
\newcommand{\kb}[1]{\mathcal{K}_{D,#1}}

\def \kbdd {\mathcal{K}^{{\rm d,d}}_{D}}
\def \kbdw {\mathcal{K}^{{\rm d,w}}_{D}}
\def \kbwd {\mathcal{K}^{{\rm w,d}}_{D}}

\DeclareMathOperator{\tr}{tr}
\DeclareMathOperator{\ran}{ran}

\title{Mathematical analysis of electromagnetic scattering \\ by dielectric nanoparticles with high refractive indices}
\begin{document}
\author{
Habib Ammari\footnote{Department of Mathematics, ETH Z\"{u}rich, R\"{a}mistrasse 101, CH-8092 Z\"{u}rich, Switzerland. 
The work of this author is partially supported by the
Swiss National Science Foundation (SNSF) grant 200021-172483.
(habib.ammari@math.ethz.ch).}
\and Bowen Li\footnote{Department of Mathematics, The Chinese University of Hong Kong, Shatin, N.T., Hong Kong. (bwli@math.cuhk.edu.hk).}
\and Jun Zou\footnote{Department of Mathematics, The Chinese University of Hong Kong, Shatin, N.T., Hong Kong.
The work of this author was
substantially supported by Hong Kong RGC grant (Projects 14306719 and 14306718).
(zou@math.cuhk.edu.hk).}
}
\date{}
\maketitle

\begin{abstract}
In this work, we are concerned with the mathematical modeling of the electromagnetic (EM) scattering by 
arbitrarily shaped non--magnetic nanoparticles with high refractive indices. 
When illuminated by visible light, such particles can exhibit a very strong isotropic magnetic response, resulting from the coupling of the incident wave with the circular displacement currents of the EM fields. The main aim of this work is to mathematically illustrate this phenomenon.
We shall first introduce the EM scattering resolvent and the concept of dielectric subwavelength resonances. Then we 
derive the a priori estimates for the subwavelength resonances and the associated resonant modes. We also show the existence of resonances and obtain their asymptotic expansions in terms of the small particle size and the high contrast parameter. After that, we investigate the enhancement of the scattering amplitude and the cross sections when the resonances occur.
In doing so, we develop a novel multipole radiation framework  that directly separates the electric and magnetic multipole moments and allows us to clearly see their orders of magnitude and blow--up rates. We prove that at the dielectric subwavelength resonant frequencies, 
the nanoparticles with high refractive indices behave like the sum of the electric dipole and the resonant magnetic dipole. Some explicit calculations and numerical experiments are also provided to validate our general results and formulas. 

\end{abstract}

\section{Introduction}

In the past decades, the plasmonic resonant nanostructures with localized surface plasmons have been extensively studied in nanophotonics and used as the building blocks to make novel optical devices and complex metamaterials \cite{sarid2010modern}. The plasmonic materials that are usually made of metals may have  negative real permittivity in the visible spectral range and can be strongly coupled with the electric part of the incident field, inducing light enhancement and confinement \cite{jain2006calculated,baffou2010mapping}. It has been mathematically demonstrated that the quasi-static plasmonic resonance can be treated as an eigenvalue problem of the Neumann--Poincar\'{e} operator \cite{ammari2016surface,ammari2016plasmaxwell,ammari2017mathematicalscalar, add2,add7}. However, the metal structures suffer from a significant heat dissipation caused by the imaginary part of the electric permittivity and the anisotropic magnetic response. These physical properties severely limit the efficiency and functionality of the optical plasmonic nanodevice, which motivate 
the search for alternatives to the metal subwavelength resonators \cite{evlyukhin2012demonstration}. 

There is an immensely increasing interest in recent years in the study of the dielectric and semiconductor resonators with high refractive indices. These  
dielectric resonant structures have low absorbing properties and isotropic magnetic responses, making it
possible to complement the plasmonic elements in many potential applications \cite{kuznetsov2016optically,krasnok2012all}. In particular, the silicon nanoparticle has emerged as a popular choice of the subwavelength resonator for designing the dielectric metamaterials because of its high electric permittivity and very low dissipative loss
\cite{staude2017metamaterial}. It was observed experimentally \cite{evlyukhin2012demonstration,kuznetsov2016optically,garcia2011strong} that for a single spherical silicon nanoparticle, when the wavelength inside the particle is comparable to its diameter, the electric and magnetic dipole radiations can have comparable strengths, and the magnetic dipole resonance can be excited in the visible region. 
Nevertheless, compared to considerable evidence from engineering and physics literature, the mathematical understanding of the origin of the dielectric resonance and the mechanism underlying the strong magnetic response is still limited, apart from the case of spherical nanoparticles which has been well illustrated and studied by the Mie scattering theory
\cite{tzarouchis2018light}. In  \cite{ammari2019subwavelength}, the authors considered the Helmholtz equation and used asymptotic analysis to characterize the quasi--static dielectric resonant frequencies for the high contrast nanoparticles of arbitrary shape and obtained the corresponding first--order corrections. In \cite{meklachi2018asymptotic}, the case of nonlinear scatterers of the Kerr type was investigated and the corresponding asymptotic formulas for the resonances were derived.

This work is devoted to modeling the EM scattering of strongly coupled nanoparticles with high refractive indices by using the full Maxwell system. We shall provide a solid mathematical framework generalizing the existing ones for the scalar case \cite{meklachi2018asymptotic,ammari2019subwavelength,add1} to analyze the dielectric subwavelength resonances and the scattering amplitude. The newly obtained results can help rigorously confirm the aforementioned physical phenomena, which are unique to the electromagnetic wave, from the mathematical point of view. \mb{
Our analysis is based on the volume potential operators and the Lippmann--Schwinger equation, which can be naturally connected with the scattering resolvent associated with the Maxwell operator. It is possible to analyze the problem via other approaches, such as the variational formulation with the matched asymptotic expansions method \cite{add8,add9,add10,add11} and the boundary integral equation with the asymptotic analysis \cite{add12,add2,ammari2016surface}. But for both approaches, the calculations and analyses would be much more involved. More precisely, for the variational method we need to consider the transparent boundary condition involving DtN operators,
while for the boundary integral equation there are additional unknowns on the boundary to solve before we treat the original unknowns (EM fields).}

\mb{
In this work, we first define the  scattering resolvent of the Maxwell operator associated with high refractive indices as an operator--valued meromorphic function by exploiting the Lippmann--Schwinger representation and analytic Fredholm theorem. As in the scalar wave case \cite{lax1990scattering,gopalakrishnan2008asymptotic,zworski2017mathematical}, the EM scattering resonances are then defined as the poles of the (Maxwell) scattering resolvent.} In the current work, we are mainly interested in the dielectric subwavelength resonances, namely, those poles of the resolvent in the high contrast and quasi--static regime (cf.\,Definition \ref{def:drs_1}). It turns out that the integral operator $\td$ (cf.\,\eqref{def:electricpotential}) and its resolvent $(\lad - \td)^{-1}$ essentially determine the properties of the scattering resolvent and hence are of central importance for our purposes. 
We would like to stress that both the surface plasmon resonances and the dielectric subwavelength resonances are closely related to the integral operator $\td$, but with a very different underlying mechanism; see the discussion after Lemma \ref{lem:tbz1}. Regarding the spectral properties of $\td$, in \cite{costabel2010volume,costabel2012essential,costabel2015volume}, the authors gave a complete characterization of the essential spectrum of various integral operators arising from the EM scattering problems for both smooth and non--smooth (Lipschitz) domains. Along this direction, we showed in \cite{ammari2020superli} that the spectrum of $\td$ is a disjoint union of the essential spectrum on the real axis and the eigenvalues of finite type in the upper--half plane with $l^2$-summable imaginary parts. 


To analyze the dielectric subwavelength resonances, we consider the meromorphic family of operators $(\tau^{-1} - \TT_D^{\ww})^{-1}$ of the complex variable $\ww$ with the characteristic size $\d = 1$ by scaling. \mb{
By the known spectral results of $\td$ and the theory of trace class operators, we characterize the resonances in the concerned regime as zeros of an analytic function (cf.\,Proposition \ref{prop:gloana}). Therefore, one may regard subwavelength resonances $\ww(\tau)$ as a multivalued function (curve) with respect to $\tau$. Our first aim is to study the analytical properties of 
the subwavelength resonances $\ww(\tau)$, in particular, 
we will establish their a priori estimates and their existence, 
and then derive their asymptotic behaviors as $\tau \to \infty$.}  To do so, we first note that the standard asymptotic analysis was performed in the scalar case \cite{meklachi2018asymptotic,ammari2019subwavelength} to determine the limiting eigenvalue problem and further calculate the quasi--static resonances and their first--order corrections. However, it is highly nontrivial to extend such analysis to Maxwell's equations. 
The main difficulty comes from the fact that the leading--order operator of $\td$ has an infinite--dimensional kernel consisting of magnetostatic fields (cf.\,Lemma \ref{lem:tbz1}). To cope with this technical issue,
we shall utilize the Helmholtz decomposition 
for divergence--free vector fields to reformulate the considered operator--valued analytic function as a spectrally equivalent analytic family of operator matrices  (cf.\,\eqref{eigen:target0}). Then we perform the asymptotic analysis to 
derive the a priori estimates for the subwavelength resonances and the corresponding resonant modes in Theorem \ref{thm:priest0}, which shows that 
the subwavelength resonances $\ww$ are near the points $\sqrt{\lad_i \tau}^{\sss -1}$ and the resonant modes are almost transverse electric in the quasi--static and high contrast regime
and hence present the features of the magnetostatic fields. \mb{The a priori estimate \eqref{eq:prioriestfre1} suggests us to consider the scaled resonance $\h{\ww} = \sqrt{\tau} \ww$. With this new variable $\h{\ww}$, we are able to apply the  Gohberg--Sigal theory (i.e., generalized argument principle and generalized Rouch\'{e}'s theorem, see \cite{gohberg1990classes,gohberg1971operator}) to ensure the existence of dielectric subwavelength resonances (cf.\,Theorem \ref{thm:existissue}). Then by use of the theories of symmetric polynomials and algebraic functions, we prove that the multivalued functions $\ww(\tau)$ are algebraic functions with possible algebraic singularities at $\tau = \infty$ (cf.\,Theorem \ref{thm:splitting}).} \mb{We also discuss the physical role that the characteristic size $\d$ plays at the end of Section \ref{sec:resoana}. Briefly speaking, in order for the subwavelength resonances to be excited 
under visible light illumination, we need the condition $\tau \sim \d^2$, 
that is, the wavelength inside the particles is comparable to the size of the particles; see Appendix \ref{app:S_1}  for a set of experimental parameters.}  


\mb{Our second aim is to understand the structure of the scattering amplitude. To better connect the mathematical results with the physical phenomena, we restrict ourselves to the regime $\tau = \d^{-2}$ where the strong magnetic responses of the high contrast nanoparticles are experimentally observed \cite{evlyukhin2012demonstration,kuznetsov2016optically,garcia2011strong}, but we emphasize that our analysis does not essentially depend on this condition.} For our aim, we propose a new approach based on the Helmholtz decomposition to systematically investigate the multipolar response of nanoparticles of arbitrary shape, which is closely related to the classical Cartesian multipole moment expansion \cite{jackson1999classical}; see Remarks \ref{rem:cntnewandclas} and \ref{rem:compare_non}. By the new class of EM 
multipole moments in Definition \ref{def:EMmoment}
and the asymptotic analysis for the resolvent in Proposition \ref{prop:ledsolmain}, we give the quasi--static approximation of the scattering amplitude in Theorem \ref{thm:mainmultimom_3} which uniformly holds with respect to $\d$ and $\ww$. Then, under certain conditions, we  rigorously show in Corollary \ref{thm:mainmultimom_2} that 
when the incident frequency approximates the resonance, the scattering amplitude can be approximated by the resonant magnetic dipole radiation. \mb{This should be compared with 
the known result that the small particles typically behave 
like an electric dipole that scatters light symmetrically \cite{evlyukhin2016optical,add8,ammari2016surface}.} 
We shall also present an explicit calculation for the case of a single spherical nanoparticle, which generalizes some results in \cite{add1}
and also serves to validate our general results for particles of arbitrary shape. It is interesting to note that for the spherical domain, the quasi--static resonances can be characterized by the zeros of spherical Bessel functions (cf.\,Theorem \ref{thm:eigvalspher}).

The work is organized as follows. In Section \ref{sec:prosetanddef}, we formulate the EM scattering problem by nanoparticles with high refractive indices and 
introduce the scattering resolvent and 
the dielectric subwavelength resonances. Sections \ref{sec:resoana} and \ref{sec:multirad} are devoted to the resonance analysis and the far--field analysis for the scattering problem, respectively. The explicit calculations for the spherical nanoparticle are provided in Section \ref{sec:spernano} which helps validate our general results and formulas. We end this work with some concluding remarks and discussions. 


\mb{We close this section with a list of notation and conventions.}

\begin{itemize}
\mbb{
\item $d \sigma$ is the surface integral element. 
\item $\S$ is the unit sphere in $\R^3$, and $\S_r$ is the sphere in $\R^3$ with radius $r$ centered at the origin. For a vector $x \in \R^3$, we denote its polar form by $(|x|,\h{x})$ with $\h{x}:=x/|x| \in \S$. 
\item $B(x,r)$ is a ball in $\R^d$ or $\C$ with radius $r$ centered at the point $x$. The ambient space is clear from the context.  
\item  
For two vectors $u \in \R^n$ and $v \in \R^m$, $u \otimes v$ is a $n \times m$ matrix given by $(u \otimes v)_{ij} = u_i v_j$. 
\item For a (block) matrix $A$, we write $A_{i,j}$ for its $(i.j)$th (block) element. 
\item We use $(\dd,\dd)_H$ and $\oplus$ for the inner product of a Hilbert space $H$ and the orthogonal direct sum, respectively. We identify the direct sum with the product of its factors.
\item Let $f$ and $g$ be in a normed vector space and depend on a parameter $\epsilon$. We write $f(\epsilon) \backsimeq g(\epsilon)$, $\epsilon \to \epsilon_0$ if the relative error between $f$ and $g$ tends to zero:
\begin{align*} 
\lim_{\epsilon \to \epsilon_0} \frac{\norm{f(\epsilon) - g(\epsilon))}}{\norm{g(\epsilon)}} = 0\,.
\end{align*}
\item For $f$ in a normed vector space, we write $f = O(\epsilon)$ if $\norm{f}\le C |\epsilon|$ for some constant $C > 0$ independent of $\epsilon$, and $f = o(\epsilon)$ if $\norm{f}/|\epsilon| \to 0$ as $\epsilon \to 0$, and $f \sim \ep$ if $ C_1 \le \norm{f}/|\epsilon| \le C_2$ for $C_1, C_2 > 0$ independent of $\epsilon$. 
\item For real $a,b >0$, we write $a \ll b$ if $a \le C b$ for small enough $C$ independent of $a,b$, and $a \gg b$ if $a \ge C b$ for large enough $C$ independent of $a,b$.}
\mb{\item We use the bold typeface to indicate the spaces consisting of vector fields in $\R^3$. We shall need 
several important spaces in our subsequent studies: 
the space of $L^2$--integrable functions with compact support ${\bf L}^2_{{\rm comp}}(\R^3)$, 
the space of divergence--free vector fields ${\bf H}(\ddiv0, D) = \{\vp \in {\bf L}^2(D)\,;\ \ddiv \vp = 0\}$ and its subspace 
$${\bf H}_0(\ddiv0, D) = \{\vp \in {\bf H}(\ddiv0, D)\,;\ \n \dd \vp = 0\ \text{on}\ \p D\}\,,$$
the space of irrotational vector fields ${\bf H}(\ccurl0, D) = \{\vp \in {\bf L}^2(D)\,;\ \curl \vp = 0\}$, the space 
$$
\hbb = \{u \in H^{-1/2}(\p D)\,;\ \int_{\p D_i} u d \sigma = 0\,,\  D_i\ \text{is a connected component of}\ D\}\,,
$$
and the space ${\bf L}_{\rm T}^p(\S)$ of $p$--integrable tangential vector fields on $\S$. We would like to remark 
that the meaning of the notation $H^{\sss -1/2}_{00}(\p D)$ is completely different from the one used in the trace theory (cf.\cite{add14}).}
\mb{\item For an analytic family of linear operators $A(\lad)$ on an  open  set $V \subset \C^d$, $\lad$ is a characteristic value if $\ker A(\lad) \neq \{0\}$, and we call the elements of $\ker A(\lad)$ the eigenfunctions associated with $\lad$.}
\end{itemize}
 
 





\section{EM scattering resolvent and resonances}  \label{sec:prosetanddef}

In this section, we introduce the EM scattering problem by strongly coupled nanoparticles with high refractive indices in $\R^3$. We will mathematically define the scattering resolvent of the Maxwell operator associated with a dielectric obstacle, via the volume potential operators and the Lippmann--Schwinger equation. These concepts further allow us to introduce the dielectric subwavelength resonances, which is the main object of this work.

\mb{Without loss of generality, let $D$ be a (reference) bounded open set, containing the origin, with the smooth boundary $\p D$ and 
the exterior unit normal vector $\n$. We denote the union of all nanoparticles by 
\begin{align} \label{def:particlescaling}
    D_\d = \d D, 
\end{align}
where $0 < \d \ll 1$ is the characteristic size of the small particles. We also assume that both $D$ and $D^c := \R^3 \backslash \bar{D}$ may have multiple connected components; see Figure \ref{setup}.
\begin{figure}[!htbp]
 \centering
\includegraphics[clip,width=0.4\textwidth]{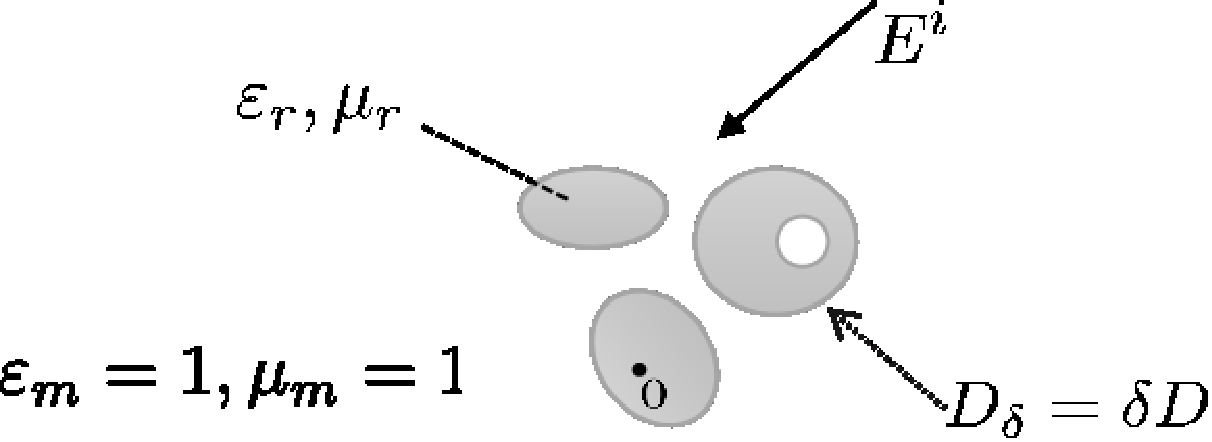}  
\caption{Setup for the scattering problem} \label{setup}
\end{figure}

}



By choosing the appropriate physical units, we let the electric permittivity $\ep_{{\rm m}}$
and the magnetic permeability $\mu_{{\rm m}}$ in the free space be \mb{the constant one}. The  physical properties of the nanoparticles are characterized by the relative electric permittivity $\varepsilon_r$ and the magnetic permeability $\mu_r$. To model high refractive indices,
we assume that the nanoparticles are 
non--magnetic, i.e., $\mu_r = 1$, and $\ep_r(x) \in L^\infty(\R^3)$ takes the form: 
\begin{align*}
    \ep_r = n_\d^2 = 1 + \tau \chi_{D_\d} \q \text{in}\ \R^3\,,
\end{align*}
where the contrast parameter $\tau$ satisfies 
\begin{equation} \label{asp:highcont}
  \mb{\tau \in S_{\tau_0} := \{\tau \in \C\backslash (-\infty, 0]\,;\ |\tau| \ge \tau_0\}\,,\q \tau_0 \gg 1\,.}
\end{equation}
Here $n_\d$ is the refractive index in $\R^3$ and $\chi_{D_\d}$ is the characteristic function of the set $D_\d$.
\mb{It is necessary to give some remarks on the above assumption. First, for a natural dielectric object, we typically have $\re \tau > 0$ and $\im \tau \ge 0$; see Figure
\ref{fig:parasili} for the possible values of $\ep_r$ of silicon. 
The assumption  \eqref{asp:highcont} is mainly for mathematical generality and also for ease of analysis. Second, the half--line $(-\infty,0]$ is removed so that $\sqrt{\tau}$ is a well--defined analytic function on $S_{\tau_0}$.}


Let $(E^i, H^i)$ be the incident plane wave: 
\begin{equation} \label{eq:indplawave}
    E^i = \ei e^{i\ww \di \dd x}, \quad H^i = \frac{1}{i\ww}\curl E^i = \di \t \ei e^{i \ww \di \dd x}\,,  \q \ww \neq 0\,,
\end{equation}
with the incident direction $\di \in \S$ and the polarization vector $\ei \in \S$ satisfying 
$\ei \dd \di = 0$. With these notions, the EM scattering problem, by nanoparticles with high refractive indices, is modelled by
\begin{equation}\label{eq:model}
\left\{ 
    \begin{array}{ll}
 \curl E = i \ww H &   \text{in} \ \R^3 \backslash \p D_\d\,, \\
   \curl H = - i \ww n_\d^2 E      &  \text{in} \ \R^3 \backslash \p D_\d\,, \\
   \left[\n \t E\right] = 0\,,\ \left[\n \t H\right] = 0 &  \text{on} \  \p D_\d \,,
    \end{array}
    \right. 
\end{equation}
where $[\dd] = \dd|_- - \dd|_+$ is the jump across the boundary 
and the subscripts $\pm$ denote the limits from outside and inside $D_\d$,  respectively. The system \eqref{eq:model} is complemented by the outgoing condition for the scattered wave $(E^s, H^s) := (E - E^i, H - H^i)$. For real $\ww > 0$, it is the well known Silver--M\"uller radiation condition: 
\begin{equation*} 
    \lim_{|x|\to\infty}|x|\dd\left(H^s \t \h{x} - E^s\right) = 0\,.
\end{equation*}
When $\ww$ is complex, one may specify the outgoing condition by the meromorphic continuation of the scattering resolvent, or, more simply, by reformulating \eqref{eq:model} as a Lippmann--Schwinger equation (cf.\cite{gopalakrishnan2008asymptotic, meklachi2018asymptotic}). \mb{This will become clear in the following exposition.}

We introduce the outgoing resolvent of the free (vector) Laplacian $- \Delta$: 
\begin{align*}
    \np^\ww[\vp]:= (- \Delta - \ww^2)^{-1}[\vp] = \int_{\R^3} g(x-y,\ww) \vp(y) dy\  :  \ {\bf L}^2_{{\rm comp}}(\R^3) \to {\bf H}^2_{{\rm loc}}(\R^3)\,,
\end{align*}
which is clearly an entire family of operators. Here, the Schwartz kernel of $\np^\ww$,
\begin{align} \label{def:freegreen}
    g(x,\ww) := \frac{e^{i\ww|x|}}{4 \pi |x|}\,, \q \ww \in \C\,,
\end{align}
is the (outgoing) Green's function of $-(\Delta + \ww^2)$ in $\R^3$. \mbb{Then the free resolvent of the Maxwell operator $\M_0: = \curl \curl$ can be defined by
\begin{align} \label{def:m0resol}
  (\M_0 - \ww^2)^{-1}[f] := \frac{1}{\ww^2} \TT^\ww[f]\,,\q  \ww \in \C \backslash \{0\}\,, 
\end{align}
where 
\begin{align} \label{def:tw}
  \TT^\ww[f]: = (\ww^2 + \nabla \ddiv) \np^\ww[\varphi]\  : \  {\bf L}_{{\rm comp}}^2 (\R^3) \to {\bf H}_{{\rm loc}}(\ccurl,\R^3)\,.
\end{align}
We note from the definitions \eqref{def:m0resol} and \eqref{def:tw} 
that $\TT^\ww$ is entire in $\ww$,  while $(\M_0 - \ww^2)^{-1}$ is only analytic on $\C \backslash \{0\}$ (one may also recall the closely related fact that $\M_0$ has an infinite--dimensional kernel \cite{birman19872}). For our subsequent use, we define volume potential operators:
\begin{align*}
    \kd: = \np^\ww|_{{\bf L}^2(D)}\,, \q     \td: = \TT^\ww|_{{\bf L}^2(D)}\,.
\end{align*}
\mb{By approximating ${\bf L}^2(D)$ by $C_{{\rm c}}^\infty(D)$ and using integration by parts, it is easy to verify} 
\begin{align} \label{auxpropty_1}
    \ddiv \td[E] = - \ddiv E \, : \, {\bf L}^2 (D) \to {\bf H}^{-1}(D)\,.
\end{align}
Then, it follows that 
\begin{align} \label{def:electricpotential}
   \td\, : \, {\bf H}(\ddiv0,D) \to {\bf H}(\ddiv0,D)\,.
\end{align}
$\np_{D_\d}^\ww$ and $\TT_{D_\d}^{\ww}$, $0 < \d \ll 1$, can be defined in the same way.}

For $\ww \neq 0$, eliminating the magnetic field $H$ from the system \eqref{eq:model}, we obtain the second--order equation for the scattered electric field $E^s$: 
\begin{align} \label{eq:model2}
    \curl \curl E^s - \ww^2 E^s = \ww^2 \tau \chi_{D_\d} E\,.
\end{align}
Then, using $\TT_{D_\d}^\ww$ defined above, we readily write 
the Lippmann--Schwinger equation for \eqref{eq:model2}: 
\begin{equation}
    E^s(x) = E(x)-E^i(x) = \tau \TT_{D_\d}^\ww[E](x)\,, \quad x \in \R^3\,,
\end{equation}
which directly gives us the electric field inside $D_\d$ \mb{(if $1 - \tau \td$ is invertible)}:
\begin{equation}  \label{eq:liprep_1}
    E = (1 - \tau \TT_{D_\d}^\ww)^{-1}[E^i] \q \text{in}\  D_\d\,,
\end{equation}
and the scattered field outside $D_\d$:
\begin{equation} \label{eq:liprep_2}
     E^s = \tau \TT_{D_\d}^\ww [E] \q \text{in}\   \R^3 \backslash \bar{D_\d}\,.
\end{equation}

\mbb{
We next define the scattering  resolvent for the problem \eqref{eq:model2} (equivalently, for \eqref{eq:model}). For this, we consider the general form of \eqref{eq:model2}: 
\begin{equation} \label{mod:eq} 
\left\{ 
\begin{aligned}
   &\curl  \curl E - \ww^2 n_\d^2 E = f \q \text{in} \ \R^3\,,   \\
     &  E \ \text{satisfies the outgoing condition\,,}
\end{aligned} \right.
\end{equation}    
where $f \in {\bf L}^2_{{\rm comp}}(\R^3)$ is the given source. In the case of incident plane wave, it is clear that $f = \ww^2 \tau \chi_{D_\d} E^i$.  By the Lax--Milgram theorem, the equation \eqref{mod:eq} is uniquely solvable in ${\bf H}(\ccurl,\R^3)$ for $\im \ww > 0$, $\tau > 0$, and $f \in {\bf L}^2(\R^3)$. In this case, the resolvent is defined as the solution operator of \eqref{mod:eq}:
\begin{align*}
    E = (\M_0 - \ww^2 n_\d^2)^{-1}[f] \, :\,  {\bf L}^2(\R^3) \to {\bf H}(\ccurl,\R^3)\,.
\end{align*}
Moreover, since $g(x,\ww)$ is exponentially decaying for $\im \ww > 0$, $\TT^\ww$
defined in \eqref{def:tw} can be extended to ${\bf L}^2(\R^3)$. Then the Lippmann--Schwinger representation for the problem \eqref{mod:eq} follows: 
\begin{align} \label{auxeq_1}
    E = \tau \TT_{D_\d}^\ww [E] + \ww^{-2}\TT^\ww[f] \in {\bf H}(\ccurl,\R^3)\,, \q \im \ww > 0 \,, \ \tau > 0 \,.
\end{align}
A direct reformulation of \eqref{auxeq_1} gives
\begin{align} \label{eq:lseres}
   E = (\M_0 - \ww^2 n_\d^2)^{-1}[f] :=  \TT_{D_\d}^{\ww} \w{\A}_\d(\tau,\ww)^{-1} \big[\ww^{-2} \chi_{D_\d} \TT^\ww[f]\big] + \ww^{-2}\TT^\ww[f]\,,
\end{align}
with 
\begin{align} \label{eq:tautd}
    \w{\A}_\d(\tau,\ww) := \tau^{-1} - \TT_{D_\d}^{\ww}  \, :\,  {\bf L}^2(D_\d) \to {\bf H}(\ccurl,D_\d)\,.
\end{align}

To meromorphically extend
$(\M_0 - \ww^2 n_\d^2)^{-1}$ with respect to $(\ww,\tau)$ from the regime: $\im \ww > 0$, $\tau > 0$, we recall that the essential spectrum of $\TT_{D_\d}^\ww$, for any $\ww \in \C$, \mb{is given by (cf.\cite{costabel2012essential})}
\begin{equation*}
    \sigma_{ess}(\TT_{D_\d}^{\ww}) = \{-1,-\frac{1}{2},0\}\,, 
\end{equation*} 
which implies that, for a given contrast $\tau$ with $\tau^{-1} \in \C \backslash  \sigma_{ess}(\TT_{D_\d}^{\ww})$, $\w{\A}_\d(\tau, \ww)$ is a Fredholm operator. Hence, $\w{\A}_\d(\tau, \ww)$ is an analytic family of Fredholm operators on ${\bf L}^2(D_\d)$ for $(\tau,\ww) \in \C \backslash \{-1,-2\} \t \C$. We observe  from Lemma \ref{lem:tbz1} below that  $\w{\A}_\d(\tau, 0) = \tau^{-1} -  \TT^0_{D_\d}$ is invertible for $\tau > 0$. As a consequence of the analytic Fredholm  theorem \cite{gohberg1990classes} and its  multidimensional version  \cite{kuchment2012floquet,stessin2011analyticity,MichaelTaylor2010s}, we have 
the following result. 

\begin{lemma} \label{thm:dsrdiscrete} $ $ 
\begin{enumerate}
    \item For a given contrast $\tau \in \C \backslash \{ -1, -2\}$, the operator--valued function $\w{\A}_{\d}(\tau,\dd)^{-1}$ can be extended to a meromorphic function on $\C$ with a discrete set of poles given by the characteristic values of $\w{\A}_\d(\tau,\dd)$.
Moreover, there hold
\begin{align*}
     {\rm ind}\,\w{\A}_\d(\tau,\ww) = 0\q \text{and} \q \dim \ker \w{\A}_\d(\tau, \ww) < +\infty\,, \q \text{for}\ \ww \in \C\,.
\end{align*}
\item The set
$$   
\{(\tau,\ww) \in (\C \backslash \{-1,-2\}) \t \C\,; \  \w{\A}_\d(\tau,\ww)\ \text{is not invertible}\}
$$
is either empty, or an analytic subset of $\C^2$ of codimension 1. 
\end{enumerate}
\end{lemma}

As a corollary of the above lemma, we see from \eqref{eq:lseres} that, for any $\tau \notin \{-1,-2\}$, $(\M_0 -\ww^2 n_\d^2)^{-1}$ has a meromorphic continuation to $\C \backslash \{0\}$: 
\begin{align} \label{def:resolmep}
     (\M_0 - \ww^2 n_\d^2)^{-1} = & \ww^{-2} \left\{ \TT_{D_\d}^\ww \w{\A}_\d(\tau, \ww)^{-1} (\chi_{D_\d}\TT^\ww) + \TT^\ww\right\} \notag \\&  : \,  {\bf L}_{{\rm comp}}^2(\R^3) \to {\bf H}_{{\rm loc}}(\ccurl,\R^3) \,, \q \ww \in \C \backslash \{0\}\,.
\end{align}
\begin{definition}
For $\tau \notin \{-1,-2\}$, we call $(\M_0 - \ww^2 n_\d^2)^{-1}$ defined in \eqref{def:resolmep} the scattering resolvent of the Maxwell operator with the refractive index $n_\d$, and call its poles the EM scattering resonances.
\end{definition}

\mb{Before proceeding further, we note from \cite[Lemma\,3.1]{ammari2020superli} that 
$\ker(\TT_{D_\d}^\ww) =\ker (\TT_{D_\d}^{- \bar{\ww}}) = 0$, and hence there holds $\overline{\ran(\chi_{D_\d}\TT^\ww)} = \ker (\TT_{D_\d}^{- \bar{\ww}})^{\perp} = {\bf L}^2(D_\d)$.} It then follows that $(\M_0 - \ww^2 n_\d^2)^{-1}$, 
as a function of $\ww$, has the same poles as $\w{\A}_\d(\tau,\dd)^{-1}$. 
We also observe that $\ww$ is a pole of $\w{\A}_\d(\tau, \dd)^{-1}$ if and only if the eigenvalue problem:  
\begin{equation*}   
 (1 - \tau \TT_{D_\d}^{\ww})[E] = 0
\end{equation*}
has nontrivial solutions $E \in {\bf L}^2(D_\d)$, which, \mb{by \eqref{auxpropty_1} and $\tau \neq -1$},
must be divergence--free, i.e., $E \in {\bf H}(\ddiv0, D_\d)$. Moreover, the following scaling property: 
\begin{equation*}
    \TT_{D_\d}^{\ww}\big[E\big](x) = \TT_{D_\d}^{\ww}\big[E\big]\big(\d \w{x}\big) = \mathcal{T}_D^{\d\ww}\big[\w{E}\big]\big(\w{x}\big)\,, 
\end{equation*}
can be easily checked by change of variables: $x = \d \w{x}$ and $\w{E}\big(\w{x}\big) := E\big(\d \w{x}\big)$. Therefore, we define (cf.\,\eqref{def:electricpotential})
\begin{align} \label{redef:Ad}
      \A_\d(\tau,\ww) := \tau^{-1} -  \TT_{D}^{\d \ww}  \, :\, {\bf H}(\ddiv0, D) \to {\bf H}(\ddiv0, D)\,,
\end{align}
and have the following result.



\begin{proposition} \label{coro:discrete}
For $\tau \notin \{-1,-2\}$, the EM scattering resonances are given by the poles of $\A_\d(\tau,\dd)^{-1}$. 
\end{proposition}}

In this work, we are interested in the properties of the resolvent $\A_\d(\tau,\ww)^{-1}$ and its poles in the 
high contrast regime \eqref{asp:highcont} and the quasi--static regime (i.e., the wavelength in the free space is much \mb{larger than} the characteristic size of the nanoparticles):
\begin{equation} \label{eq:subwave}
S_\d := \left\{\ww\in \C\,;\ |\ww|\ll 2\pi\d^{-1}\right\},
\end{equation}
which are fundamental for illustrating the strong magnetic radiation by nanoparticles with high refractive indices \cite{kuznetsov2016optically}. For ease of exposition, we introduce the following definition.

\begin{definition}  \label{def:drs_1}
\mbb{For a given $\tau \in S_{\tau_0}$, we define the set of dielectric subwavelength resonances by} 
\begin{align*}
    \mbb{\Omega(\d,\tau) := \{\ww \in S_\d\,; \ \ww \ \text{is a pole of}\ \A_\d(\tau,\dd)^{-1}\}\,.}
\end{align*}
\end{definition}

\section{Dielectric subwavelength resonances}  \label{sec:resoana} 

\mb{One of the main aims of this work is to understand the structure of the set $\Omega(\d,\tau)$ of subwavelength resonances. By definitions, one may readily see that $\ww$ is a pole of $\A_\d(\tau,\dd)^{-1}$ if and only if $\d \ww$ is a pole of $\A(\tau,\dd)^{-1}$, where 
\begin{align} \label{def:adt}
    \A(\tau,\ww) := \A_\d(\tau,\ww)|_{\d = 1} = \tau^{-1} - \td\,.
\end{align}
Therefore, to analyze $\Omega(\d,\tau)$, it suffices to consider the case $\d = 1$.  

This section is devoted to the analysis of the set $\Omega(1,\tau)$. 
We first provide some necessary preliminaries in Section \ref{subsec:asypre}, and then derive the a priori estimates for dielectric subwavelength resonances 
and the associated eigenfunctions in Section \ref{subsec:priori}. 
We will address in Section \ref{sec:resoexit} the existence of subwavelength resonances, 
and consider their asymptotic behaviors as the contrast $\tau$ tends to infinity. At the end of this section, we discuss the role that the parameter $\d$ plays physically and how to translate the obtained results to the case $0 < \d \ll 1$ and the set $\Omega(\d,\tau)$.} 

\mbb{
Before we proceed to the next section, we consider the set $\Omega$:
\begin{align} \label{eq:targetdone}
\Omega := \{(\tau,\ww) \in (\C \backslash \{-1,-2\}) \t \C\,;\  \A(\tau,\ww) \ \text{is not invertible}\}\,.
\end{align}
Similarly to Lemma \ref{thm:dsrdiscrete}, $\Omega$ is an analytic set of codimension $1$, and hence is locally given by the zero set of an analytic function of variables $\ww$ and $\tau$. We now claim that the analytic function can be globally defined. 

\begin{proposition} \label{prop:gloana}
There exists an analytic function $f(\tau,\ww)$ on $(\C \backslash (-\infty, -1]) \t \C$ such that 
\begin{align*} 
    \Omega \bigcap ((\C \backslash (-\infty, -1]) \t \C) = \{(\tau,\ww)\,; \ f(\tau,\ww) = 0\}.
\end{align*}
\end{proposition}

\begin{proof}
Note that $\tau^{-1} - \TT_D$, $\tau \in \C \backslash (-\infty, -1]$, is invertible (see Lemma \ref{lem:tbz1} below). We write 
\begin{align*}
    \A(\tau,\ww) =  (\tau^{-1} - \TT_D) (1 + (\tau^{-1} - \TT_D)^{-1}  (\TT_D - \TT_D^{\ww})),  
\end{align*}
and see that $\Omega \cap ((\C \backslash (-\infty, -1]) \t \C)$ can also be given by 
\begin{align*}
 \left\{ (\tau,\ww) \in (\C \backslash (-\infty, -1]) \t \C\,;\ \W(\tau,\ww) := 1 + (\tau^{-1} - \TT_D)^{-1}  (\TT_D - \TT_D^{ \ww}) \ \text{is not invertible}\right\}.
\end{align*}
Since $\TT_D - \TT_D^{\ww}$ has a smooth integral kernel, we conclude that it is a trace class operator (cf.\cite[Chapter 30]{pdlax}). Recall the facts that the trace class is a two--sided ideal in the space of bounded linear operators; for a trace class operator $A$, the Fredholm determinant of $1 + A$, $\det (1 + A)$, is well--defined; and $I + A$ is singular if and only if $\det (I + A) = 0$; see \cite{pdlax,gohberg1990classes}. Therefore, we can define 
$\det (\W(\tau,\ww))$, which
is analytic by \cite[Lemma 2]{stessin2011analyticity} and gives the desired function.   
\end{proof}

Let $\Omega(\tau) := \{\ww \in \C\,;\ (\tau,\ww) \in \Omega\}$ be the section of $\Omega$ at $\tau \in \C \backslash \{-1,-2\}$. By Definition \ref{def:drs_1} and Proposition \ref{prop:gloana} above, it is clear that the set of subwavelength resonances $\Omega(1,\tau)$ are given by the elements of $\Omega(\tau)$ near the origin for $\tau \in S_{\tau_0}$, which are discrete with finite cardinality (could be empty) and 
can be characterized by the zeros of $f(\tau,\dd)$. In view of this, it may be convenient and illuminating to regard $\Omega(1,\tau)$ as a multivalued function.}

\subsection{Asymptotics and preliminaries} \label{subsec:asypre}
Let us first prepare some analysis tools and recall some useful known spectral results. To start with, we give the asymptotics of the operators $\np_D^{\ww}$ and $\TT_D^{\ww}$, which follows from the Taylor expansion of the Green's function $g(x, \ww)$:
\begin{equation} \label{eq:expgreengn}
    g(x, \ww) = \sum_{n = 0}^\infty \ww^n g_n(x) \quad \text{with}\q g_n(x) := \frac{i^n|x|^{n-1}}{4 \pi n!}\,.
\end{equation}
\begin{lemma} \label{lem:asytbw}
 For the bounded linear operators $\np_D^{\ww}$ on ${\bf L}^2(D)$ and $\TT_D^{\ww}$ on $\bf{H}(\ddiv0,D)$, $\ww \in \C$, it holds that
\begin{align*}
   \np_D^{\ww} = \sum_{n = 0}^\infty \ww^n \kb{n} \,, \quad \TT_D^{\ww} = \sum_{n = 0}^\infty  \ww^n \tb{n}\,,
\end{align*}
where the series converge in the operator norm
for each $\ww$.
The operators $\kb{n}$, $n \ge 0$, are defined by 
\begin{equation*}
    \kb{n}[\vp] := \int_D g_n(x-y) \vp(y) d y\,,
\end{equation*}
with $g_n$ being given in \eqref{eq:expgreengn}. The operators $\tb{n}, n \ge 0$, are defined by 
\begin{align}
    &\tb{0}[\vp]  := \na \ddiv \int_D \frac{1}{4 \pi |x-y|}\vp(y)dy\,, \quad \tb{1}[\vp] := i  \na \ddiv \frac{1}{4 \pi}\int_D \vp(y)dy = 0\,, \notag \\
    & \tb{n}[\vp] := \kb{n-2}[\vp] + \na \ddiv \kb{n}[\vp]\,, \quad n \ge 2\,. \label{eq:operexp}
\end{align}
\end{lemma}

In what follows, we shall denote $\tb{0} = \TT_D^0$ and $\kb{0} = \np_D^0$ by $\tbb$ and $\np_D$, respectively, for simplicity of notation. By Lemma \ref{lem:asytbw} above, 
we consider the asymptotic expansion of the operator $\tau^{-1} - \TT_D^{\ww}$: 
\begin{equation} \label{eq:asymeig}
    \tau^{-1} - \TT_D^{\ww} = \tau^{-1} - \tbb - \ww^2 \tb{2} - \ww^3 \tb{3} + O(\ww^4) \q \text{as} \ \ww \to 0\,.
\end{equation} 
As it was stated in the introduction, a major technical difficulty we will meet when using \eqref{eq:asymeig} to find the dependence of $\ww$ on $\tau$ is that the leading--order operator $\TT_D$ of $\tbb^{\ww}$ has an infinite--dimensional kernel (cf.\,Lemma \ref{lem:tbz1}), so that the standard perturbation argument, as well as the Gohberg--Sigal theory, can not be easily applied. To remedy this, we shall use the $L^2$--Helmholtz decomposition to transform the operator $\tbb^{\ww}$ into an operator matrix and then conduct the asymptotic analysis. This allows us to clearly see the asymptotic properties of $\tbb^{\ww}$ on different components of divergence--free vector fields in terms of $\tau$. \mb{We next briefly recall the Helmholtz decomposition; see \cite{amrouche1998vector,monk2003finite,girault2012finite} for more detailed discussion.}

\begin{lemma} \label{prop:helmddifree}
${\bf H}(\ddiv0,D)$ has the $L^2$--orthogonal decomposition:
   \begin{equation}\label{eq:helmddifree}
    {\bf H}(\ddiv0,D) = \hzz \oplus W = \ccurl X_N^0(D)\oplus K_T(D) \oplus W\,,
   \end{equation}
   where $W$ is the space consisting of the gradients of harmonic $H^1$--functions, the space $X_N^0(D)$ is defined by  
 \begin{equation*}
    X_N^0(D) := {\bf H}_0(\ccurl,D) \bigcap {\bf H}(\ddiv0,D)\,,
\end{equation*}
and $K_T(D)$ is the tangential cohomology space defined by
\begin{equation*}
    K_T(D) := \{u \in {\bf H}_0(\ddiv,D)\,;\, \curl u = 0,\ \ddiv u = 0\ \text{in} \ D\}\,,
\end{equation*}
with its dimension equal to the genus of $D$. A field $\phi \in \ccurl X_N^0(D)$ can uniquely determine a potential $A$ in the quotient space $\w{X}_N^0(D) : =  X_N^0(D)/K_N(D)$ such that $\phi = \curl A$, where $K_N(D): =\{u \in {\bf H}_0(\ccurl0,D)\,;\ \ddiv u = 0\ \text{in} \ D\}$ is the normal cohomology space.
\end{lemma}

We now introduce the Neumann--Poincar\'{e} operator $\np_{\p D}^*: H^{-1/2}(\p D) \to H^{-1/2}(\p D)$ and the single layer potential $\ss_{\p D} : H^{-1/2}(\p D) \to H^{1/2}(\p D)$ by 
\begin{equation*}
    \np_{\p D}^*[\vp] := \int_{\p D} \frac{\p}{\p \n_x} \frac{1}{4 \pi |x-y|}\vp(y) d\sigma(y)\q \text{{\rm and}}\q  \ss_{\p D}[\vp] := \int_{\p D}\frac{1}{4 \pi |x-y|}\vp(y) d\sigma(y) \,,
\end{equation*}
respectively. \mb{We recall the normal trace formula for $\na \ss_{\p D}$ \cite{colton2013integral,verchota1984layer,khavinson2007poincare}}:
 \begin{equation} \label{eq:trasinpo} 
    \frac{\p}{\p \n}\ss_{\p D}[\vp]_{\pm} = (\mp \frac{1}{2} + \np_{\p D}^*)[\vp] \,,\q \vp \in H^{-1/2}(\p D)\,,
\end{equation}
\mb{and the following basic result about the spectrum of the Neumann--Poincar\'{e} operator; see \cite{khavinson2007poincare} and \cite[Chapter 3]{folland1995introduction}.}

\begin{lemma} \label{lem:tbz2}
$\np_{\p D}^*$ is compact with $\sigma(\np_{\p D}^*) \subset[-\frac{1}{2},\frac{1}{2}]$ and  
\begin{align*}
    \dim \ker\big(\frac{1}{2}+\np_{\p D}^*\big) = m\,,\q \dim \ker\big(-\frac{1}{2}+\np_{\p D}^*\big) = n\,,
\end{align*}
\mb{where $m$ and $n + 1$ are the numbers of the connected components of $D$ and $D^c$, respectively.} Moreover, it holds that $\hbb$ is an invariant subspace of $\np_{\p D}^*$, and $\frac{1}{2}+\np_{\p D}^*$ is invertible on $\hbb$.  
\end{lemma} 

We denote by $\gamma_n \vp = \n \dd \vp$ the normal trace mapping on the space ${\bf H}(\ddiv, D)$. Then the restriction of $\gamma_n$ on $W$, denoted by $\w{\gamma}_n = \gamma_n|_W$, has the following properties.  

\begin{lemma} \label{lem:resfornormal}
$\w{\gamma}_n$ is an isomorphism from $W$ to $\hbb$ with $\wg [\phi]$ having the representation:
\begin{align} \label{auxeq_2}
\wg [\phi] = \na u \q \text{with}\q u = \ss_{\p D}(\frac{1}{2}+\np_{\p D}^*)^{-1}[\phi]\,,   
\end{align}
where $u$ is a solution to the interior Neumann problem:
\begin{equation} \label{eq:neumann}
    \Delta u = 0 \ \text{in}\ D\,,\q 
\frac{\p u}{\p \n} = \phi \ \text{on}\ \p D\,.
\end{equation}
\end{lemma}

In our subsequent analysis, we also need the spectral theory for the  leading--order operator $\tbb$ of $\TT_D^{\ww}$, which was given in \cite[Theorem 3.2]{costabel2012essential}; see also the earlier work \cite{friedman1984spectral} where the variational method was used. By 
Lemmas \ref{lem:tbz2}, \ref{lem:resfornormal} and \cite[Theorem 3.2]{costabel2012essential}, we arrive at the next result. 

\begin{lemma} \label{lem:tbz1} $ $ 
\begin{enumerate}
\item $\tbb: {\bf H}(\ddiv0, D) \to {\bf H}(\ddiv0, D)$ is self--adjoint with $W$ and $\hzz$ being its invariant subspaces. In particular, it holds that 
\begin{align} \label{auxeq_3}
    \tbb|_{\hzz} = 0\,, \q  \tbb|_W[\dd] = - \wg (\frac{1}{2} + \np_{\p D}^*) \w{\gamma}_n[\dd]\,.
\end{align} 
\item  $\sigma(\tbb|_W) \subset [-1,0)$ and $\sigma(\tbb) = \sigma(\tbb|_W) \cup\{0\}$. Hence, we have
\begin{equation} \label{eq:normtbwinver}
    \norm{ (\tbb|_W)^{-1}} =  \frac{1}{\min_{\lad \in \sigma(\tbb|_W)}|\lad|}\,.    
\end{equation}
Furthermore, $-\frac{1}{2}$ is 
the only possible accumulation point of $\sigma(\tbb)$, and there holds 
$$ \dim \ker(1+\tbb) = \dim \ker(1+\tbb|_W) = n\,, $$
where $n$ is the same as Lemma \ref{lem:tbz2}. 
\end{enumerate}
\end{lemma}

As a corollary, we can observe, from \eqref{auxeq_2} and \eqref{auxeq_3},  
\begin{equation}  \label{eq:inversetd}
    (\tbb|_W)^{-1}[\vp] = -\na \ss_{\p D}(\frac{1}{2} + \np_{\p D}^*)^{-2}[\n \dd \vp]\,,\q \vp \in W\,.
\end{equation}

We would like to point out that the connection between  $\tbb$ and $\np_{\p D}^*$ in \eqref{auxeq_3} 
was exploited in \cite{ammari2018shape} to develop
 a new approach to analyze the plasmonic resonances. From the physical perspective, $W$ corresponds to the electric components of the EM fields  (see the next paragraph), and hence it is the strong coupling between the metallic nanoparticles and the electric component of the incident wave that induces plasmonic resonances \cite{sarid2010modern,evlyukhin2012demonstration}. 
%
%

By contrast, for the dielectric nanoparticles with high refractive indices, we claim that the resonances are excited by the magnetic components of the incident EM fields (see, e.g., Corollary \ref{thm:mainmultimom_2}).
To give a glimpse of this essential difference between the plasmonic resonances and the dielectric resonances,  we recall from the Helmholtz decomposition (i.e., Lemma \ref{prop:helmddifree}) that for an open set $D$ with genus zero, $K_T(D)$ is a null space and then a divergence--free vector field $\psi \in {\bf H}(\ddiv 0, D)$ has the decomposition:
\begin{equation} \label{eq:helmgezero}
    \psi = \curl A + \na \phi\,, \q A \in \w{X}_N^0(D)\,, ~~\phi \in H^1(D)\,.
\end{equation}
One should note that these two fields $\curl A$ and $\na \phi$ present the forms of the magnetostatic and electrostatic fields, respectively, where $A$ is the so--called magnetic vector potential and $\phi$ is the electric scalar potential. Thus, the orthogonal projections associated with the decomposition \eqref{eq:helmddifree}:
\begin{equation*}  
    \pd: \hz \to  \hzz\,, \quad \pw : \hz \to W\,,
\end{equation*}
may be regarded as the projections from a divergence--free vector field to its magnetic and electric components, respectively. We shall see soon in Theorem \ref{thm:priest0} that it is exactly the space ${\bf H}_0(\ddiv 0, D)$ (the kernel of $\tbb$) that is responsible for the excitation of dielectric resonances and supports the (limiting) resonant modes.

For the sake of simplicity, in what follows, we use the shorthand notations $\TT^{{\rm s,t}}$ for $\P_{\rm s} \TT \P_{\rm t}$, ${\rm s,t = d, w}$, for a bounded linear operator $\TT$ on ${\bf H}(\ddiv0,D)$ (for instance, we write $\TT^{{\rm w,d}}$ for $\pw \TT \pd$). By definition and integration by parts, we have the useful relation for 
$\tb{n} (n \ge 2)$:
\begin{equation} \label{rela:simply}
     (\phi,\tb{n}[\vp])_{{\bf L}^2(D)} = (\phi, (\kb{n-2} + \na \ddiv \kb{n})[\vp])_{{\bf L}^2(D)} = (\phi,\np_{D,n-2}[\vp])_{{\bf L}^2(D)}\,,
\end{equation}
where $\phi \in {\bf L}^2(D)$, $\vp \in \hzz$, or $\vp \in {\bf L}^2(D)$, $\phi \in \hzz$. This relation enables us to write
\begin{equation} \label{rela:simply_2}
    \kb{n-2}^{{\rm s,t}} = \P_s \tb{n} \P_t \q \text{for}\q {\rm s,t = d, w} \q \text{except the case}\q {\rm s,t = w}\,.
\end{equation}
In particular, when $n = 3$, by noting that 
\begin{equation*}
    \kb{1}[\vp](x) = \frac{i}{4\pi}\int_D \vp(y)dy\,,
\end{equation*}
and making use of the Green's formula:
\begin{equation} \label{eq:greenfor_1}
    \int_D \vp(y) dy = - \int_D y \ddiv \vp(y) dy + \int_{\p D}\vp(y) \dd \n y d\sigma(y)\,,\q \vp \in {\bf H}(\ddiv,D)\,,
\end{equation}
\eqref{rela:simply} and \eqref{rela:simply_2} give us 
\begin{equation} \label{rela:tbb3}
    \kb{1}^{{\rm s,t}} = \P_s \tb{3} \P_t = 0 \q \text{for}\q {\rm s,t = d, w} \q \text{except}\q {\rm s,t = w}\,.
\end{equation}

Finally, we consider the compact self--adjoint 
operator $\kbdd = \pd \tb{2} \pd$.


\begin{lemma} \label{lem:specresotb2}
     $\kbdd: \hzz \to \hzz$ is a compact self--adjoint operator with the  eigen--decomposition: 
    \begin{equation*}
        \kbdd [E] = \sum_{n=0}^\infty \lad_n (E, E_n)_{{\bf L}^2(D)} E_n\,,\quad E \in \hzz\,,
    \end{equation*}
    where $\lad_0 \ge \lad_1 \ge ... \ge \lad_n \cdots > 0 $ are the eigenvalues of $\kbdd$, counted with their (finite) multiplicities,  with $0$ being the only accumulation point, and $\{E_n\}_{n\ge 0}$ consists of a complete orthonormal basis in $\hzz$. 
\end{lemma}
It is worth mentioning that in Lemma \ref{lem:specresotb2} above, we have concluded that  $0 \notin \sigma_p(\kbdd)$ and $\kbdd$ is positive semidefinite, and thus the sequence of positive eigenvalues $\{\lad_n\}$ exists and is infinite with $\lad_n \to 0$ as $n \to \infty$. These facts do not directly follow 
from the standard spectral theory. The first claim $0 \notin \sigma_p(\kbdd)$ will be justified later in the proof of Proposition \ref{prop:ledingqua}, while the second claim that $\kbdd$ is positive semidefinite can be checked as follows:
\begin{align*}
    (u, \np_D [u])_{{\bf L}^2(D)} & = \lim_{\eta \to 0^+} (u,\np_D^{i\eta}[u])_{{\bf L}^2(D)} = \lim_{\eta \to 0^+} (\F[u\chi_D],  \frac{1}{4\pi^2|\xi|^2+ \eta^2}\F[u\chi_D])_{{\bf L}^2(\R^3)} \ge 0 \,,\q u \in {\bf L}^2(D)\,.
\end{align*}
Here $\F$ denotes the $L^2$--Fourier transform \cite{folland1995introduction}. 

\subsection{A priori estimates} \label{subsec:priori}

\mb{This subsection aims for the a priori estimates for $\ww \in \Omega(1,\tau)$ with $\tau \in S_{\tau_0}$ and $|\ww| \ll 1$.}
For ease of exposition and analysis, we shall first write the operator--valued function $\A(\tau, \ww)$ in \eqref{def:adt} as an analytic family of operator matrices by the Helmholtz decomposition. 
We start with a concept of global equivalence from \cite{gohberg1990classes}, which  
generalizes the concept of similarity for linear operators and applies to nonlinear eigenproblems.

\begin{definition}\label{def:gloequi}
    Suppose that $V$ is an open set in $\C^d$, and $T(\lad):X_1 \to Y_1$ and $S(\lad):X_2 \to Y_2$ are  bounded linear operators acting on Banach spaces for each $\lad \in V$. The operator--valued functions $T(\dd)$ and $S(\dd)$ are called globally equivalent on $V$ if there exist  operator--valued analytic functions $P: V \to \L(X_1,X_2)$ and $Q: V \to \L(Y_2,Y_1)$, which are invertible for each $\lad \in V$, such that
    \begin{equation*}
        T(\lad) = P(\lad) S(\lad) Q(\lad), \quad \lad \in V\,.
    \end{equation*}
\end{definition}
We readily see from the above definition that if two analytic operator functions are globally equivalent on some open set $V$, then these two analytic functions must \mb{have the same} characteristic values on $V$. For our purpose, we introduce the product (Hilbert) space corresponding to the orthogonal direct sum \eqref{eq:helmddifree}:
\begin{equation*}
     \xx:= {\bf H}_0(\ddiv0,D) \t W\,, 
\end{equation*}
which is isomorphic with ${\bf H}(\ddiv0,D)$ via the mapping $\vp \to [\pd \vp, \pw \vp]$. Then it easily follows that $\A(\tau, \ww)$, $\tau \in \C \backslash \{0\}$, $\ww \in \C$, is globally equivalent to 
\begin{equation} \label{eq:equiA:matrix}
\mm \pd\A(\tau, \ww)\pd & \pd\A(\tau, \ww)\pw\\
\pw\A(\tau, \ww)\pd & \pw\A(\tau, \ww)\pw \nn.
\end{equation}
By multiplying the first row of \eqref{eq:equiA:matrix} by the factor $\tau$, we obtain the following system: 
\begin{equation} \label{eigen:target0}
  \AA(\tau, \ww) := \mm 1 - \tau \pd \TT_D^{\ww}\pd & - \tau \pd \TT_D^{\ww}\pw \\ - \pw \TT_D^{ \ww}\pd  & \frac{1}{\tau} - \pw \TT_D^{\ww}\pw \nn,
\end{equation}
which is also globally equivalent to $\A(\tau,\ww)$ on $(\C \backslash \{0\}) \t \C$.  

We remark that the choice of the form of $\AA(\tau,\ww)$ in \eqref{eigen:target0} is not necessary for the analysis in this subsection, but it is 
the most appropriate one for the arguments 
in Section \ref{sec:resoexit} below,
 compared to the form in \eqref{eq:equiA:matrix} and other possible choices, \mb{since the leading--order term of $\AA(\tau,\ww)$ is a triangular system (after a change of variable); see \eqref{eq:expa0}.}
 We specify it here to make our representation consistent. 

We now state and prove the main results of this subsection.

\begin{theorem} \label{thm:priest0}
\mb{There exist $0 < c_0 \ll 1$ and $\tau_0 \gg 1$ such that for any $\tau \in S_{\tau_0}$ and $\ww \in \Omega(1,\tau)$ with $|\ww| \le c_0$, the estimate holds}
\begin{equation}   \label{eq:prioriestfre1}
   |\frac{1}{\tau \ww^2} - \lad_i| \lesssim |\ww| + |\tau|^{-1}\,,
\end{equation}
\mb{where $\lad_i$ is given by}
\begin{align} \label{eq:eigladi}
   \mb{\lad_i = \arg\min\big\{\big|(\tau \ww^2)^{-1} - \lad_j\big|\,;\ \lad_j \in \sigma_p(\kbdd)\big\}.}
\end{align}
Moreover, the associated eigenfunctions $E^\ww$ of $\A(\tau,\ww)$ admit the estimate:
\begin{equation*} 
     \norm{\pw E^\ww}_{{\bf L}^2(D)} = \norm{E^\ww - \pd E^\ww}_{{\bf L}^2(D)} \lesssim |\tau|^{-1} + |\ww|^2\,. 
\end{equation*}

\end{theorem}

\begin{proof}
By the definition of $\AA$ in \eqref{eigen:target0}, we consider an eigenpair $(\ww, E^\ww) \in \C \times {\bf H}(\ddiv0, D)$ satisfying 
\begin{equation} \label{eigen:target}
    \AA(\tau, \ww) \mm \pd E^\ww \\ \pw E^\ww \nn = 0\q \text{with}\ \big\|E^\ww\big\|_{{\bf L}^2(D)} = 1\,.
\end{equation}
 We first show that $\pw E^\ww$ is a higher--order term in terms of $\tau^{-1}$. To do so, we apply 
the asymptotics of $\TT_D^{\ww}$ in Lemma \ref{lem:asytbw} to the 
second equation in \eqref{eigen:target} and 
see the following estimate:
\begin{align} \label{eq:prf1:1}
  \big\|\tau^{-1} \pw E^\ww -  \tbb[\pw E^\ww] - \ww^2 \pw \tb{2}[E^\ww]\big\|_{{\bf L}^2(D)} \lesssim |\ww|^3 \,,
\end{align}
for $|\ww| \le c_0$, where $c_0 > 0$ is a constant to be specified later. Here we have also used $\tbb|_{\hzz} = 0$. Since $\tbb$ is invertible on $W$ and $\big\|E^\ww\big\|_{{\bf L}^2(D)} = 1$, by \eqref{eq:normtbwinver} and \eqref{eq:prf1:1}, a simple use of triangle inequality
gives 
\begin{align} \label{eq:bounterm}
   \big\|\pw E^\ww\big\|_{{\bf L}^2(D)} &\lesssim \frac{1}{\min_{\lad \in \sigma(\tbb|_W)}|\lad|} \big(\big\|\tau^{-1} \pw E^\ww\big\|_{{\bf L}^2(D)} + \big\|\ww^2 \pw \tb{2}[E^\ww]\big\|_{{\bf L}^2(D)} + |\ww|^3 \big) \notag \\
    & \lesssim |\tau|^{-1} + |\ww|^2 \,.
\end{align}


We next derive the a priori estimate \eqref{eq:prioriestfre1} for the resonant frequency $\ww$. Considering the first equation in \eqref{eigen:target}, by the asymptotics of $\tbb^{\ww}$ and $\tbb|_{\hzz} = 0$, we can derive 
\begin{equation} \label{eq:prf1:2}
    \big\|\pd E^\ww - \tau \ww^2 \kbdd [E^\ww] - \tau \ww^2 \kbdw [E^\ww]\big\|_{{\bf L}^2(D)} \lesssim |\tau \ww^3|\,.
\end{equation}
Then, by the estimate \eqref{eq:bounterm} for $\pw E^\ww$, we readily see from \eqref{eq:prf1:2}: 
\begin{equation} \label{eq:projeqker}
    |\tau \ww^2| \big\|\frac{1}{\tau \ww^2 } \pd E^\ww - \kbdd [E^\ww]\big\|_{{\bf L}^2(D)} \lesssim  |\ww|^2 + |\tau \ww^3| \,.
\end{equation}
By the eigen--decomposition of $\kbdd$ in Lemma \ref{lem:specresotb2} and dividing \eqref{eq:projeqker} by $|\tau  \ww^2|$, we obtain
\begin{equation} \label{auxeq_4}
\big\|\sum_{j = 0}^\infty (\frac{1}{\tau  \ww^2} - \lad_j)(E^\ww,E_j)_{{\bf L}^2(D)} E_j \big\|_{{\bf L}^2(D)} \lesssim |\tau|^{-1} + |\ww|\,,
\end{equation}
which implies, by use of Parseval's identity, 
\begin{align} \label{eq:prf1:3} 
    \inf_{j}\big|\frac{1}{\tau  \ww^2} - \lad_j\big|^2 \sum_{j = 0}^\infty \big|(E^\ww,E_j)_{{\bf L}^2(D)}\big|^2 & \le   \sum_{j = 0}^\infty \Big|(\frac{1}{\tau  \ww^2} -\lad_j)(E^\ww,E_j)_{{\bf L}^2(D)}\Big|^2 \notag \\
    & \lesssim |\ww|^2 + |\tau|^{-2}\,.
\end{align}
Recalling \eqref{eq:bounterm} which yields 
\begin{equation*}
   \big\|\pd E^\ww\big\|^2_{{\bf L}^2(D)} = 1 - \big\|\pw E^\ww\big\|^2_{{\bf L}^2(D)} = \sum_{j = 0}^\infty \big|(E^\ww,E_j)_{{\bf L}^2(D)}\big|^2 \to 1 \q \text{as}\ \tau \to \infty\,, \ \ww \to 0\,,
\end{equation*}
by \eqref{eq:prf1:3},  we can choose large enough $\tau_0$ and small enough $c_0$ such that 
\begin{equation} \label{eq:prf1:4}
    \inf_{j}\big|\frac{1}{\tau \ww^2} - \lad_j\big|^2 \lesssim |\ww|^2 + |\tau|^{-2}\,,
\end{equation}
for $|\ww| \le c_0$ and $|\tau| \ge \tau_0$.
Since the sequence $\{\lad_j\}$ has only $0$ as its accumulation point, the infimum in \eqref{eq:prf1:4} is  attainable at some $\lad_i$, and then the desired estimate holds:
\begin{equation} \label{auxeq_prio}
   \big|\frac{1}{\tau \ww^2} - \lad_i\big|  =  \inf_{j}\big|\frac{1}{\tau \ww^2} - \lad_j\big|\lesssim |\ww| + |\tau|^{-1}\,.
\end{equation}
\end{proof}


We should note that there might be many other eigenvalues  of $\kbdd$ that also satisfy \eqref{eq:prioriestfre1}.
With a little more effort, we can have sharper estimates for subwavelength resonances and the associated eigenfunctions (see Corollary \ref{thm:priest1} below).
For this, we make some observations. Let $c_0$ and $\tau_0$ be the constants given by Theorem \ref{thm:priest0}. For a fixed $r_0 > 0$ and resonances $\ww \in \Omega(1,\tau)$ with $\tau \in S_{\tau_0}$ and $|\ww|\le c_0$, we define
\begin{equation} \label{def:sigmaspec}
    \Sigma := \big\{\lad_j \in \sigma_p(\kbdd)\,;\ \big|\frac{1}{\tau \ww^2}-\lad_j\big|\ge r_0\big\}\,,
\end{equation}
and the corresponding $L^2$--projection on ${\bf H}(\ddiv0,D)$:
\begin{equation}
    \P_{\Sigma} [\dd] := \sum_{\lad_j \in \Sigma} (\dd,E_j)_{{\bf L}^2(D)} E_j\,.
\end{equation}
By \eqref{auxeq_4} and the definition of $\Sigma$, we find
\begin{align} \label{eq:esteigenfun0} 
    \big\|\P_{\Sigma} E^\ww \big\|_{{\bf L}^2(D)} \lesssim  |\ww| + |\tau|^{-1}\,,
\end{align} 
which, along with \eqref{eq:bounterm}, implies that $E^\ww$ can approximated by $(1 - \pw - \P_{\Sigma})E^\ww$ with an error of order $|\ww| + |\tau|^{-1}$. Clearly, $1 - \pw - \P_{\Sigma}$ is nothing else than the spectral projection to the set:
\begin{align} \label{def:dominset}
\sigma_p(\kbdd) \backslash \Sigma = \big\{\lad_j \in \sigma_p(\kbdd)\,;\ \big|\frac{1}{\tau \ww^2} - \lad_j\big| < r_0\big\}\,.
\end{align}
This motivates the additional condition for $\lad_i$ in Corollary \ref{thm:priest1}, which in some sense controls the elements in $\sigma_p(\kbdd) \backslash \Sigma$. For our subsequent use, we introduce the eigen--projection $\P_{\lad}$ for the eigenvalue $\lad$ of $\kbdd$:
\begin{equation*}
    \P_\lad[\dd] = \sum_{\lad_j = \lad} (\dd,E_j)_{{\bf L}^2(D)}E_j\,.
\end{equation*}
\mbb{    
\begin{corollary}\label{thm:priest1}
Let $N$ be a positive integer. There exist $0 < c_0 \ll 1$ and $\tau_0 \gg 1$ such that for $\ww \in \Omega(1,\tau)$ with $\tau \in S_{\tau_0}$ and $|\ww| \le c_0$, 
 the following estimates hold \mb{when the index $i$ of $\lad_i$ defined in \eqref{eq:eigladi} satisfies $|i|\le N$}:
    \begin{equation} \label{eq:prioriestfre3}
        |1-\tau \ww^2 \lad_i| \lesssim |\tau|^{-1}\,,
    \end{equation}
and
\begin{equation} \label{eq:esteigenfun1}
     \norm{E^\ww - \mathbb{P}_{\lad_i}E^\ww}_{{\bf L}^2(D)} \lesssim |\tau|^{-1/2}\,.
\end{equation}
\end{corollary}

\begin{proof}
By assumptions, we have $c_0$ and $\tau_0$ such that the estimate  \eqref{eq:prioriestfre1} holds with the index of $\lad_i$ satisfying $|i|\le N$.
It follows that there exists $C_1 > 0$, depending on $N$ but independent of $\ww$ and $\tau$, such that 
\begin{equation} \label{eq:esttau}
 |\tau \ww^2|^{-1} \ge C_1\,.
\end{equation}
Indeed, by \eqref{eq:eigladi} and triangle inequality, we have 
\begin{align*}
\lad_i - \big|\frac{1}{\tau  \ww^2}\big| \le        \big|\frac{1}{\tau \ww^2} - \lad_i\big| \le  \lim_{j \to \infty} \big|\frac{1}{\tau \ww^2} - \lad_j\big|\,,
\end{align*}
and then \eqref{eq:esttau} follows, since $\lad_j \to 0$ as $j \to \infty$ 
and $\lad_i \ge \min \{\lad_j\,;\ j = 0,\ldots, N\} > 0$.





To obtain \eqref{eq:prioriestfre3}, we first prove \eqref{eq:esteigenfun1} and note 
the following estimate from \eqref{eq:prioriestfre1}: 
\begin{equation} \label{eq:prioriestfre2}
    |\frac{1}{\tau \ww^2} - \lad_i|\lesssim |\tau|^{-1/2}\,,
\end{equation}
by using $|\ww| \lesssim |\tau|^{-1/2}$ from \eqref{eq:esttau}. For the estimate \eqref{eq:esteigenfun1}, we define the constant $r_0$ in \eqref{def:sigmaspec} by 
\begin{equation} \label{eq:priest:1}
        r_0 := \frac{1}{2} \min\{|\lad_i - \lad_j|\,;\ i, j = 0,1,\cdots,N\,,\ \lad_i \neq \lad_j \}\,.
\end{equation}
Clearly, by \eqref{eq:prioriestfre2}, for $\tau_0$ large enough, the set $\sigma_p(\kbdd) \backslash \Sigma$ defined in \eqref{def:dominset} includes only one element $\lad_i$. Then, \eqref{eq:esteigenfun1}  
is a direct consequence of the arguments before Corollary \ref{thm:priest1}.

\if \commentflag = \ct
We let $\ww_i$ be $(\lad_i \tau \d^2 )^{-1/2}$ and define a complex--valued function $f(s)$ by
\begin{equation*}
    f(s) = \frac{1}{\tau \d^2 (\ww_i + s(\ww-\ww_i))^2}\,,\q s \in [0,1]\,.
\end{equation*}
Then we readily have $f(0) = \lad_i$ and the derivative of $f$: 
\begin{equation} \label{auxderiv}
    f'(s) = \frac{-2}{\tau \d^2 (\ww_i + s(\ww-\ww_i))^3} (\ww-\ww_i)\,.
\end{equation}
With these auxiliary functions, by \eqref{eq:prioriestfre1}, \eqref{eq:esttau} and a simple use of the mean--value theorem, we can derive
\begin{equation*}
   |f'(\xi)| = |f(1) - f(0)| =  |\frac{1}{\tau \d^2 \ww^2} - \lad_i| \lesssim |\d \ww|\,,
\end{equation*}
where $\xi \in (0,1)$. Then, by \eqref{auxderiv} and triangle inequality, it follows that 
\begin{align} \label{auxeq_11}
    |\ww-\ww_i| & = \frac{1}{2} |f'(\xi)| \dd |\tau \d^2 (\ww_i + \xi (\ww-\ww_i))^3| \notag \\
    & \lesssim |\d \ww| \dd |\tau \d^2| \dd (|\ww_i| + |\ww|)^3 \notag \\
     & \lesssim |\d \ww^2|\,,
 \end{align}
 which is our desired estimate \eqref{eq:prioriestfre2}.  In the last inequality of  \eqref{auxeq_11},  we have used the following estimates from \eqref{eq:esttau} and the definition of $\ww_i$:
 \begin{align*}
     |\tau \d^2| \dd (|\ww_i| + |\ww|)^2  \le C_3\,,
 \end{align*}
 and 
 \begin{align*}
    |\ww_i| + |\ww| \lesssim |\tau \d^2|^{-1/2} + |\ww| \lesssim |\ww|. 
 \end{align*}
 \fi
 

To complete the proof, we show \eqref{eq:prioriestfre3}. Since $\tau \ww^2 = O(1)$ holds by \eqref{eq:esttau}, it follows from \eqref{eq:asymeig} that 
    \begin{equation*}
        (1 - \tau \tbb - \tau \ww^2 \tb{2} - \tau \ww^3 \tb{3}) [E^\ww] = O(\ww^2)\,.
    \end{equation*}
    To get the equation for the resonant frequency $\ww$, we take the inner product of the above equation with $\mathbb{P}_{\lad_i}E^\ww$, and then obtain 
\begin{align} \label{eq:fstappfre0}
    (1-\tau \ww^2 \lad_i)(\mathbb{P}_{\lad_i}E^\ww, E^\ww)_{{\bf L}^2(D)} -\tau \ww ^3(\mathbb{P}_{\lad_i}E^\ww,\tb{3}[E^\ww])_{{\bf L}^2(D)} = O(\ww^2)\,,
\end{align} 
which further implies, by the formula \eqref{rela:tbb3},
\begin{equation} \label{eq:fstappfre1}
    (1-\tau \ww^2 \lad_i)(\mathbb{P}_{\lad_i}E^\ww, E^\ww)_{{\bf L}^2(D)} = O(\ww^2)\,.
\end{equation}
Note that \eqref{eq:esttau} and \eqref{eq:prioriestfre2} give $1 - \tau \ww^2\lad_i = O(\tau^{-1/2})$; and  \eqref{eq:esteigenfun1} allows us to write 
\begin{align} \label{auxeqqq}
    (\mathbb{P}_{\lad_i}E^\ww, E^\ww)_{{\bf L}^2(D)} = 1 + O(\tau^{-1/2}). 
\end{align}
Then, by \eqref{eq:esttau}, \eqref{eq:fstappfre1}, and \eqref{auxeqqq}, we have 
\begin{align*}
   |1-\tau \ww^2 \lad_i| \lesssim |\tau|^{-1}\,.
\end{align*}
The proof is now complete. 
\end{proof}}
The estimates \eqref{eq:prioriestfre2} and \eqref{eq:prioriestfre3} could be regarded as the 
quasi--static approximation and its first--order correction for the dielectric subwavelength resonances $\ww$ in terms of $\tau$, respectively. The above approach can also be easily generalized to derive the 
higher--order a priori estimates for the resonances $\ww$ and the eigenfunctions $E^\ww$.



\subsection{Existence and asymptotics of dielectric resonances} \label{sec:resoexit}

In the previous section, we have provided a priori estimates for the dielectric subwavelength resonances and the associated eigenfunctions. 
Nevertheless, it is still unknown whether the subwavelength resonances (i.e., the poles of $\A(\tau,\dd)^{-1}$ near the origin) exist or not for a given high contrast $\tau$. 
\mb{Besides, for a better understanding for the limiting behavior of subwavelength resonances when $\tau \to \infty$, it is necessary to investigate the asymptotic expansions of the resonances with respect to $\tau$.
These two issues will be the main focus of this subsection.} 



Recall that the operator--valued analytic function $\A(+\infty,\ww) = - \TT_D^{\ww}$ has a unique characteristic value $0$ but with an infinite--dimensional eigenspace. So the Gohberg--Sigal theory can not be directly applied to $\A(\tau,\ww)$
to guarantee the existence of resonances near the origin and further derive their asymptotic expansions. 
However, the estimate \eqref{eq:prioriestfre1} in Theorem \ref{thm:priest0} suggests that the resonance $\ww$ is near $\sqrt{\tau \lad_i}^{\sss -1}$ in some sense. Therefore, we are motivated to
introduce a new complex variable: 
\begin{align} \label{def:transtau}
    \h{\ww} = \sqrt{\tau} \ww\,,
\end{align}
and the corresponding analytic family of Fredholm operators: for $(\tau, \h{\ww}) \in S_{\tau_0} \t \C$, 
\begin{align*}
   \h{\A}(\tau, \h{\ww}) := \tau^{-1} - \TT_D^{\sqrt{\tau}^{-1}\h{\ww}}\, : \, {\bf H}(\ddiv0,D) \to {\bf H}(\ddiv0,D)\,.
\end{align*}
To analyze $\h{\A}(\tau,\h{\ww})$, as in \eqref{eigen:target0}, we introduce the operator matrix: 
\begin{equation} \label{def:aadw}
\h{\AA}(\tau,\h{\ww}) := \mm 1 - \tau \pd \TT_D^{\sqrt{\tau}^{-1} \h{\ww}}\pd & - \tau \pd \TT_D^{\sqrt{\tau}^{-1}\h{\ww}}\pw \\ - \pw \TT_D^{\sqrt{\tau}^{-1}\h{\ww}}\pd  & \frac{1}{\tau} - \pw \TT_D^{\sqrt{\tau}^{-1}\h{\ww}}\pw \nn \q \text{for} \ \tau \in S_{\tau_0},\  \h{\ww} \in \C\,,
\end{equation}
which is globally equivalent to $\h{\A}(\tau,\h{\ww})$.  It is easy to see that $\h{\A}(\tau,\dd)^{-1}$ and $\h{\AA}(\tau,\dd)^{-1}$ have the same poles that are also the zeros of 
the analytic function $\h{f}(\tau,\dd)$ which is defined by
\begin{align} \label{def:anatrans}
    \h{f}(\tau, \h{\ww}) := f(\tau, \sqrt{\tau}^{-1} \h{\ww}) \q \text{on}\ S_{\tau_0} \t \C\,,
\end{align}
where $f(\tau,\ww)$ is given in Proposition \ref{prop:gloana}. Once we understand the poles of $\h{\AA}(\tau,\dd)^{-1}$ well enough, we can readily go back to the original variable $\ww$ and obtain the desired result.


 
The expansion of $\h{\AA}(\tau,\h{\ww})$ with respect to $\tau$ follows from Lemma \ref{lem:asytbw}: 
\begin{align} \label{eq:asyeqahat}
    \h{\AA}(\tau,\h{\ww}) = \AA_0(\h{\ww}) + \sqrt{\tau}^{-1} \AA_1(\h{\ww}) + \tau^{-1} \AA_2(\h{\ww}) + O(\tau^{-3/2})\,,    
\end{align}
where the operators $\AA_0(\h{\ww})$ and $\AA_1(\h{\ww})$ are computed as follows, by using the relation \eqref{rela:tbb3},
\begin{equation} \label{eq:expa0}
\AA_0(\h{\ww}) = \mm 1 - \h{\ww}^2 \kbdd  & - \h{\ww}^2 \kbdw \\ 
0 &  - \tbb 
\nn \,, \q 
    \AA_1(\h{\ww}) = 0\,.
\end{equation}
Then, for the leading--order operator $\AA_0(\h{\ww})$, by the invertibility of $\tbb|_W$ and Lemma \ref{lem:specresotb2}, we obtain that $\hzz$ is an invariant subspace of $\AA_0(\h{\ww})$, and $\AA_0(\h{\ww})^{-1}$ is a meromorphic function on $\C$ with the poles given by 
\begin{align*}
\pm \h{\ww}_i := \pm (\lad_i)^{-1/2}\,, \q i = 0,1,2,\cdots\,,
\end{align*}
counted with their multiplicities. Without loss of generality, we proceed to consider the local behavior of $\AA_0(\h{\ww})^{-1}$ near the pole $\h{\ww}_0$. 
By a simple matrix computation, we derive
\begin{align}  \label{eq:invera00}
\AA_0(\h{\ww})^{-1}\mm f\\g \nn = \mm (1 - \h{\ww}^2 \kbdd )^{-1} [f - \h{\ww}^2 \kbdw \tbb^{-1} g] \\ - \tbb^{-1}g \nn ,
\end{align}
for $(f,g)\in \xx$ and $\h{\ww}$ in a small punctured neighborhood of $\h{\ww}_0$, which can be further factorized as:
\begin{align}  \label{eq:invera0}
   \AA_0(\h{\ww})^{-1}\mm f\\g \nn
     & = \mm (1 - \h{\ww}^2 \kbdd )^{-1}& 0 \\ 0 & I \nn \mm I & -  \h{\ww}^2 \kbdw \tbb^{-1} \\ 0 & - \tbb^{-1}\nn  \mm f \\ g \nn \notag\\
     & = \mm \frac{- \h{\ww}_0}{2(\h{\ww}- \h{\ww}_0)}\poo + \rr(\h{\ww}) & 0 \\ 0 & I \nn \mm I & - \h{\ww}^2 \kbdw \tbb^{-1}  \\ 0 & - \tbb^{-1}\nn  \mm f \\ g \nn\,,
\end{align}
where we have used the eigen--decomposition of $\kbdd$ in Lemma \ref{lem:specresotb2} and the identity:
\begin{equation*}
    \frac{1}{1 - \h{\ww}^2 \lad_0} = \frac{\h{\ww}_0^2}{\h{\ww}_0^2 - \h{\ww}^2} = \frac{-\h{\ww}_0}{2(\h{\ww} - \h{\ww}_0)} + \rr(\h{\ww})\,.
\end{equation*}
Here and throughout this work, $\rr(\dd)$ is a generic analytic function or operator--valued analytic function, whose exact definition may change from line to line. 



We summarize the above facts about $\AA_0(\h{\ww})$ in the following proposition.
\begin{proposition} \label{prop:polepena0}
   The operator $\AA_0(\h{\ww}):\xx \to \xx$ has the following inverse:
    \begin{align} \label{eq:fullinvera0}
        \AA_0(\h{\ww})^{-1} = \mm (1 - \h{\ww}^2 \kbdd )^{-1}& 0 \\ 0 & I \nn \mm I & -  \h{\ww}^2 \kbdw \tbb^{-1} \\ 0 & - \tbb^{-1}\nn\,,
    \end{align}
which is a meromorphic function on $\C$ with the simple poles $\pm \h{\ww}_i = \pm \sqrt{\lad_i}^{-1}$, $i = 0,1,2,\cdots$ satisfying 
$
    \pm \h{\ww}_i \to \pm \infty$ as $i \to \infty\,.
$
Moreover, the pole--pencil expansion of $\AA_0(\h{\ww})^{-1}$ near the pole $\h{\ww}_0$ is given by 
\begin{equation} \label{eq:polea0}
\AA_0(\h{\ww})^{-1}  = \frac{1}{\h{\ww} - \h{\ww}_0} \L_0  + \rr(\h{\ww})\,,
\end{equation}
for $\h{\ww}$ in a small punctured neighborhood $U$ of $\h{\ww}_0$, where the residue $\L_0$ at $\h{\ww}_0$ is given by 
\begin{equation} \label{def:lw}
    \L_0 =  -\frac{\h{\ww}_0}{2} \mm \P_{\lad_0} & - \h{\ww}_0^2\P_{\lad_0} \kbdw \tbb^{-1} \\ 0 & 0 \nn\,.
\end{equation}
\end{proposition}

Since $\pm \h{\ww}_i$ is a simple pole of $\AA_0(\h{\ww})^{-1}$, the geometric multiplicity of $\pm \h{\ww}_i$ (i.e., $\dim \ker (\AA_0(\pm \h{\ww}_i))$) is equal to its algebraic multiplicity denoted by  $M(\pm\h{\ww}_i,\AA_0(\dd))$ \cite[p.\,205]{gohberg1990classes}, which can  be computed as:
\begin{equation} \label{def:algmulti}
M(\pm\h{\ww}_i,\AA_0(\dd)) = \frac{1}{2\pi i} \tr \int_{\Gamma_0} \AA_0(\h{\ww})^{-1}\AA'_0(\h{\ww})d \h{w}\,,
\end{equation}
by the generalized argument principle \cite[p.\,206,\,Theorem 9.1]{gohberg1990classes}, where $\Gamma_0$ is a Cauchy contour enclosing only  $\pm \h{\ww}_i$ among all the poles of $\AA_0(\h{\ww})^{-1}$ and $\tr$ denotes the trace. We next use the generalized Rouch\'{e} theorem \cite[p.\,206,\,Theorem 9.2]{gohberg1990classes} to find the poles of $\h{\AA}(\tau,\dd)^{-1}$, for $\tau \in S_{\tau_0}$ with larger enough $\tau_0$, 
near the simple poles of $\AA_0(\h{\ww})^{-1}$. The result is given as follow; see \cite{gohberg1990classes,gohberg1971operator} for a review of 
Gohberg--Sigal theory and the related concepts used here.  



\begin{theorem}  \label{thm:existissue} 
Let $C_r > 0$ be a given constant with $C_r > |\h{\ww}_0|$ and such that $\AA_0(\h{\ww})$ is invertible on $\p B(0, C_r)$. We define the nonempty set: 
\begin{equation*}
   \{\pm \h{\ww}_i\,;\ i = 0,1,\cdots,n_0\} = \{\pm \h{\ww}_i\,;\ i = 0,1,\cdots\} \bigcap B(0,C_r) \subset \C\,.
\end{equation*}
For any small enough $\ep_0$ such that $B(\pm \h{\ww}_i, \ep_0) \subset B(0,C_r)$ for each $i \in \{0,1,\cdots,n_0\}$, and $\pm \h{\ww}_i$ is the only pole of $\AA_0(\h{\ww})^{-1}$ in $B(\pm \h{\ww}_i,\ep_0)$, there exists large enough $\tau_0$ depending on $\ep_0$ such that for $\tau \in S_{\tau_0}$, there are $m_i$ poles of $\h{\AA}(\tau,\dd)^{-1}$ in each $B(\pm \h{\ww}_i,\ep_0)$, counted with their algebraic multiplicities, where
$$
m_i := \dim \ker(\lad_i - \kbdd)\,.
$$
Moreover, these $\sum_{i = 0}^{n_0} 2 m_i$ poles give all the poles of $\h{\AA}(\tau,\dd)^{-1}$ in $B(0, C_r)$. 
\end{theorem}

\begin{proof}
By assumptions, for large enough $\tau_0$, there holds 
\begin{equation} \label{eq:auxeqq}
    \Big\|\big(\h{\AA}(\tau,\h{\ww}) - \AA_0(\h{\ww})\big)\AA_0(\h{\ww})^{-1}\Big\|< 1\,,\q \forall \h{\ww} \in \p B(\pm \h{\ww}_i, \ep_0)\,, \ \tau \in S_{\tau_0}\,.
\end{equation}
Then, the generalized Rouch\'{e} theorem gives
\begin{align*}
    M(\p B(\pm \h{\ww}_i, \ep_0), \h{\AA}(\tau, \dd))  = M(\pm \h{\ww}_i, \AA_0(\dd)) = m_i\,, \ i = 0,1,\cdots,n_0\,,
\end{align*}
where $M(\p B(\pm \h{\ww}_i,\ep_0),\h{\AA}(\tau,\dd))$ is given by 
\begin{equation*}
  M(\p B(\pm \h{\ww}_i,\ep_0),\h{\AA}(\tau,\dd)) = \frac{1}{2\pi i}\tr\int_{\p B(\pm \h{\ww}_i,\ep_0)} \h{\AA}(\tau,\h{\ww})^{-1} \frac{\p}{\p \ww}\h{\AA}(\tau,\h{\ww})  d\h{\ww}\,.
\end{equation*}
It follows that there exist $m_i$ poles of $\h{\AA}(\tau, \dd)^{-1}$ in $B(\pm \h{\ww}_i,\ep_0)$, $i = 0,1,\ldots,n_0$. Similarly, since $\AA_0(\h{\ww})$ is invertible on $\p B(0, C_r)$, if $\tau_0$ is large enough, 
\eqref{eq:auxeqq} holds for $\tau \in S_{\tau_0}$ and $\h{\ww} \in \p B(0, C_r)$, which, by the generalized Rouch\'{e} theorem, implies that these poles near $\pm \h{\ww}_i$ are all the poles of $\h{\AA}(\tau, \dd)^{-1}$ in $B(0,C_r)$. 
\end{proof}


\mbb{
Without loss of generality, we focus on the poles of $\h{\AA}(\tau,\dd)^{-1}$ near $\h{\ww}_0 = \sqrt{\lad_0}^{\sss -1}$, denoted by $\{\h{\ww}_{0,j}(\tau)\}_{j = 1}^{m_0}$. 
It is clear that $\sqrt{\tau}^{\sss -1}\h{\ww}_{0,j}$ are the poles of $\AA(\tau,\dd)^{-1}$ near the origin for large enough $\tau$, that is, $\sqrt{\tau}^{\sss -1}\h{\ww}_{0,j}$ are dielectric subwavelength resonances. This addresses the existence issue.

We re-arrange the order of $\h{\ww}_{0,j}$ such that $\{\h{\ww}_{0,j}(\tau)\}_{j = 1}^{m_\tau}$ gives the set of the distinct values of $\{\h{\ww}_{0,j}(\tau)\}_{j = 1}^{m_0}$, where $m_\tau \le m_0$
is a positive integer possibly
depending on $\tau$. 
We next consider the asymptotic behaviors of $\h{\ww}_{0,j}(\tau)$ with respect to $\tau$, and will see that when $\tau$ is large enough, $m_\tau$ is actually independent of $\tau$.

For our purpose, we follow the framework established in \cite{ammari2004splitting}. We define the functions, for $l \ge 1$, 
\begin{align} \label{def:powersum_1}
     \w{p}_l(\tau) &:= \frac{1}{2 \pi i} \tr \int_{\p B(\h{\ww}_0,\ep_0)} (\h{\ww} - \h{\ww}_0)^l \h{\AA}(\tau,\h{\ww})^{-1} \frac{\p}{\p \ww}\h{\AA}(\tau,\h{\ww})  d\h{\ww}\,.
\end{align}
By use of the general form of the argument principle \cite[Theorem 4.1]{gohberg1971operator}, we see that $\{\w{p}_l(\tau)\}_{l \ge 1}$ are nothing else than 
the power sum symmetric polynomials in variables  $x_j := \h{\ww}_{0,j}(\tau) - \h{\ww}_0$, $1 \le j \le m_0$: 
\begin{align} \label{def:powersum_2}
    \w{p}_l(\tau) = \sum_{j = 1}^{m_0} (\h{\ww}_{0,j}(\tau) - \h{\ww}_0)^l = \sum_{j = 1}^{m_0} x_j^l\,.
\end{align}
We introduce 
\begin{align} \label{def:powersum_3}
    z := \sqrt{\tau}^{-1}\,,\q p_l(z) := \w{p}_l(z^{-2})\q \text{for}\ \tau \in S_{\tau_0}\,.
\end{align}
In view of \eqref{eq:asyeqahat} and \eqref{def:powersum_1}, $p_l(z)$ has an analytic continuation into $B(0,\w{c}) = \{z \in \C\,; \ |z| \le \w{c}\}$ for small enough $\w{c}$. 
The key tool in the argument is Newton's identities, which relates the power sum symmetric polynomials and the elementary symmetric polynomials \cite{artin2011algebra}. 
We now state our result as follows. 

\begin{proposition} \label{prop:zeroasym}
Let the analytic functions $p_l(z)$, $l \ge 1$, be defined via \eqref{def:powersum_1} and \eqref{def:powersum_3} on $B(0,\w{c})$.  We recursively define $s_k(z)$, $0 \le k \le m_0$, by 
\begin{align} \label{def:recurela}
       k s_k + \sum_{i = 0}^{k-1} s_i p_{k - i} = 0 \q \text{with}\ s_0 = 1\,.
\end{align}
Then there hold
\begin{enumerate}
    \item the functions $s_i(z)$, $1 \le i \le m_0$, are analytic on $B(0,\w{c})$ and satisfy $s_i(0) = 0$;
    \item $\h{\ww}_{0,j}(\tau) - \h{\ww}_0$, $1 \le j \le m_0$, with $\tau = z^{-2}$, are all the zeros of the polynomial:
\begin{align} \label{def:poly} 
 \mathcal{Q}_z(\lad) := \sum_{i = 0}^{m_0} s_{i}(z) \lad^{m_0 - i} \,;
  \end{align}
 \item the analytic function $\h{f}$ in \eqref{def:anatrans} admits the factorization: $\h{f}(z^{-2},\h{\ww}) = \mathcal{Q}_z(\h{\ww} - \h{\ww}_0) g(z,\h{\ww})$ in a neighborhood of $(0,\h{\ww}) \in \C^2$, where $g$ is analytic in $z$ and $\h{\ww}$ and does not vanish everywhere. 
\end{enumerate}
\end{proposition}

\begin{proof}
It is easy to observe that if $p_i$, $1 \le i \le m_0$, are given, then one can readily use the recurrence relation \eqref{def:recurela} to find $s_i$, $1 \le i \le m_0$, by solving a triangular linear system,  which also implies that $s_i$ is the linear combination of $p_j$, $1 \le j \le i$, with rational coefficients. Moreover, since $p_l(z)$ is analytic in a small neighborhood of $z = 0$, so is $s_i(z)$. By \eqref{def:powersum_2} or \eqref{def:powersum_1},  $p_l(0)= 0$ for $l \ge 1$. Then \eqref{def:recurela} gives $s_i(0) = 0$ for $1 \le i \le m_0$. 

For the second statement, we define 
\begin{align*}
\w{\mathcal{Q}}_z(\lad) = \prod_{j = 1}^{m_0} (\lad - (\h{\ww}_{0,j}(z^{-2}) - \h{\ww}_0))\,,
\end{align*}
and it suffices to prove $\w{\mathcal{Q}}_z(\lad) = \mathcal{Q}_z(\lad)$. By elementary algebra, we have 
\begin{align*}
    \w{\mathcal{Q}}_z(\lad) =  \sum_{i = 0}^{m_0} (-1)^i \h{e}_i \lad^{m_0 - i} \,,
\end{align*}
where $\h{e}_i$ are elementary symmetric polynomials in variables  $x_j = \h{\ww}_{0,j}(z^{-2}) - \h{\ww}_0$. Then, noting the formula \eqref{def:powersum_2}, it follows that $p_l(z)$ and $\h{e}_i$ satisfy Newton’s identities \cite[p.\,505]{artin2011algebra}, which further yields $s_i(z) = (-1)^i \h{e}_i(z)$ by \eqref{def:recurela}.  

For the third statement, by definitions, for every $z \in B(0,\w{c})$, each zero of $\mathcal{Q}_z(\dd - \h{\ww}_0) 
$ is also a zero of $\h{f}(z^{-2},\dd)$. It follows that $g$ is well--defined pointwisely. To show that $g$ is analytic, it is sufficient to use Cauchy's integral formula:
\begin{align*}
    g(z,\h{\ww}) = \frac{1}{2 \pi i}\int_\Gamma \frac{\h{f}(z^{-2},\lad)}{\mathcal{Q}_z(\lad - \h{\ww}_0)} \frac{d \lad}{\lad - \h{\ww}}\,,
\end{align*}
and note that its right--handed side is analytic,
where $\Gamma$ is a suitably chosen Cauchy contour.
The proof is complete.  
\end{proof}

By Proposition \ref{prop:zeroasym} above, the splitting of $\h{\ww}_0$, when $\tau$ moves from $\infty$ to a finite value, is governed by the polynomial equation:  
\begin{align*}
    \mathcal{Q}_z(\lad) = \lad^{m_0} + s_1(z) \lad^{m_0 - 1} + \ldots + s_{m_0}(z) = 0\,, \q |z| \le \w{c}\,,\q s_i(0) = 0\,.
\end{align*}
Clearly, $\mathcal{Q}_z(\lad)$ is a Weierstrass polynomial (cf.\cite{suwa2007introduction}) and can be uniquely factorized as 
\begin{align} \label{eq:facqzlad}
     \mathcal{Q}_z(\lad) = \prod_{j = 1}^{l}  q_{z,j}^{k_j}(\lad) \,,
\end{align}
where $q_{z,j}$ are irreducible Weierstrass polynomials in $\lad$ of degree $d_j$ with $\sum_{j = 1}^l d_j k_j = m_0$, since the polynomial ring is a unique factorization domain. Then the equation $\mathcal{Q}_z(\lad) = 0$ is equivalent to the equations $q_{z,j}(\lad) = 0$, $1 \le j\le l $. By 
\cite[p.\,406,\,Theorem 3]{baumgartel1985analytic} or \cite[p.\,306,\,Theorem 4]{ahlfors1966complex},
the zeros of $q_{z,j}(\lad) = 0$ is given by a $d_j$--valued algebraic function $\{\lad_{j,i}(z)\}_{i = 1}^{\sss d_j}$ with the Puiseux expansion near $z = 0$. Moreover, since $q_{z,j}$ are mutually relatively prime, for small enough $z$ except $z = 0$, $q_{z,j}(\lad)$ have no common roots (cf.\cite[p.\,398]{baumgartel1985analytic}). Therefore, we arrive at the following result about the asymptotic behaviors of $\{\h{\ww}_{0,j}(\tau)\}_{j = 1}^{\sss m_0}$ for large enough $\tau$. 

\begin{theorem} \label{thm:splitting}
There exists large enough $\tau_0$ such that for $\tau \in S_{\tau_0}$, 
the $m_0$ poles $\{\h{\ww}_{0,j}(\tau)\}_{j = 1}^{m_0}$ of $ \h{\AA}(\tau,\dd)^{-1}$ near $\h{\ww}_0$ have $\w{m}_0$ different values, denoted by $\{\h{\ww}_{0,j}(\tau)\}_{j = 1}^{\w{m}_0}$,  with $\w{m}_0 \le m_0$ independent of $\tau$. Moreover, the set $\{\h{\ww}_{0,j}(\tau)\}_{j = 1}^{\w{m}_0}$ can be divided into $l$ subsets:
\begin{align*}
    \h{\ww}_{0,j}(\tau)\,,\q 1 + \sum_{s = 1}^{t - 1} d_s \le j \le \sum_{s = 1}^{t}  d_s\,, \q t = 1,\ldots, l\,, 
\end{align*}
with $d_s$ being positive integers satisfying $\sum_{s = 1}^l d_s = \w{m}_0$. Each subset forms a $d_t$--valued algebraic function with the Puiseux series expansion: for $\tau \in S_{\tau_0}$, 
\begin{align} \label{eq:exppui}
  \h{\ww}_{0,j}(\tau) - \h{\ww}_0 = \sum_{n = 2 d_t}^\infty c_{t,n} (\mu_t^i \tau^{- 1/(2 d_t)})^n\,,\q j = i + \sum_{s = 1}^{t - 1} d_s\,,  \q i = 1,2,\ldots, d_t\,,
\end{align}
where $\tau^{-1/(2d_t)}$ is a fixed root of $\tau^{-1/2}$ and $\mu_t$ is a fixed primitive $d_t$--th root of unity. 
\end{theorem}

\begin{remark}
To find the coefficients $c_{t,n}$ in 
\eqref{eq:exppui}, we need first to derive the Taylor expansion of the polynomial coefficients $s_i$ in \eqref{def:poly} with respect to $z$ (equivalently, $\tau$), which can be easily done by expanding $p_l$ and using the recurrence relation \eqref{def:recurela}. For a polynomial with degree $\le 4$, it is known that its zeros can be explicitly represented by the coefficients of the polynomial with only the algebraic and radical operations. Hence, if $m_0 \le 4$, we can use this approach to find the full expansions of $\h{\ww}_{0,j}(\tau) - \h{\ww}_0$ with respect to $\tau$. However, such calculations are not available for $m_0 > 5$ by Galois theory, so that in this case some numerical treatments are necessary; see \cite[Section 9]{ammari2004splitting} for a related discussion. 
\end{remark}

It is clear that the integers $d_s$ in Theorem \ref{thm:splitting} above
are given by the degrees of the prime factors $q_{z,j}$ of  $\mathcal{Q}_z$ in \eqref{eq:facqzlad}. 
To derive \eqref{eq:exppui}, we have also used the estimate \eqref{eq:prioriestfre3}. Indeed, it is easy to see that for $\{\h{\ww}_{0,j}(\tau)\}_{j = 1}^{m_0}$,
Corollary \ref{thm:priest1} holds with $N = 0$ and $\ww = \sqrt{\tau}^{\sss -1}\h{\ww}_{0,j}(\tau)$. Then it follows from \eqref{eq:prioriestfre3} that
\begin{equation} \label{eq:changepriori} 
    |\h{\ww}_{0,j}(\tau) - \frac{1}{\sqrt{\lad_0}}|\lesssim |\tau|^{-1}\,,
\end{equation}
which further implies that the coefficients of the expansion \eqref{eq:exppui} satisfy
\begin{equation*} \label{vanishcoeff}
    c_{t,n} = 0\,, \q 1 \le  t \le l  \,, \q    1 \le  n \le  2 d_t - 1\,.
\end{equation*}

We have now completed our understanding for the poles of $\h{\AA}(\tau,\dd)^{-1}$ for large enough $\tau$. Briefly and roughly speaking, 
these poles are given by the algebraic functions $\h{\ww}(\tau)$ near each point $\pm \h{\ww}_i = \pm \sqrt{\lad_i}^{\sss -1}$, $i = 0,1,2,\cdots$. 


Recalling Definition \ref{def:drs_1}, the assumption $\d = 1$ and the change of variable \eqref{def:transtau} we made above, by Theorem \ref{thm:existissue}, we can easily construct the dielectric subwavelength resonances by $\ww = \d^{\sss -1}\sqrt{\tau}^{\sss -1}\h{\ww}$, and obtain the following result. 

\begin{corollary}
Let $\tau_0$ and $\{\h{\ww}_{0,j}\}_{j = 1}^{m_0}$ be given by Theorem \ref{thm:splitting}. Then we have 
\begin{align} \label{eq:existence_resonance}
\ww_{0,j}(\d,\tau) := \d^{-1}\sqrt{\tau}^{-1}\h{\ww}_{0,j}(\tau) \in  \Omega(\d,\tau)\,, \q j = 1,\ldots, m_0\,,
\end{align}
for $\tau \in S_{\tau_0}$ and $0 < \d \ll 1$, where $ \Omega(\d,\tau)$ is defined in Definition \ref{def:drs_1}. 
\end{corollary}
To see \eqref{eq:existence_resonance}, we need to check: $\A_\d(\tau,\ww)$ is not invertible and $\ww \in S_\d$ (i.e., $|\d \ww| \ll 2 \pi$), both of which can be directly verified. We next give some remarks with physical significances. 

\begin{remark} \label{rem:resophy_2}
Similarly to Theorem \ref{thm:existissue}, as $\tau_0$ increases, one can show there are poles of $\h{\AA}(\tau,\dd)^{-1}$, $\tau \in S_{\tau_0}$, near more and more $\pm \h{\ww}_i$. All these poles can be scaled by the factor $\d^{\sss -1} \sqrt{\tau}^{\sss -1}$ and give the desired subwavelength resonances.
From the physical perspective, in this case, the ability of the nanoparticles to concentrate EM energy is enhanced \cite{evlyukhin2012demonstration}, since more resonances can be excited in a given incident frequency range (i.e., more resonances in a given neighborhood of the origin).



\end{remark}

\begin{remark} \label{rem:resophy_3}
We expect from the scaling factor $\d^{-1}\sqrt{\tau}^{-1}$ that with the increasing particle size $\d$, the resonant frequencies have the red--shift phenomenon, which has been reported in \cite{evlyukhin2012demonstration,kuznetsov2016optically}.
\end{remark}

\mb{In order for the resonances $\ww_{0,j}$ to be located in the visible light spectrum range (i.e., $\ww \sim 1$;
see Appendix \ref{app:S_1}) and to be observed in the experimental environment, we need the assumption:
\begin{align} \label{assp:taudelta}
    \tau \in \R \,,\q \tau \sim \d^{-2}\,.
\end{align}
The assumption $\tau \in \R$ is quite natural (see remarks after \eqref{asp:highcont}), while the other condition 
$\tau \sim \d^{-2}$, enforced by the condition $\ww = \d^{\sss -1}\sqrt{\tau}^{\sss -1} \h{\ww}\sim 1 $, means that the wavelength inside the nanoparticle is comparable with its characteristic size. In this configuration, 
the dielectric resonances were indeed experimentally observed in the visible light region \cite{kuznetsov2016optically}.}

In the remaining part of this work,
we mainly adopt and restrict ourselves to the following explicit form of the assumption \eqref{assp:taudelta}:
\begin{align} \label{asp:tauexpd}
    \tau = \d^{-2} \,, \q  0 < \d \ll 1\,,
\end{align}
for ease and clarity of exposition, although many results holds under the general one $\tau(\d) \sim \d^{-2}$ and even beyond this regime. We end this section by applying the results obtained above to the case \eqref{asp:tauexpd}. We first give a well known and important lemma \cite{colton2012inverse,dyatlov2019mathematical}. We reproduce its proof here for completeness.} 

\begin{lemma} \label{lem:location}
For the operators $\A_\d(\tau,\ww)$ in \eqref{redef:Ad} with $\tau > 0$ and $\d > 0$, the poles of $\A_\d(\tau,\dd)^{-1}$
are located in the lower half--plane $\{\ww \in \C\,;\  \im \ww < 0\}$, and are symmetric with respect to the imaginary axis.
\end{lemma}

\begin{proof}
The symmetry of the poles follows from the relation:
\begin{equation*}
    \overline{\A_\d(\tau,\ww)[E]} = (\tau^{-1} -  \tbb^{-\d \bar{\ww}}) \left[\bar{E}\right] = \A_\d(\tau,-\bar{\ww})\left[\bar{E}\right],
\end{equation*}
from which we can conclude if $\ww$ is a characteristic value, so is $-\bar{\ww}$. As a consequence of Rellich's lemma \cite[Theorem 6.10]{colton2012inverse}, we have 
that there are no real characteristic values of $\A_\d(\tau,\ww)$ for $\tau > 0$ (see \cite[Theorem 2.1]{costabel2011kleinman}). 
To finish the proof of the lemma, we show that if $(\ww,E^\ww) \in \C \times {\bf H}(\ddiv0,D)$ satisfies $\A_\d(\tau,\ww)[E^\ww] = 0$ with $\im \ww > 0$, then the associated eigenfunction $E^\ww$ must be zero. We let $u = \tbb^{\d\ww}[E] \in {\bf H}_{\rm loc}(\ccurl,\R^3)$, 
which exponentially decays as $|x| \to \infty$, due to $\im \ww > 0$. Then, integrations by parts outside and inside the domain lead to
\begin{equation} \label{eq:auxitp_1}
    \int_{\R^3\backslash \bar{D}} |\curl u| - (\d\ww)^2 |u|^2 dx = \int_{\p D}\n \t \curl u \dd \bar{u} d\sigma\,,
\end{equation}
and 
\begin{equation} \label{eq:auxitp_2}
     \int_D   |\curl u|^2 - (\d\ww)^2 (1 + \tau) |u|^2 dx  =  -\int_{\p D} \n \t \curl u \dd \bar{u} d \sigma\,,
\end{equation}
respectively. We hence have, by taking the imaginary parts of both sides of \eqref{eq:auxitp_1} and \eqref{eq:auxitp_2}, 
\begin{equation*}
    \im \int_{\p D} \n \t \curl u \dd \bar{u} d \sigma = - \im (\d\ww)^2 \int_{\R^3 \backslash \bar{D}} |u|^2 dx = \im (\d\ww)^2 (1 + \tau)\int_D |u|^2 dx  =0 \,.
\end{equation*}
If $\ww \neq i$ ($\im \ww^2 \neq 0$), the above formula directly yields $u = 0$. If $\ww = i$, then we observe from \eqref{eq:auxitp_1} and \eqref{eq:auxitp_2} that 
\begin{equation*}
    \int_{\R^3\backslash \bar{D}} |\curl u| + \d^2 |u|^2 dx =  \int_D   |\curl u|^2 + \d^2 (1 + \tau) |u|^2 dx = 0\,,
\end{equation*}
which also implies $ u = 0$. The proof is complete.
\end{proof}
\mbb{
We note that under the relation \eqref{asp:tauexpd}, \eqref{eq:existence_resonance} gives
\begin{align*}
\ww_{0,j}(\d) := \h{\ww}_{0,j}(\d^{-2}) \in \Omega(\d,\tau)\,, \q j = 1,\ldots, m_0\,,
\end{align*}
and then Theorem \ref{thm:splitting} readily applies to $\ww_{0,j}(\d)$. 
We summarize the properties of $\ww_{0,j}(\d)$ in the following theorem. 

}

\begin{theorem} \label{thm:exitphydsr}
\mbb{Suppose there holds $\tau = \d^{-2}$. Then for small enough $\d$, the meromorphic function $\A_\d(\d^{-2},\dd)^{-1}$ has $m_0$ poles $\{\ww_{0,j}(\d)\}_{j = 1}^{m_0}$ in the fourth quadrant $\{\ww\in\C\,;\ \re \ww > 0\,,\ \im \ww < 0\}$ near $\sqrt{\lad_0}^{\sss -1}$, which are the dielectric subwavelength resonances in the sense of Definition \ref{def:drs_1}. Moreover, 
$\ww_{0,j}(\d) =  \h{\ww}_{0,j}(\d^{-2})$ satisfies Theorem \ref{thm:splitting}. 
    That is, $\{\ww_{0,j}(\d)\}_{j = 1}^{m_0}$ has $\w{m}_0$ distinct values given by one or several algebraic functions with possible algebraic singularities at $\d = 0$. In particular, $\ww_{0,j}(\d)$ is differentiable at $\d = 0$ with the asymptotic estimate:}
    \begin{align} \label{aim1}
    \mbb{|\ww_{0,j}(\d) - \frac{1}{\sqrt{\lad_0}}|\lesssim \d^2\,.}
\end{align}
\end{theorem}

\section{Multipole radiation and scattering enhancement} \label{sec:multirad}

In this section we consider the behavior of the scattering amplitude of the scattered wave $E^s$ in \eqref{eq:model2}. It is known that for any EM radiating field $E$, its scattering amplitude $E_\infty$, as an analytic tangential vector field on $\S$, can be defined by the following asymptotic expansion:
\begin{equation} \label{addeqe}
    E = \frac{e^{i\ww|x|}}{|x|} E_\infty(\ww,\h{x}) +  O\Big(\frac{1}{|x|^2}\Big)\q \text{as}\ |x| \to \infty\,.
\end{equation}
For our problem \eqref{eq:model2}, recalling the Lippmann--Schwinger representation \eqref{eq:liprep_2}, and exploiting the far--field expansion of the EM Green's tensor \cite{colton2012inverse}:
\begin{equation*} \label{eq:asygreenfar}
    (\I + \frac{1}{\ww^2}\na_x \ddiv_x)g(x-y,\ww) = \frac{e^{i\ww|x|}}{4 \pi |x|}e^{-i\ww \h{x}\dd y}\left(\I - \h{x}\otimes \h{x}\right) + O\Big(\frac{1}{|x|^2}\Big)\q \text{as}\ |x|\to \infty\,,
\end{equation*}
we can easily derive the following integral representation for the scattering amplitude $E_\infty^s$ associated with $E^s$:
\begin{equation} \label{eq:intrepes}
    E^s_\infty(\ww,\h{x}) = \frac{\tau \ww^2}{4 \pi}(\I - \h{x}\otimes \h{x}) \int_{D_\d} e^{-i\ww \h{x}\dd y} E(y)dy\,,
\end{equation}
where the matrix $\I-\h{x}\otimes \h{x}$ is the projection on the tangent plane of the unit sphere $\S$ at $\h{x}$.


\mb{The main aim of this section is to derive the small volume expansion (approximation) of the scattering amplitude $E_\infty^s(\ww,\h{x})$ under the ${\bf L}^\infty_{\rm T}(\S)$-norm with respect to $\d$, which holds uniformly for
incident frequencies $\ww$ near resonances, and then to analyze the enhancement of the scattered field. 
We will also compute the corresponding extinction and scattering cross sections and analyze their blow--up rates. These results shall mathematically justify the strong magnetic responses of nanoparticles with high refractive indices.  Without loss of generality, we introduce the following assumptions for simplicity and clarity:}
\mbb{
\begin{enumerate}[label=\textbf{C.\arabic*},ref=C.\arabic*]
    \item Let the high contrast condition \eqref{asp:tauexpd}: $\tau = \d^{-2}$, $0 < \d \ll 1$, hold. 
    \label{1}
    \item Let $U$ be the punctured neighborhood of $\ww_0 := \sqrt{\lad_0}^{-1}$ such that \eqref{eq:polea0} holds.  \label{2}
    \item Let the reference set $D$ have genus zero and $m_0 = \dim \ker (\lad_0 - \kbdd) = 1$ hold with the corresponding normalized eigenfunction denoted by $E_0$. \label{3}
\end{enumerate}
We assume \eqref{1}, \eqref{2} and \eqref{3} hold in this whole section, but in order to avoid any confusion we shall explicitly point out which assumptions are used for each result and definition. 
In view of \eqref{1}, by abuse of notation, we define 
\begin{equation*}
    \AA(\d, \ww) := \h{\AA}(\d^{-2},\ww)\,, \q 0 < \d \ll 1\,,\ \ww \in \C\,.
\end{equation*}
where $\h{\AA}(\dd,\dd)$ is the analytic family of operators given in \eqref{def:aadw}. Then, \eqref{eq:asyeqahat} yields the expansion of $\AA(\d,\ww)$: 
\begin{align} \label{eq:asyeqahat22} 
    \AA(\d,\ww) = \AA_0(\ww) + \d^2 \AA_2(\ww) + O(\d^3) \q \text{as}\ \d \to 0\,.
\end{align}


As we have emphasized earlier (before Lemma \ref{lem:location}), although $\tau = \d^{-2}$ is assumed here, the general assumption $\tau \sim \d^{-2}$ in \eqref{assp:taudelta} is enough for most results and formulas below. In view of this, we keep the notation $\tau$ as much as possible in the following exposition.}




\subsection{Multipole expansion for the scattering amplitude} 

To find a uniformly valid small volume approximation of $E_\infty^s$, we shall first develop a new framework for the full Cartesian multipole expansions of the scattering amplitudes.
We start with the following formula from \eqref{eq:intrepes}, by the Taylor expansion of $e^{-i\ww\h{x}\dd y}$ and the change of variable $y = \d \w{y}$,
\begin{equation} \label{eq:cartesp}
    E^s_\infty(\ww,\h{x}) =  \frac{\tau \ww^2}{4 \pi}(\I - \h{x}\otimes \h{x}) \int_D \sum_{n = 0}^\infty \d^{n+3} \frac{(-i \ww \h{x}\dd \w{y})^n}{n!} \w{E}(\w{y})d\w{y}\q \text{with}\ \w{E}(\w{y}) := E(\d\w{y})\,,
\end{equation}
which is used in physics \cite{jackson1999classical,evlyukhin2016optical} as a starting point for deriving the classical multipole moment expansion; see \eqref{eq:multiradiation2} below. However, as pointed out in \cite{jackson1999classical}, to arrive at \eqref{eq:multiradiation2} from \eqref{eq:cartesp}, 
one needs different calculation techniques to disentangle the electric and magnetic multipole fields from each term in the expansion \eqref{eq:cartesp}, and the details and technicalities involved become 
increasingly prohibitive when the order $n$ increases.
Instead, we take a different approach for the multipole moment expansion, with the help of the Helmholtz decomposition \eqref{eq:helmddifree} and Lemma \ref{lem:forappall} below, the advantage of which is that 
it can handle all the orders in a unified manner.

The following lemma is a consequence of integration by parts.

\begin{lemma} \label{lem:forappall}
For $\vp \in {\bf H}_0(\ccurl, D)$, we have
\begin{align} \label{eq:maginbp}
     \int_D (\h{x}\dd y)^l  \curl \vp (y) dy = \int_D  l(\h{x}\dd y)^{l-1}\vp(y)  dy \t \h{x}\,, \q l \ge 1\,.
\end{align}
\end{lemma}
For the open set $D$ with genus zero, we have written the following decomposition for $\psi \in {\bf H}(\ddiv0,D)$ in \eqref{eq:helmgezero}:
\begin{equation} \label{eq:decomstand}
\psi = \pd \psi + \pw \psi \q \text{with}\q  \pd \psi = \curl \vp\q \text{and}\q \pw \psi = \na p\,,
\end{equation}  
where $\vp$ is uniquely determined in $\w{X}_N^0(D)$ and $p \in H^1(D)$ is unique up to a constant. In view of the above decomposition \eqref{eq:decomstand} and the physical explanations provided after \eqref{eq:helmgezero}, it is natural for us to separate the physically different electric and magnetic radiations by writing 
\begin{equation} \label{eq:sepemet}
    \int_D (\h{x}
    \dd \w{y})^l \w{E}(\w{y})d\w{y} = \int_D (\h{x}
    \dd \w{y})^l \left(\pd\w{E}(\w{y}) + \pw \w{E}(\w{y})\right)d\w{y}\,, 
\end{equation} 
for each term in the sum \eqref{eq:cartesp}, and introduce the following new class of electric and magnetic multipole moments based on Lemma \ref{lem:forappall} and \eqref{eq:decomstand}.

\begin{definition} \label{def:EMmoment}
\mb{Under the condition \eqref{3},} for  a divergence--free vector field $\psi \in {\bf H}(\ddiv0,D)$ with the Helmholtz decomposition \eqref{eq:decomstand}, we define 
the magnetic $l-$moment, as an $l-$tensor on the 
 $l-$ary Cartesian power of $\R^3$: $\Pi^l \R^3$, by 
    \begin{equation*}
\mq^\psi_l: =  
    \int_D l \vp(y) \otimes \left(\otimes^{l-1}y\right) dy \,, \q  \ l = 1,2,\cdots \q \text{and}\q  \mq^\psi_l = 0\,,\q l = 0\,,
\end{equation*}
and the electric $l-$moment, as an $(l+1)-$tensor on $\Pi^{l+1} \R^3$, by
\begin{equation*}
   \eq^\psi_l := \int_{D} \na p(y) \otimes\left( \otimes^l y \right)d\sigma(y)\,,\q l = 0, 1,\cdots\,.
\end{equation*}
\end{definition}

We point out that $\mq^\psi_l$ is well--defined for $\vp \in \w{X}_N^0(D)$, since so is \eqref{eq:maginbp}.

\begin{remark} \label{rem:cntnewandclas}
 Definition \ref{def:EMmoment} may provide a new interpretation, as well as a new derivation, for those EM moments occurring in the physics literature. For instance,
 in classical electrodynamics, given a current distribution ${\bf J}$, the corresponding electric dipole moment $\ed$ and magnetic dipole moment $\md$ are given by 
\begin{equation*}
    \ed = \frac{1}{i\ww}\int x \ddiv {\bf J}(x) dx \q \text{and}\q  \md = \frac{1}{2}\int x \t {\bf J}(x)dx\,,
\end{equation*}
respectively; see \cite[Chapter 9]{jackson1999classical}. If the distribution ${\bf J}$ admits the decomposition ${\bf J} = \curl A + \na p$, a formal integration by parts gives
\begin{equation*}
    \ed = \frac{1}{i\ww}\int x \ddiv \na p = - \frac{1}{i\ww}\int \na p\,, \q \md = \frac{1}{2}\int x \t \left(\curl A\right) = \int A.
\end{equation*}
Thus these physically defined dipole moments are (at least formally) equal to our newly introduced  EM moments $\eq^{\psi}_0$ and $\mq^{\psi}_1$,  up to constant factors. A similar formal verification can be conducted for the higher--order moments.
\end{remark}

We remark that an $l-$tensor is a multilinear function on the product vector space. By convention, we identify the dual space of $\R^3$ with itself. Therefore, the linear mappings 
$\mq^\psi_l(\dd,\Pi^{l-1}\h{x})$ and $\eq^\psi_l(\dd,\Pi^{l}\h{x})$ on $\R^3$ 
can be regarded as vectors. 
With the help of these notions, there holds, by \eqref{eq:sepemet} and Definition \ref{def:EMmoment}, 
\begin{equation} \label{eq:cassepem}
\int_D (\h{x}
    \dd \w{y})^l \w{E}(\w{y})d\w{y} = \mq_l^{\w{E}}(\dd,\Pi^{l-1} \h{x})\t \h{x} +  \eq_l^{\w{E}}(\dd ,\Pi^l \h{x}) \,,
\end{equation}
which, combined with the expansion \eqref{eq:cartesp}, implies the following proposition.

\begin{proposition} \label{lem:samfuexp}
\mb{Under the condition \eqref{3},} the scattering amplitude $E_\infty^s(\ww,\h{x})$ corresponding to the problem \eqref{eq:model2} has the following multipole moment expansion:
\begin{equation} \label{eq:samfuexp}
    E^s_\infty(\ww,\h{x}) =  \frac{\tau \ww^2 \d^3}{4 \pi}(\I - \h{x}\otimes \h{x})\sum_{l = 0}^\infty  \frac{(-i\d\ww)^l}{l!} \left(\mq_l^{\w{E}}(\dd,\Pi^{l-1} \h{x})\t \h{x} +  \eq_l^{\w{E}}(\dd ,\Pi^l \h{x}) \right) \,,
\end{equation}
where the field $\w{E} \in {\bf H}(\ddiv0, D)$ is the solution to the equation: 
\begin{equation} \label{eq:targetsolve} 
    (1 - \tau \TT_D^{\d \ww})[\w{E}] = \w{E}^i \q \text{with}\q   \w{E}^i(\w{x}) = E^i(\d \w{x})\,.
\end{equation}
\end{proposition}

Since we are working in the quasi--static regime (i.e., $\d$ is small enough), we may expect that the first several terms in \eqref{eq:cartesp} (or \eqref{eq:samfuexp}) are sufficient to characterize the behavior of $E^s_\infty$. Clearly, to approximate $E^s_\infty$ up to an error of a certain order, we need to further derive the asymptotic expansion of the solution $\w{E}$ to \eqref{eq:targetsolve} in terms of $\d$. Consider the system below that is equivalent to \eqref{eq:targetsolve}:
\begin{equation} \label{eq:targetsys}
    \AA(\d,\ww) \mm \pd \w{E}\\  \pw \w{E} \nn 
    = \mm \pd \w{E}^i \\ \tau^{-1} \pw \w{E}^i \nn.
\end{equation}
\mb{We have the following lemma regarding the inverse $ \AA(\d,\ww)^{-1}$, by a Neumann series argument.} 
\mbb{\begin{lemma} \label{lem:estreol}
Under the conditions \eqref{1} and \eqref{2},
      there exists constant $C_A$, depending on the neighborhood $U$, such that for $\ww \in U$, $0 \le \d \le \sqrt{C_A |\ww - \ww_0|}$, it holds that   
  \begin{align}  \label{est:neum_1}
    \norm{\AA(\d,\ww)^{-1}}_{\L(\xx,\xx)} \lesssim \frac{1}{|\ww - \ww_0|} \,,
\end{align}
\begin{align} \label{est:neum_2}
    \Big\| \Big(\sum_{n = k}^\infty (-1)^n \big(\AA_0(\ww)^{-1} (\AA(\d,\ww) - \AA_0(\ww))\big)^n \AA_0(\ww)^{-1} \mm f \\ g \nn \Big)_1 \Big\|_{{\bf L}^2(D)}  \lesssim \frac{ (2 C_A)^{-k} \d^{2k}}{|\ww - \ww_0|^{k+1}}(\norm{f}_{{\bf L}^2(D)} + \norm{g}_{{\bf L}^2(D)})\,, 
\end{align}
and  
\begin{align} \label{est:neum_3}
    \Big\| \Big(\sum_{n = k}^\infty (-1)^n \big(\AA_0(\ww)^{-1} (\AA(\d,\ww) - \AA_0(\ww))\big)^n \AA_0(\ww)^{-1} \mm f \\ g \nn \Big)_2 \Big\|_{{\bf L}^2(D)} \lesssim \frac{(2 C_A)^{-k} \d^{2k}}{|\ww - \ww_0|^{k}}(\norm{f}_{{\bf L}^2(D)} + \norm{g}_{{\bf L}^2(D)})\,.
\end{align}
\end{lemma}

\begin{proof}
We write $\AA(\d,\ww) = \AA_0(\ww) + \mathbb{B}(\d,\ww)$ with $\mathbb{B}(\d,\ww) =  \AA(\d,\ww) - \AA_0(\ww)$. By \eqref{eq:polea0} and \eqref{eq:asyeqahat22}, it is easy to see that for $\ww$ in the punctured neighborhood $U$ of $\ww_0$ and $0 \le \d \ll 1$, there holds 
\begin{align} \label{auxeqq_neu_3}
    \norm{\AA_0(\ww)^{-1}\mathbb{B}(\d,\ww)}_{\L(\xx,\xx)} \le \frac{(2 C_A)^{-1} \d^2}{|\ww - \ww_0|}\,,
\end{align}
with the constant $C_A$ independent of $\ww$ and $\d$. Then, let $\d_0 > 0$ be defined by $(2C_A)^{-1} \d_0^2/|\ww - \ww_0| = 1/2$. We have  
for $\ww \in U$ and $0 \le  \d \le \d_0$, $\AA(\d,\ww)^{-1}$ admits a Neumann series expansion:
\begin{align*}
\AA(\d,\ww)^{-1} = \sum_{n = 0}^\infty  (-1)^n \big(\AA_0(\ww)^{-1}\mathbb{B}(\d,\ww) \big)^n \AA_0(\ww)^{-1}\,,
\end{align*}
which readily implies, by \eqref{eq:polea0}, 
\begin{align*}
    \norm{\AA(\d,\ww)^{-1}}_{\L(\xx,\xx)} \lesssim \norm{\AA_0(\ww)^{-1}}_{\L(\xx,\xx)} \lesssim \frac{1}{|\ww - \ww_0|} \,.   
\end{align*}
For the estimates \eqref{est:neum_2} and  \eqref{est:neum_3}, 
we recall
\eqref{eq:invera00} and obtain, by a simple estimate, 
\begin{align} \label{auxeqq_neu_1}
   \resizebox{.90\hsize}{!}{$
    \Big\| (\AA_0(\ww)^{-1} \mm f \\ g \nn)_1 \Big\|_{{\bf L}^2(D)} \lesssim \frac{1}{|\ww - \ww_0|} (\norm{f}_{{\bf L}^2(D)} + \norm{g}_{{\bf L}^2(D)}) \,,\q \Big\| (\AA_0(\ww)^{-1} \mm f \\ g \nn)_2 \Big\|_{{\bf L}^2(D)} \lesssim  \norm{g}_{{\bf L}^2(D)}\,. $}
\end{align}
Then, by \eqref{auxeqq_neu_3}, \eqref{auxeqq_neu_1} and noting
\begin{equation*}
    \big(\AA_0(\ww)^{-1}\mathbb{B}(\d,\ww)\big)^n \AA_0(\ww)^{-1} =    \AA_0(\ww)^{-1} \Big(\mathbb{B}(\d,\ww)\big(\AA_0(\ww)^{-1}\mathbb{B}(\d,\ww)\big)^{n - 1} \AA_0(\ww)^{-1} \Big)\,,
\end{equation*}
we have
\begin{align*}
  &\Big\|  \big(\big(\AA_0(\ww)^{-1}\mathbb{B}(\d,\ww)\big)^n \AA_0(\ww)^{-1} \mm f \\ g \nn \big)_1 \Big\|_{{\bf L}^2(D)} \lesssim \frac{ (2 C_A)^{-n} \d^{2n}}{|\ww - \ww_0|^{n+1}}(\norm{f}_{{\bf L}^2(D)} + \norm{g}_{{\bf L}^2(D)})\,,
\end{align*}
and 
\begin{align*}
  &\Big\|  \big(\big(\AA_0(\ww)^{-1}\mathbb{B}(\d,\ww)\big)^n \AA_0(\ww)^{-1} \mm f \\ g \nn \big)_2 \Big\|_{{\bf L}^2(D)} \lesssim \frac{(2 C_A)^{-n} \d^{2n}}{|\ww - \ww_0|^{n}}(\norm{f}_{{\bf L}^2(D)} + \norm{g}_{{\bf L}^2(D)})\,,
\end{align*}
where the generic constants are independent of $n$, $\ww$ and $\d$. The proof is complete by taking the
sum of the above two estimates from $n = k$ to infinity respectively. 
\end{proof}}


By Parseval's identity and the Green's formula \eqref{eq:greenfor_1}, we observe $(\w{E}^i(0),E_n\big)_{{\bf L}^2(D)} = 0$ and 
\begin{equation*}
    \big\|\pd \w{E}^i\big\|_{{\bf L}^2(D)}^2 = \sum_{n = 0}^\infty \big|\big(\w{E}^i - \w{E}^i(0),E_n\big)_{{\bf L}^2(D)}\big|^2\,,
\end{equation*}
which implies 
\begin{equation} \label{auxeq:targetsysrhs}
    \pd \w{E}^i = O(\d) \q \text{and}\q  [\pd \w{E}^i, \tau^{-1} \pw \w{E}^i] = O(\d)\,,
\end{equation}
since $\w{E}^i - \w{E}^i(0) = O(\d)$. Then, by the estimate \eqref{est:neum_1}, we have 
\begin{equation} \label{eq:targetsyscoro}
\mbb{   \w{E} = O(\d/|\ww - \ww_0|)\,,}
\end{equation}
and the following lemma holds from the expansion \eqref{eq:cartesp}.
\begin{lemma}  \label{lem:multiradapp}
\mb{Under the conditions \eqref{1} and \eqref{2}, for $\ww \in U$ and $0 \le \d \le  \sqrt{C_A |\ww - \ww_0|}$,} the scattering amplitude $E_\infty^s$ can be approximated by  
\begin{align}  \label{eq:multiradiation3} 
    E_\infty^s = \frac{\tau \ww^2 \d^{3}}{4 \pi}(\I - \h{x}\otimes \h{x}) \Big\{\int_D \w{E}(\w{y})d\w{y} - i \d \ww  &\int_D \h{x} \dd \w{y} \w{E}(\w{y}) d\w{y} - \frac{\d^2 \ww^2}{2} \int_D (\h{x}\dd \w{y})^2\w{E}(\w{y})d\w{y} \Big\} + \mbb{ O\Big(\frac{\d^5}{|\ww - \ww_0|}\Big).}
\end{align}
\end{lemma}

\begin{remark} \label{rem:compare_non}
If $\tau$ is a constant of order one, we will have $[\pd \w{E}^i, \tau^{-1} \pw \w{E}^i] = O(1)$, and hence also $\w{E} = O(1)$. In this case, the error in the approximation \eqref{eq:multiradiation3} would be of order $\d^6$. By using \eqref{eq:cassepem} to separate the electric and magnetic multipoles, we can derive the following formula in classical electrodynamics \cite{jackson1999classical,evlyukhin2016optical}: 
\begin{align} 
    E^s_\infty(\ww,\h{x}) \approx \frac{\ww^2}{4 \pi}\big\{ \underbrace{\h{x}\t (\ed \t \h{x})}_{O(\d^3)} + \underbrace{\md \t \h{x} 
    + \h{x} \t (\h{x} \t \eq \h{x})}_{O(\d^4)} + \underbrace{\h{x} \t \mq \h{x}  + \h{x} \t \left(\h{x} \t \eo(\dd,\h{x},\h{x})\right)}_{O(\d^5)}\big\} \,,\label{eq:multiradiation2}
\end{align} 
where the electric dipole (ED) $\ed$, magnetic dipole (MD) $\md$, electric quadrupole (EQ) $\eq$, magnetic quadrupole (MQ) $\mq$ and electric octupole (EO) $\eo$ are given by $\eq_0^{\w{E}}$, $\mq_1^{\w{E}}$, $\eq_1^{\w{E}}$,  $\mq_2^{\w{E}}$ and $\eq_2^{\w{E}}$, respectively, up to some constant factors. 
The approximation \eqref{eq:multiradiation2} permits us to conclude that in the quasi--static regime, any dielectric nanoparticle behaves like an electric dipole in the far field, and its magnetic dipole radiation is a higher--order term. However, we shall see in Theorem \ref{thm:mainmultimom} below that in the case of high refractive index nanoparticles, the orders of these lower--order moments will be completely changed, and \mb{such change} does not depend on the excitation of dielectric resonances. 
\end{remark}

We return to our goal of finding the asymptotics of the solution $\w{E}$ to \eqref{eq:targetsys}. In view of Lemma \ref{lem:multiradapp} and the formula:
\begin{align*}
    \int_D \w{E} = \int_D \pw \w{E} 
\end{align*}
from $\int_D \pd \w{E} = 0$, in order to approximate the scattering amplitude $E^s_\infty$ up to $O(\d^5)$, we need to compute the second--order approximation of the field $\w{E}$, and also, particularly, the third--order approximation of $\pw \w{E}$. For this, by \eqref{eq:targetsyscoro}, it suffices to consider the second--order Neumann series expansion of $\AA(\d,\ww)^{-1}$:
\begin{align} \label{eq:neuexpad}
    \AA(\d,\ww)^{-1} =  &(I + \d \AA_0^{-1} \AA_1 + \d^2 \AA_0^{-1}\AA_2 + O(\d^3))^{-1} \AA^{-1}_0  \notag \\
    = &\AA^{-1}_0 - \d^2 \AA_0^{-1}\AA_2 \AA^{-1}_0  + O(\d^3)\,,
\end{align}
which follows from \eqref{eq:asyeqahat22}. By acting the expansion \eqref{eq:neuexpad} on $(\pd\w{E}^i, \tau^{-1}\pw \w{E}^i)$, we have the following proposition, the detailed proof of which is given in Appendix \ref{app:prfapp}. 

\begin{proposition} \label{prop:ledsolmain}
\mb{Under the conditions \eqref{1} and \eqref{2},} 
the solution to the equation \eqref{eq:targetsys} has the expansion: 
   \begin{equation} \label{eq:repsolmain}
    \resizebox{.90\hsize}{!}{$
        \left\{  \begin{aligned}
              &\pd \w{E} = \frac{-\ww_0}{2(\ww-\ww_0)} \P_{\lad_0}\w{E}^i + \frac{\ww_0  \tau^{-1} \ww^2}{2(\ww-\ww_0)}  \poo \kbdw (\tbb|_W)^{-1} \pw \big[\w{E}^i\big] + 
              \d \rr(\ww) + \mb{O\Big( \frac{\d^3}{|\ww-\ww_0|^2} \Big)}
              \,, \\
              & \pw \w{E} = - \tau^{-1}(\tbb|_W)^{-1} \pw \big[\w{E}^i\big] + \frac{\d^2\ww^2 \ww_0}{2(\ww-\ww_0)} (\tbb|_W)^{-1}\kbwd \P_{\lad_0} \big[\w{E}^i\big]  + 
              \d^3 \rr(\ww) + 
              \mb{O\Big(\frac{\d^4}{|\ww-\ww_0|} + \frac{\d^5}{|\ww - \ww_0|^2}\Big)}\,,
          \end{aligned}\right. $}
      \end{equation}
      for $\ww \in U$ and $0 \le \d \le \sqrt{C_A |\ww - \ww_0|}$.
\end{proposition}
\mbb{
Several remarks are in order. First, $\d^n \rr(\ww)$ means a term that allows a power series expansion in $\d$ of order $n$ with analytic coefficients, but it is convenient to understand it as a term which is of order $\d^n$, and analytic in a small neighborhood of $\ww_0$. Second, the remainder term of $\pw \w{E}$ in \eqref{eq:repsolmain} is estimated by $\d^4/|\ww - \ww_0|$ and $\d^5/|\ww - \ww_0|^2$, either of which may be dominant in the overall error, if there is no further assumption on the relative rate between $\d \to 0$ and $\ww \to \ww_0$. Third, recalling the representation  \eqref{def:resolmep}, 
Proposition \ref{prop:ledsolmain} essentially gives the approximation of the (scaled) scattering resolvent $(\M_0 - (\d\ww)^2 n^2)^{-1}$ with plane wave sources $f = (\d \ww)^2 \tau \chi_D \w{E}^i$ and $n = \sqrt{1 + \tau \chi_D}$. Indeed, by a simple computation, we have, in this case,
the resolvent is exactly the standard Lippmann--Schwinger formula:
\begin{align*}
  E^s = (\M_0 - (\d\ww)^2 n^2)^{-1} \big[(\d \ww)^2 \chi_D \w{E}^i\big]  = \TT^{\d \ww}_D \A_\d(\tau,\ww)^{-1} [\w{E}^i]\,,
\end{align*}
where $\A_\d(\tau,\ww)^{-1} [\w{E}^i]$ can be approximated by \eqref{eq:repsolmain}.


Proposition \ref{prop:ledsolmain} clearly shows the dominant parts of the $\hzz$ and $W$ components of $\w{E}$, and the blow--up rate and the order of the amplitude of each term involved in it. We are now ready to give the main results of this section. We first observe from the asymptotic formula \eqref{eq:repsolmain} and $\P_{\lad_0}\w{E}^i = O(\d)$ that 
\begin{equation} \label{est:prinpart}
    \pd \w{E} = O\Big(\frac{\d}{|\ww - \ww_0|} \Big)\,, \q \pw \w{E} = O\Big(\d^2 + \frac{\d^3}{|\ww - \ww_0|} \Big)\,.
\end{equation}
To see this, it is sufficient to write \eqref{eq:repsolmain} in the following form: 
 \begin{equation} \label{est:prinparttt}
        \left\{  \begin{aligned}
              &\pd \w{E} =  \frac{\d}{|\ww - \ww_0|} \Big( O(1) + O(\d) + O(|\ww - \ww_0|) 
              + O\Big(\frac{\d^2}{|\ww - \ww_0|} \Big) \Big)
              \,, \\
              & \pw \w{E} = \d^2 \Big(O(1) + O(\d) + O \Big(\frac{\d^2}{|\ww - \ww_0|} \Big) \Big)  + \frac{\d^3}{|\ww - \ww_0|} \Big(O(1) + O\Big(\frac{\d^2}{|\ww - \ww_0|}\Big)\Big) \,, 
          \end{aligned}\right.
\end{equation}
and note that $\d^2/|\ww - \ww_0|$ is bounded by a constant for $\ww \in U$ and $ 0 \le \d \le \sqrt{C_A|\ww - \ww_0|}$. Here, similarly to the remainder term of $\pw \w{E}$ in \eqref{eq:repsolmain}, $\pw \w{E}$ involves two leading--order terms: $- \tau^{-1}(\tbb|_W)^{-1} \pw \big[\w{E}^i\big]$ which is of order $\d^2$, and $\frac{ \d^2\ww^2 \ww_0}{2(\ww-\ww_0)} (\tbb|_W)^{-1} \kbwd\P_{\lad_0} \big[\w{E}^i\big]$ which is of order $\d^3$ but with a blow--up rate $(\ww-\ww_0)^{-1}$.

Then it follows from \eqref{est:prinpart} and Definition \ref{def:EMmoment} that
\begin{equation} \label{est:emmoments}
   \mq_l^{\w{E}} = O\Big(\frac{\d}{|\ww - \ww_0|} \Big)\q \text{for}\q  l = 1,2, \cdots\,,\q \eq_l^{\w{E}} = O\Big(\d^2 + \frac{\d^3}{|\ww - \ww_0|} \Big)\q \text{for}\q l =0, 1,2, \cdots\,.
\end{equation}
Thus, by \eqref{eq:samfuexp} and Lemma \ref{lem:multiradapp}, there are four moments involved: $\eq_l^{\w{E}}$, $l = 0, 1$ and $\mq_l^{\w{E}}$, $l = 1,2$, in order to achieve a fourth--order approximation of $E_\infty^s$ (in terms of $\d$), which is summarized as follows.} 
\begin{theorem} \label{thm:mainmultimom}
\mb{Under the conditions \eqref{1}, \eqref{2} and \eqref{3},} the scattering amplitude $E^s_\infty$ has the asymptotics:
\begin{align} \label{eq:highappfor}
    E^s_\infty(\ww,\h{x}) =  \frac{\tau \ww^2 \d^3}{4 \pi}(\I - \h{x}\otimes \h{x}) \Big\{\underbrace{\big(\eq_0^{\w{E}} - i\d \ww \mq_1^{\w{E}} \t \h{x}  \big)}_{O(\d^2/|\ww - \ww_0|)} - \underbrace{\big(i\d\ww  \eq_1^{\w{E}}\h{x} + \frac{\d^2\ww^2}{2} \mq_2^{\w{E}}(\dd,\h{x}) \t \h{x}\big)}_{O(\d^3/|\ww - \ww_0|)} \Big\} + \mbb{O\Big(\frac{\d^5}{|\ww - \ww_0|}\Big)} \,,
\end{align}
for $\ww \in U$ and $0 \le \d \le \sqrt{C_A |\ww - \ww_0|}$. 
\end{theorem}

By Theorem \ref{thm:mainmultimom} above, we can see the change of the orders of the multipole moments compared to those in the asymptotics \eqref{eq:multiradiation2}. This theoretically justifies the experimentally observed phenomenon \cite{kuznetsov2016optically}: when 
the wavelength inside the particle has the same order as its size (equivalently, $\tau \sim \d^{-2}$), the electric and magnetic dipole 
radiations may have the comparable strengths; see the first term in \eqref{eq:highappfor}.
However, this is true only when $\d$ is small enough so that the effect of the factor $1/|\ww - \ww_0|$ is not apparent. When the incident frequency $\ww$ approximates $\ww_0$ but $\d$ is not so small, we may expect from \eqref{est:emmoments} and \eqref{eq:highappfor} that the contribution of the magnetic dipole
is much stronger than the one of the electric dipole.
\mbb{To make this precise, we focus on the quasi--static approximation of $E^s_\infty(\ww,\h{x})$ from \eqref{eq:highappfor}:
\begin{align*}
    E^s_\infty(\ww,\h{x}) =  \frac{\tau \ww^2 \d^3}{4 \pi}(\I - \h{x}\otimes \h{x}) \left(\eq_0^{\w{E}} - i\d \ww \mq_1^{\w{E}} \t \h{x}  \right) + O\Big(\frac{\d^4}{|\ww - \ww_0|}\Big)\,,
\end{align*}
and proceed to compute $\eq_0^{\w{E}}$ and $\mq_1^{\w{E}}$ as follows.} By Definition \ref{def:EMmoment} and  Proposition \ref{prop:ledsolmain}, we have 
\begin{align} \label{auxest_6}
    \eq_0^{\w{E}} & = \int_D \pw \w{E} = \int_D - \tau^{-1}(\tbb|_W)^{-1} \pw \big[\w{E}^i\big] + \mbb{O\Big(\frac{\d^3}{|\ww-\ww_0|}\Big)}
    \\ 
    & = \tau^{-1} \int_{\p D} y(\frac{1}{2} + \np_{\p D}^*)^{-1}[\n \dd \w{E}^i] d \sigma(y) + \mbb{O\Big(\frac{\d^3}{|\ww-\ww_0|}\Big)} \,,
\end{align}
where we have used the following formula from \eqref{eq:inversetd} and \eqref{eq:greenfor_1}:
\begin{align*} 
    &\int_D (\tbb|_W)^{-1}[\vp] dy = - \int_{\p D} y (\frac{1}{2} + \np_{\p D}^*)^{-1}[\vp \dd \n] d\sigma(y)\,,\q \vp \in W\,.
\end{align*}
For the magnetic dipole moment $\mq_1^{\w{E}}$, we write $E_0 = \curl \vp_0$ for a unique $\vp_0 \in X_N^0(D)$ by Lemma \ref{prop:helmddifree} and obtain
\begin{equation} \label{eq:auxrepthmmain}
    \poo[\dd] = (\dd, \curl \vp_0)_{{\bf L}^2(D)} \curl \vp_0\,.
\end{equation}
By definition, it follows that 
\begin{align} \label{auxest_3}
   \mq_1^{\w{E}} =  \frac{-\ww_0}{2(\ww-\ww_0)} (\w{E}^i, \curl \vp_0)_{{\bf L}^2(D)} \int_D \vp_0(y) dy + \mbb{\d \rr(\ww) + O\Big(\frac{\d^2}{|\ww - \ww_0|} + \frac{\d^3}{|\ww - \ww_0|^2} \Big)}\,.
\end{align}
To further simplify the representations \eqref{auxest_6} and \eqref{auxest_3}, we observe 
\begin{align*}
     (\w{E}^i,E_0)_{{\bf L}^2(D)} 
    = i \d \ww(\di \t \ei e^{i\d \ww \di \dd x} ,\vp_0)_{{\bf L}^2(D)} = i\d \ww (\di \t \ei, \vp_0)_{{\bf L}^2(D)} + O(\d^2)\,,
\end{align*}
and 
\begin{align*}
  \n \dd \w{E}^i  = \n \dd \ei + O(\d)\,,
\end{align*}
by the Taylor expansion of $E^i$. Combining the above calculations, we arrive at the following theorem.

\begin{theorem} \label{thm:mainmultimom_3}
\mb{Under the conditions \eqref{1}, \eqref{2} and \eqref{3},} the scattering amplitude has the quasi--static approximation:  
    \begin{equation} \label{eq:quasiappmom_3}
        E^s_\infty(\ww,\h{x}) = \frac{ \ww^2 }{4 \pi} \{\h{x} \t (\ed\t \h{x}) + \md \t \h{x} \} + O\Big(\frac{\d^4}{|\ww - \ww_0|} + \frac{\d^5}{|\ww - \ww_0|^2} \Big)\,,
    \end{equation}
where the approximate electric dipole $\ed$ is given by 
\begin{equation*}
    \ed = \d^3 \int_{\p D}y \otimes (\frac{1}{2} + \np_{\p D}^*)^{-1}[\n]d \sigma(y)\ei\,, 
\end{equation*}
and the approximate magnetic dipole $\md$ has the pole--pencil expansion:
\begin{equation*}
  \md =  \ww^2 \d^3 \bigg\{\frac{-\ww_0}{2(\ww-\ww_0)} \Big(\int_D \vp_0 \otimes \int_D \overline{\vp_0} \Big) \di \t \ei + \rr(\ww)\bigg\}\,,
\end{equation*}
for $\ww \in U$ and $0 \le \d \le \sqrt{C_A|\ww - \ww_0|}$. 
\end{theorem}
\mbb{
We now consider the limit: $\ww \to \ww_0$. Theorem \ref{thm:mainmultimom_3} above suggests that in this case, $\md$ may blow up and the high contrast nanoparticles may exhibit strong magnetic dipole responses. To rigorously formulate this fact, we consider
\begin{align} \label{asp:dww}
    \d  = o(|\ww - \ww_0|^{1/2})\,,
\end{align}
which is a slightly stronger assumption than the one $\d = O(|\ww - \ww_0|^{1/2})$ in Theorem \ref{thm:mainmultimom_3}. 
The next result easily follows from Theorem \ref{thm:mainmultimom_3},} \mb{which clearly shows that the magnetic dipole radiation will dominate in the scattering amplitude
when the incident frequency $\ww$ approaches the limiting resonance $\ww_0$. 
This is in sharp contrast with the approximation derived in \cite{add8} in the non--resonant case, where it is the electric dipole radiation that is the leading--order term (see also Remark \ref{rem:compare_non}). We should also compare the result with the case of non--magnetic plasmonic nanoparticles where the scattering amplitude is approximated by the resonant electric dipole radiation \cite{ammari2016plasmaxwell,ammari2016surface}.}


\begin{corollary} \label{thm:mainmultimom_2}
\mb{Under the conditions \eqref{1}, \eqref{2} and \eqref{3}, it holds for $\ww \in U$ and $\d = o(|\ww - \ww_0|^{1/2})$} that
\begin{align} \label{eq:ppsa_led}
    E_\infty^s(\ww,\h{x}) \backsimeq \frac{\ww^2}{4 \pi} \h{\md} \t \h{x} \,, \q \ww \to \ww_0\,,
\end{align}
where $\h{\md}$ is the approximate resonant magnetic dipole:
\begin{equation} \label{eq:ppmd_led}
  \h{\md} =  \ww^2 \d^3 \frac{-\ww_0}{2(\ww-\ww_0)} \Big(\int_D \vp_0 \otimes \int_D \overline{\vp_0} \Big) \di \t \ei \,.
\end{equation}
\end{corollary} 

\mb{We end this section with some remarks and discussions. First, the assumption \eqref{asp:dww} is necessary for the concise form \eqref{eq:ppsa_led}. If we only consider $\d = O(|\ww - \ww_0|^{1/2})$, some additional terms of order $\d^3/|\ww - \ww_0|$ (i.e., the same order as $\h{\md}$) would be included in \eqref{eq:ppsa_led}. Second, under the quantitative assumptions between $\d$ and $\ww - \ww_0$, such as $\d \sim |\ww - \ww_0|^{1/2 + \epsilon}$, $0 < \epsilon \ll 1$, one may work out the order of the relative error in the approximation \eqref{eq:ppsa_led}, as well as the higher--order approximation involving EM quadrupoles based on \eqref{eq:highappfor}. Such calculations are direct and largely depends on the relative rate between $\d \to 0$ and $\ww \to \ww_0$ so that we choose not to pursue these results in this work.} 

Third, we have assumed that \eqref{3} holds at the beginning of this section for ease of exposition and simplicity of formulas. If $m_0 > 1$,  \eqref{eq:auxrepthmmain} will be given by the sum over all the vector potentials associated with $\lad_0$. Then it is easy to see that Theorem \ref{thm:mainmultimom_3} and Corollary \ref{thm:mainmultimom_2} still hold but with $\md$ and $\h{\md}$ being the sum over all the related potentials. If $D$ has genus $L$ greater than zero, we need to consider the space $K_T(D)$ (cf.\,Lemma \ref{prop:helmddifree}), for which we recall from
\cite[Theorem 3.44]{monk2003finite} that any field $\phi$ from the orthogonal complement 
$K_T(D)$ of $\ccurl X_N^0(D)$ in $\H_0(\ddiv0,D)$ can be represented as 
    \begin{equation*}
        \phi = \na p\q \text{for some}\ p \in H^1(D^0)\,,
    \end{equation*}
    where $D^0$ is constructed by removing some interior cuts from $D$ such that $D^0$ is a union of simply--connected 
    domains; see \cite{monk2003finite,amrouche1998vector} for more details. Clearly, $K_T(D)$ plays a similar role to $W$ in the scattering amplitude but generates a radiating field of order $\d$. To analyze it, one may define the associated multipole moments similar to $\eq^\psi_l$ in Definition \ref{def:EMmoment} and estimate their orders and blow--up rates when $\d \to 0$ and $\ww \to \ww_0$. 
    The detailed and rigorous treatments for this case may need further investigations.

\subsection{Scattering and extinction cross sections}  
This section is devoted to the estimates of the scattering and extinction rates, as well as their enhancements in the case where the nanoparticles are resonant. 

For this purpose, let us briefly review some basic physical concepts for the description of the energy flow (cf.\cite{born2013principles}). We define the time--averaged  energy flux $ \l {\bf S} \r $
for the electromagnetic fields $(E,H)$ and the associated averaged outward energy flow $\W$ through the sphere $\S_R$ by 
\begin{equation*}
    \l {\bf S} \r := \re \left\{ E \t \overline{H}\right\}\,, \q \W := \int_{\S_R} \l {\bf S} \r \dd \n d\sigma\,,
\end{equation*}
respectively.
By the decomposition: $E = E^i + E^s$, we can write $\W$ as $\W = \W^i + \W^s + \W'$, where
\begin{align*}
  &\W^i := \int_{\S_R} \re \big\{ E^i \t \overline{H^i}\big\} \dd \n d\sigma\,, \ \quad \W^s := \int_{\S_R} \re \big\{ E^s \t \overline{H^s}\big\} \dd \n d\sigma\,,\\
  & \W' := \int_{\S_R} \re \big\{ E^i \t \overline{H^s} + E^s \t \overline{H^i}\big\} \dd \n d\sigma\,.
\end{align*}
It follows from a simple calculation that the energy flow $\W^i$ for the incident plane wave is zero. Then the conservation of energy gives us the rate of absorption $\W^a = - \W =  - \W^s - \W'\,.$
Therefore, $\W'$ is nothing else than the rate at which the energy is dissipated by heat and scattering, referred to as the  extinction  rate. 
By normalizing with respect to the incident energy flux
$\l {\bf S}^i \r = \re \{ E^i \t \overline{H^i}\}$, we define the scattering cross section $Q^s$, the absorption cross section $Q^a$ and the extinction cross section $Q'$, respectively, by 
\begin{equation*}
    Q^s := \lim_{R \to \infty} \frac{\W^s}{\l {\bf S}^i \r}, \quad Q^a := \lim_{R \to \infty} \frac{\W^a}{\l {\bf S}^i \r}, \quad Q' := \lim_{R \to \infty} \frac{\W'}{\l {\bf S}^i \r}\,.
\end{equation*}
In particular, for our case of 
incident plane wave \eqref{eq:indplawave}, we have $\l {\bf S}^i \r = |\ei|^2 = 1$,
and the above formulas for cross sections 
are then simplified.


Note that in the far field, the scattered magnetic wave $H^s$ has the form:
\begin{equation} \label{addeqh}
     H^s(x) = \frac{e^{i\ww|x|}}{|x|} H^s_\infty(\ww,\h{x}) + O\Big( \frac{1}{|x|^2}\Big)\q 
     \text{with} \ H^s_\infty(\ww,\h{x}) = \h{x} \t E^s_\infty(\ww,\h{x}) \q \text{as}\ |x| \to \infty\,.
\end{equation}
Then, by \eqref{addeqe} and \eqref{addeqh},
we can directly compute $Q^s$ as follows:
\begin{align} \label{eq:repqs}
    Q^s &= \lim_{R \to \infty} \int_{\S_R} \re \left\{ E^s_\infty(\ww,\h{x}) \t  \overline{H^s}_\infty(\ww,\h{x})\right\} \dd \n(x) \frac{1}{R^2} d\sigma(x)  + O\Big(\frac{1}{R}\Big) \notag\\
    & =  \int_{\S} |E^s_\infty(\ww,\h{x})|^2 d\sigma(\h{x})\,.
\end{align} 
Moreover, the well--known optical cross section theorem \cite{born2013principles} shows the following representation of the extinction rate:
\begin{equation} \label{eq:repextr}
    Q' = \frac{4 \pi}{\ww} \im \left\{\ei \dd E^s_\infty(\ww,\di)\right\} \q \text{for}\ \ww > 0\,.
\end{equation}
With these preparations, we now give the estimates of the averaged cross sections.
\begin{theorem} \label{thm:blowssc}
\mb{Under the conditions \eqref{1}, \eqref{2} and \eqref{3}, 
for $\ww \in U$ with $\ww > 0$ and $\d = o(|\ww - \ww_0|^{1/2})$,} the averages over the orientations of the scattering and extinction cross sections of a randomly oriented nanoparticle can be estimated by
\begin{align*} 
Q^s_m \backsimeq \d^6 \frac{|\ww_0|^2|\ww|^8}{|\ww-\ww_0|^2}\frac{4\pi}{27} \big|\int_D \vp_0(y) dy\big|^4
\quad \text{and}\q 
    Q'_m \backsimeq  \d^3 \frac{-\ww_0\ww^3}{\ww-\ww_0} \frac{16 \pi^2}{9} \big|\int_D \vp_0 dy\big|^2\,, \q \text{as}\ \ww \to \ww_0\,,
\end{align*}
\mb{where $Q^s_m$ and $Q'_m$ are, respectively, given by the averages of $Q^s$ and $Q'$ over all the directions of $\ei$ and $\di$.} 
\end{theorem}

\begin{proof}
We start with the estimate of the scattering cross section $Q^s$. 
Noting from \eqref{eq:ppsa_led} and \eqref{eq:repqs} that
\begin{equation*}
    Q^s =  \int_{\S} |E^s_\infty(\ww,\h{x})|^2 d \sigma(\h{x}) \backsimeq \frac{|\ww|^4}{16 \pi^2} \int_\S |\h{\md} \t \h{x}|^2 d\sigma(\h{x})  = \frac{|\ww|^4}{16 \pi^2} \int_\S |\h{x} \t (\h{\md} \t \h{x})|^2 d\sigma(\h{x})\,,\q \ww \to \ww_0\,,
\end{equation*}
and 
$ 
    \h{x} \t (\h{\md} \t \h{x}) = \h{\md} - (\h{x}\dd \h{\md}) \h{x}\,, 
$ 
we derive
\begin{align*}
     Q^s  \backsimeq \frac{|\ww|^4}{16 \pi^2} \int_\S |\h{\md}|^2 - |(\h{x}\dd \h{\md})|^2 d\sigma(\h{x}) 
    = \frac{|\ww|^4}{16 \pi^2} \int_\S |\h{\md}|^2 -  \h{x}_j \h{\md}^j\h{x}^i \h{\md}_i d\sigma(\h{x}) =  \frac{|\ww|^4}{6 \pi} |\h{\md}|^2\,,
\end{align*}
where we have used the Einstein summation convention and the formula: 
\begin{equation} \label{eq:avgob}
    \int_\S \h{x}_i\h{x}_j d\sigma(\h{x}) = \frac{4 \pi}{3} \d_{ij}.
\end{equation}

Then, it follows from the formula \eqref{eq:ppmd_led} that
$Q^s_m$ can be estimated by
\begin{align} \label{eq:estqsm}
    Q^s_m & \backsimeq \frac{|\ww|^4}{6 \pi} \int_\S \int_\S \left|\h{\md}\right|^2 d \sigma(\di) d\sigma(\ei) \notag\\ 
    & = |\ww|^4 \d^6 \frac{|\ww_0|^2}{4|\ww-\ww_0|^2}\frac{|\ww|^4}{6 \pi}\int_\S \int_\S |\ap \di \t \e|^2 d \sigma(\di) d\sigma(\e)\,,\q \ww \to \ww_0\,, 
\end{align}
where 
\begin{align}\label{def:matrixavp}
\ap := \int_D \vp_0 \otimes \int_D \overline{\vp_0}   \,.
\end{align}
To calculate the integral \eqref{eq:estqsm}, we introduce the Levi--Civita symbol $\ep_{kij}$ to avoid the complicated vector and matrix calculations. First, we recall an important property for the symbol $\ep_{kij}$:
\begin{equation} \label{eq:prolcs}
    \ep_{kij}\ep^{kpq} =  \d_i^p \d_j^q - \d_j^p \d_i^q\,.
\end{equation}
Using \eqref{eq:avgob} and 
\eqref{eq:prolcs}, we can derive, for a general matrix ${\rm{\bf A}}:= (a_{ij})$,
\begin{align*}
  &\frac{1}{(4 \pi)^2}\int_\S \int_\S \left|{\rm{\bf A}}\di \t \e \right|^2d \sigma(\di) d\sigma(\e)  = \frac{1}{(4 \pi)^2} \int_\S \int_\S a^{sk}\ep_{kij} d^ie^j a_{sr}\ep^{rpq} d_p e_q \\
   =& \frac{1}{9} a^{sk}\ep_{kij} \d_p^i \d_q^j a_{sr}\ep^{rpq} = \frac{1}{9}  a^{sk}a_{sr}\ep_{kpq}\ep^{rpq} 
     = \frac{2}{9} a^{sk}a_{sr} \d_k^r = \frac{2}{9} \tr \left({\rm{\bf A}}{\rm{\bf A}}^T\right)\,.
\end{align*}
For the rank--one matrix $\ap$ in \eqref{def:matrixavp}, we have 
\begin{equation*}
    \tr(\ap \ap^T) = \big|\int_D \vp_0(y) dy\big|^4\,.
\end{equation*}
Combining the above formulas with \eqref{eq:estqsm}, we immediately derive the desired estimate:
\begin{equation*}
    Q^s_m \backsimeq \d^6 \frac{|\ww_0|^2|\ww|^8}{|\ww-\ww_0|^2}\frac{4\pi}{27} \big|\int_D \vp_0(y) dy \big|^4\,,\ \ww \to \ww_0\,.
\end{equation*}

To consider the extinction cross section, we first observe from \eqref{eq:ppsa_led} and \eqref{eq:repextr} that 
\begin{equation*} 
    Q' \backsimeq \ww \im \left\{\ei \dd (\h{\md} \t \di)\right\}.
\end{equation*}
With the help of \eqref{eq:ppmd_led}, we are led to the following approximation for the averaged $Q'$ over $\ei$ and $\di$,
\begin{equation} \label{auxestqem}
     Q'_m \backsimeq \ww^3 \d^3 \frac{-\ww_0}{2(\ww-\ww_0)} \im  \int_\S \int_\S \ei \dd \left(\left(\ap\di \t \ei\right) \t \di\right)   d \sigma(\di) d \sigma(\ei)\,,
\end{equation}
where $\ap$ is given in \eqref{def:matrixavp}. Similarly, by using the Levi--Civita symbol and the properties \eqref{eq:avgob} and \eqref{eq:prolcs},
 a direct calculation gives, for a general matrix ${\rm{\bf A}}:= (a_{ij})$,
\begin{align*}
  &\frac{1}{(4 \pi)^2} \int_\S  \int_\S \e \dd (({\rm{\bf A}} \di \t \e) \t \di)   d \sigma(\di) d \sigma(\e) = \frac{1}{(4 \pi)^2} \int_\S  \int_\S e_r \ep^{rst} a_{sk} \ep^{kij}d_i e_j  d_t \\ 
   =& \frac{1}{9} \ep^{rst} a_{sk} \ep^{kij}\d_{it}\d_{rj}  = \frac{1}{9} \ep^{rst} a_{sk} \ep^{ktr} 
   = \frac{2}{9} \tr({\rm{\bf A}})\,,
\end{align*}
which, together with \eqref{def:matrixavp} and \eqref{auxestqem}, completes the estimate for $Q'_m$.
\end{proof}

\section{Explicit formulas for a spherical nanoparticle} \label{sec:spernano}

In this section, we provide some explicit calculations for the special case of a single spherical nanoparticle to validate the asymptotic expansions of subwavelength resonances in \eqref{aim1} and the scattering amplitude in \eqref{eq:ppsa_led} which hold for nanoparticles of arbitrary shape. 



\subsection{Quasi--static resonances} 

We shall first find the eigen--decomposition of the operator $\np_{D}^{\rm d,d}$ with $D = B(0,1)$, by using the spherical multipole expansion and the following proposition. 


\begin{proposition}\label{prop:ledingqua}
Let $D$ be the bounded smooth open set as in Section \ref{sec:prosetanddef}. 
Then $\lad$ is an eigenvalue of $\kbdd$ if and only if the following system has a nontrivial solution in $\H^2_{{\rm loc}}(\R^3)$ with $k^2 = \lad^{-1}$:
\begin{equation} \label{eq:eigsphere}
    \left\{ 
    \begin{aligned}
        &\curl \curl u = 0\,, \ \ddiv u = 0 &\text{in}& \ \R^3 \backslash \bar{D}\,,\\
        &\curl \curl  ( u -\wg \gamma_n u ) = k^2 (u-\wg \gamma_n u) &\text{in}&\ D\,,\\
        &\left[\n \t u \right] = 0\,,\ \left[\n \dd u \right] = 0 \,,\ \left[\n \t \curl u \right] = 0 \quad &\text{on}&\ \p D\,, \\
        &u = O(|x|^{-2}) \quad &\text{as}&\ |x| \to \infty\,.
    \end{aligned}  
    \right.
\end{equation}
Moreover, for a solution $u$ to \eqref{eq:eigsphere}, $u-\wg \gamma_n u$ is an eigenfunction of $\kbdd$ associated with the  eigenvalue $\lad = k^{-2}$.
\end{proposition}


\begin{proof}
Suppose that $(\lad,E)\in \C \t \hzz$ is an eigenpair of the operator $\kbdd$ and define $u = \np_D [E] \in \H^2_{\rm loc}(\R^3)$. It is easy to check that $\ddiv u = 0$ on $\R^3$. Since the integration of $E \in \hzz$ over $D$ is zero, by the mean value theorem, we readily have 
\begin{equation*}
    \np_D[E](x) = \int_D \Big( \frac{1}{4\pi|x-y|}  - \frac{1}{4 \pi |x|}\Big)E(y)dy = O\Big(\frac{1}{|x|^2}\Big) \q \text{as}\ |x| \to \infty\,,
\end{equation*}
which implies that $u$ satisfies 
\bb \label{eq:sys_0}
 -\Delta u &= E \chi_D,\ \ddiv u = 0 \quad &\text{in}&\ \R^3\,, \\
 u &= O(|x|^{-2}) \quad &\text{as}&\ |x| \to \infty\,.
\ee
By use of the vector identity:
$ 
    \curl\curl = \na \ddiv  - \Delta \,,
$ 
we reformulate the equation \eqref{eq:sys_0} as 
\bb \label{eq:sys_1}
 &\curl \curl u = E \chi_D,\ \ddiv u = 0 \quad &\text{in}&\ \R^3\,, \\
 & u = O(|x|^{-2}) \quad &\text{as}&\ |x| \to \infty\,.
\ee
To proceed, since there hold $\curl \np_D [E] = \curl \kbdd [E]$ and $\kbdd[E] = \lad E$ on $D$,
we note from \eqref{eq:sys_1}: 
\begin{equation*}
    \curl \curl \np_D[E] = \curl \curl \lad E = E\chi_D\q \text{in}\ D\,,
\end{equation*}
which also yields $0$ is not an eigenvalue of $\kbdd$. Then, by $\kbdd [E] = \np_D [E] - \wg \gamma_n \np_D[E]$ for $E \in \hzz$, it follows from \eqref{eq:sys_1} that $u$ solves the following equation in the weak sense:
\bb \label{eq:sys_2}
&\curl \curl u = 0\,, \ \ddiv u = 0 \quad &\text{in}& \ \R^3 \backslash \bar{D}\,,\\
&\curl \curl  (u - \wg \gamma_n u ) = \frac{1}{\lad} (u-\wg \gamma_n u)\quad &\text{in}&\ D\,,\\
&\left[\n \t u \right] = 0\,,\ \left[\n \dd u \right] = 0 \,,\ \left[\n \t \curl u \right] = 0 \quad &\text{on}&\ \p D\,, \\
&u = O(|x|^{-2}) \quad  &\text{as}&\ |x| \to \infty\,.
\ee
Letting $\lad$ in \eqref{eq:sys_2} be $k^{-2}$ for some real $k \neq 0$ by the positive-definiteness of $\kbdd$, we obtain the desired system  \eqref{eq:eigsphere}. Conversely, if $u$ satisfies \eqref{eq:sys_2} with $\lad >0$, then $\curl u$, $\curl\curl u$ and $\ddiv u$ are globally well--defined on $\R^3$ and locally $L^2$-integrable. Therefore, we have $u \in \H^2_{\rm loc}(\R^3)$ solves
\bb 
 -\Delta u &= \frac{1}{\lad} (u - \wg \gamma_n u) \chi_D,\ \ddiv u = 0 \quad &\text{in}&\ \R^3\,, \\
 u &= O(|x|^{-2}) \quad &\text{as}&\ |x| \to \infty\,,
\ee
where we have used 
the fact from \cite{amrouche1998vector} that the space $\H_{\rm loc}(\ccurl,\R^3) \cap \H_{\rm loc}(\ddiv,\R^3)$ is continuously imbedded in $\H_{\rm loc}^1(\R^3)$. The proof is complete by the uniqueness of the solution $u$ to \eqref{eq:sys_0} for a given $E \in \hzz$ \cite{folland1995introduction}. 
\end{proof}



We should emphasize that Proposition \ref{prop:ledingqua} above applies to nanoparticles with arbitrary shape and the form of the equation \eqref{eq:eigsphere} is chosen mainly for the explicit calculation for the spherical nanoparticle. It is clear that once we obtain all the nontrivial solutions to the system \eqref{eq:eigsphere} with $D = B(0,1)$ and the associated $k^2$, we can determine the eigenstructure of $\np_{B(0,1)}^{\rm d,d}$. 

For this, we consider the series solutions of \eqref{eq:eigsphere} with $D = B(0,1)$. By the entire electric multipole fields in \eqref{eq:enete} and \eqref{eq:enetm}, we assume that for a nontrivial solution $u \in \H^2_{\rm loc}(\R^3)$ to \eqref{eq:eigsphere},  $u - \wg \gamma_n u$ in $B(0,1)$ is expanded as:
\begin{align} \label{eq:solin_1}
 u(x) - \wg \gamma_n u(x) = \sum_{n=1}^\infty \sum_{m=-n}^n \anm \wete (k,x) + \bnm \wetm(k,x)\,, \q x\in B(0,1)\,.
\end{align}
 Further, by Lemma \ref{lem:resfornormal}, $\wg \gamma_n u$ is the gradient of a solution $v$ to the interior Neumann problem \eqref{eq:neumann} with the boundary condition $\frac{\p}{\p \n} v = \n \dd u$.
We hence have the following series representation for $\wg \gamma_n u$:
\begin{align}
    \wg \gamma_n u(x) & = \sum_{n=1}^\infty \sum_{m=-n}^n \frac{c_{n,m}}{n} \na \left(|x|^n \ynm(\h{x})\right)\notag \\
     & = \sum_{n=1}^\infty \sum_{m=-n}^n \frac{c_{n,m}}{n} \left(|x|^{n-1} \na_{\S}\ynm(\h{x}) +  n |x|^{n-1} \ynm(\h{x}) \h{x}\right)\,, \q x\in B(0,1)\,.
    \label{eq:solin_2}
\end{align}
For the field $u$ outside the domain $B(0,1)$, by making use of the multipole fields $\{\eth\}$ constructed in Appendix \ref{app:B} and noting from $\ddiv u = 0$ on $\R^3$ that
$\int_\S \n \dd u d\sigma = 0$ holds, Proposition \ref{prop:expout} allows us to assume 
the following ansatz:
\begin{align} \label{eq:solout}
   u = \sum_{n=1}^\infty \sum_{m=-n}^n \gnm \eth + \enm \curl \eth\,, \q x \in \R^3 \backslash \overline{B(0,1)}\,.
\end{align}
Here, $\anm,\bnm,\gnm,\enm,c_{n,m}$ in \eqref{eq:solin_1}, \eqref{eq:solin_2} and \eqref{eq:solout} are complex coefficients to be determined by matching the traces of the field $u$ inside and outside $B(0,1)$ on $\S = \p B(0,1)$. 

We next state the main result of this section, the proof of which is given in Appendix \ref{app:eigvalspher}.

\begin{theorem} \label{thm:eigvalspher}
  Let $\{k_{n,s}\}_{s = 1}^\infty$ be the positive zeros of $j_n(z)$, $n \ge 0$, with $k_{n,1} \le k_{n,2} \le \cdots $. Then, 
 \begin{enumerate}
       \item  the eigenvalues of $\np_{B(0,1)}^{\rm d,d}$ are given by $1/k_{n,s}^2$, $n =0, 1,2,3,\cdots$ and $s = 1,2,\cdots$;
       \item  the eigenspace associated with $1/k_{n,s}^2$ is spanned by $\{\w{E}^{TE}_{n+1,m}(k_{n,s},x)\}_{m = -(n+1)}^{n+1}$ and $\{\wetm(k_{n,s},x)\}_{m=-n}^n$ for $n \ge 1$,
  and by $\{\w{E}^{TE}_{1,m}(k_{0,s},x)\}_{m = -1}^1$ for the case $n=0$.
   \end{enumerate}
\end{theorem}
\begin{remark} \label{rem:quasireso}
By the interlacing property of zeros of spherical Bessel functions $j_n$ \cite{watson1995treatise,liu2007zeros}, we have $1/k^2_{0,1} = 1/\pi^2$ and $1/k^2_{1,1}$ are the first and second largest eigenvalues of $\np_{B(0,1)}^{\rm d,d}$,
which corresponds to the magnetic and electric dipole resonances respectively by Mie theory. 
\end{remark}

\begin{remark} The approach developed in \cite{add3} can be generalized to compute the first Hadamard variation of the dielectric resonances. \end{remark} 

\subsection{Validation of the quasi--static approximations}

We now revisit the scattering problem \eqref{eq:model} with $D_\d = B(0,\d)$ and the high contrast $\tau = \d^{-2}$. We first recall the Jacobi--Anger expansion for the incident plane wave \eqref{eq:indplawave} \cite{colton2012inverse,ammari2013enhancement}:
\begin{align*}
    E^i(x) = e^{i \ww \di\dd x} \ei  = - \sum_{n=1}^\infty \frac{ 4 \pi i^n}{\sqrt{n(n+1)}} \sum_{m=-n}^n   \wete(\ww,x)  \overline{\vnm}(\di) \dd \ei + \wetm(\ww,x) \overline{\unm}(\di) \dd \ei\,, \q x\in\R^3\,.
 \end{align*}
We denote by $\ww_\tau$ the wave number $\ww\sqrt{1+\tau}$ inside the particle. The following lemma is standard from 
Mie scattering theory, see, e.g.,
\cite[Lemma A.2]{ammari2016plasmaxwell}. 

\begin{lemma} \label{lem:multiexpes}
The scattered wave $E^s = E - E^i$ outside the particle $B(0,\d)$ has the following representation:
\begin{equation*}
    E^s(x) = \sum_{n=1}^\infty \sum_{m = -n}^n \gnm \ete(\ww,x) + \enm \etm(\ww,x)\,,\q |x| > \d\,,
\end{equation*}
where the coefficients $\gnm$ and $\enm$ are given by 
\begin{equation}
    \begin{aligned}
        &\gnm = \frac{4 \pi i^n \overline{\vnm}(\di) \dd \ei}{\sqrt{n(n+1)}} \dd \frac{- j_n(\d \ww_\tau)\jj_n(\d \ww) + \jj_n(\d \ww_\tau) j_n(\d \ww)}{h_n^{(1)}(\d \ww) \jj_n(\d \ww_\tau) -  j_n(\d \ww_\tau) \hh_n(\d\ww)} \,, \\
        & \enm = \frac{ 4 \pi i^n \overline{\unm}(\di) \dd\ei}{\sqrt{n(n+1)}} \dd \frac{ \frac{1}{ 1+\tau}   \jj_n(\d \ww_\tau) j_n(\d \ww) -  j_n(\d \ww_\tau ) \jj_n(\d \ww) }{ \frac{1}{ 1+\tau} h_n^{(1)}(\d \ww) \jj_n(\d \ww_\tau) -  j_n(\d \ww_\tau) \hh_n (\d \ww)}\,.
    \end{aligned}
\end{equation}
\end{lemma}

\begin{remark}
The multipole fields $\wete$ and $\wetm$ are usually referred to as the magnetic (transverse electric) multipole modes and electric (transverse magnetic) multipole modes, which represent
the EM radiations from the electric--charge density and the magnetic--moment density of order $(n,m)$, respectively \cite[Section 9.7]{jackson1999classical}. In view of this, the coefficients $\gnm$ and $\enm$ may be called the magnetic and electric coefficients, respectively \cite{tzarouchis2018light}.  
\end{remark}

It is clear from Lemma \ref{lem:multiexpes} that the scattering resonances are characterized by the complex
zeros of denominators of $\gnm$ and $\enm$:
  \begin{align*}
    & h_n^{(1)}(\d \ww)\jj_n(\d \ww_\tau) -  j_n(\d \ww_\tau)\hh_n(\d\ww) = 0\,, \\
     & (1+\tau)^{-1} h_n^{(1)}(\d \ww) \jj_n(\d \ww_\tau) - j_n(\d \ww_\tau)\hh_n (\d \ww)  = 0\,.
 \end{align*}
To numerically verify the estimate \eqref{aim1}, we note that $\lad_0 = 1/\pi^2$ and $\ww_0 = \sqrt{\lad_0}^{\sss - 1} = \pi$ for the spherical domain (see Remark \ref{rem:quasireso}) and 
consider the resonance $\ww$ that is nearest to the origin. We plot the first zero of 
$ 
     h_n^{\sss (1)}(\d \ww)  \jj_n(\d \ww_\tau) -  j_n(\d \ww_\tau)\hh_n(\d\ww)
$ 
with $\tau = \d^{-2}$ and $\d \in (0,0.22)$ in Figure \ref{fig:reso_real} below, using the Muller's method \cite{muller1956method} with the initial value $\pi$, from which we  immediately observe that the error between the subwavelength resonance $\ww(\d)$ and its quasi--static approximation $\pi$
is bounded by a term of order $\d^2$ as $\d \to 0$.

\begin{figure}[!htbp]
 \centering
\includegraphics[clip,width=0.5\textwidth]{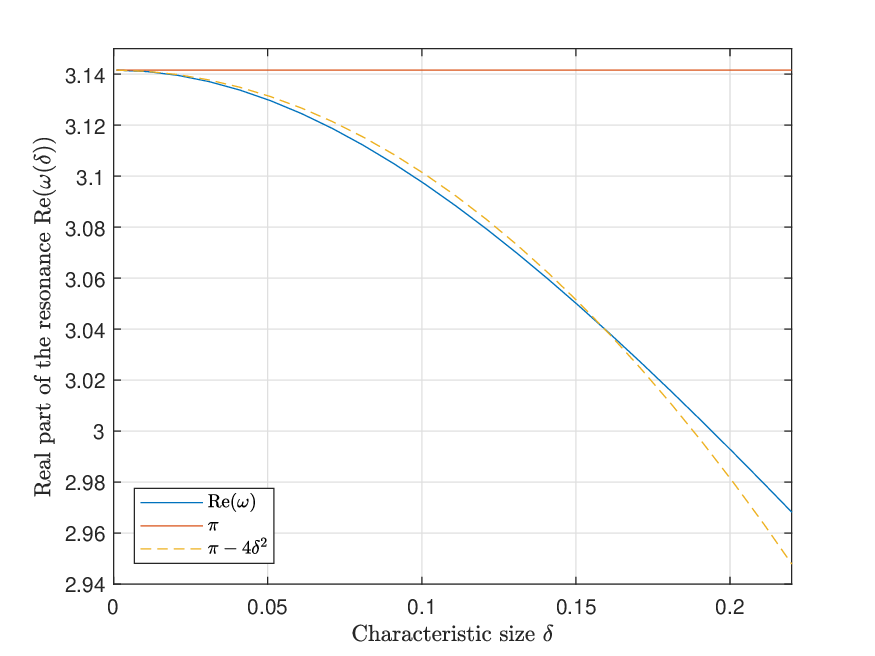}   \hskip -0.5cm
\includegraphics[clip,width=0.5\textwidth]{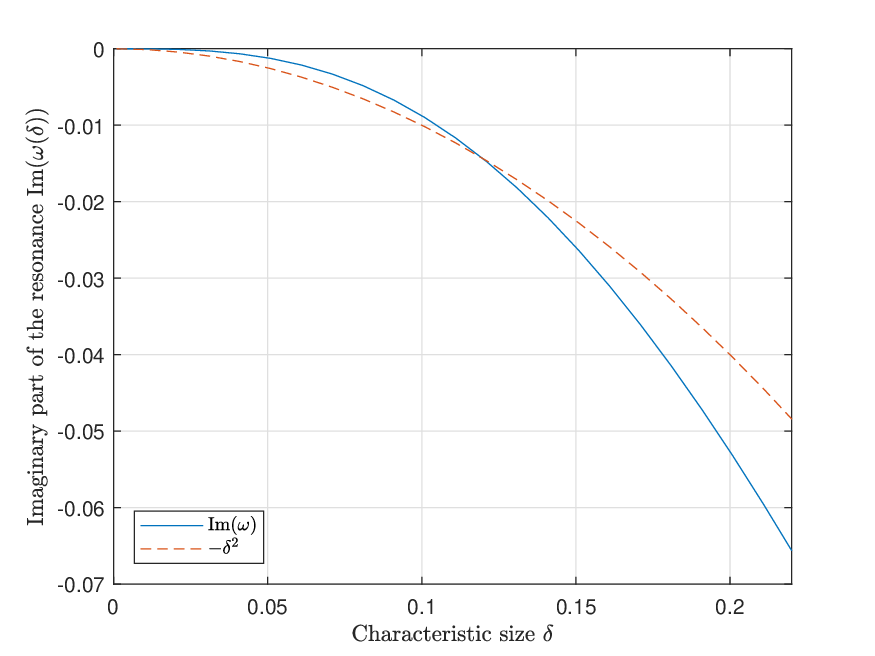} 
\caption{Lowest dielectric subwavelength resonance $\ww$ (blue curve) for a spherical nanoparticle $B(0,\d)$ with $\d \in (0,0.22)$ and the contrast $\tau = \d^{-2}$.} \label{fig:reso_real}
\end{figure}

To proceed our validation of the approximation formula \eqref{eq:ppsa_led} for the scattering amplitude, we consider the regime where $\tau \sim \d^{-2}$, $\ww$ near $\ww_0 = \pi$ and $\d = o(|\ww - \ww_0|^{1/2})$, and have the following asymptotic estimates for the coefficients $\gnm$ and $\enm$ (cf.\,\eqref{eq:estforcoeff_1}--\eqref{eq:estforcoeff_3}):
\begin{align} \label{eq:estco_gamma}
    \gnm = \frac{4 \pi i^n \overline{\vnm}(\di) \dd \ei}{\sqrt{n(n+1)}} \dd \left\{ 
\begin{aligned}
    &\frac{i}{3}(\d \ww)^3 \frac{-2j_1(\d \wt)+\jj_1(\d \ww_\tau) }{\jj_1(\d \ww_\tau) + j_1(\d \ww_\tau)} + O\Big(\frac{\d^5}{\jj_1(\d\wt) + j_1(\d \wt)}\Big)\,,  &  n = 1\,, \\
     & O\Big(\frac{\d^{2n+1}}{ \frac{1}{n}\jj_n(\d \wt) + j_n(\d \ww_\tau)}\Big)\,,     & n \ge 2\,,
\end{aligned}
    \right.
\end{align}
and 
\begin{equation} \label{eq:estco_eta}
    \enm = \frac{ 4 \pi i^n \overline{\unm}(\di) \dd\ei}{\sqrt{n(n+1)}} \dd \left\{\frac{\jj_n(\d \ww)}{\hh_n(\d\ww)} + O\Big(\frac{\d^{2n+3}}{j_n(\d \wt)}\Big)\right\}\,,\q n \ge 1\,.
\end{equation}

By estimates \eqref{eq:estco_gamma} and \eqref{eq:estco_eta}, Lemma \ref{lem:multiexpes} gives the following quasi--static approximation for the scattered wave $E^s$: for $\tau = \d^{-2}$ and $\d = o(|\ww - \ww_0|^{1/2})$,
\begin{align*}
    E^s & = \sum_{m = -1}^1\gamma_{1,m} E_{1,m}^{TE} + \eta_{1,m} E_{1,m}^{TM} + O\big(\frac{\d^5}{\ww - \ww_0}\big)  \backsimeq  \sum_{m = -1}^1\h{\gamma}_{m} E_{1,m}^{TE}\,,\q   \ww \to \ww_0\,,
\end{align*}
where $\h{\gamma}_{m}$ is defined by the leading--order term of $\gamma_{1,m}$:
\begin{equation} \label{eq:hgmspn}
    \h{\gamma}_{m} =- (\d \ww)^3 \frac{2\sqrt{2} \pi  }{3} \dd
    \frac{-2j_1(\d \wt)+\jj_1(\d \ww_\tau) }{\jj_1(\d \ww_\tau) + j_1(\d \ww_\tau)}\overline{V_1^m}(\di) \dd \ei\,.
\end{equation}
Then by \eqref{eq:ffpete}, the scattering amplitude can be approximated by
\begin{equation*}
    E_\infty^s(\ww,\h{x}) \backsimeq \sum_{m = -1}^1 \frac{\sqrt{2}}{\ww} \h{\gamma}_{m}V^m_{1}(\h{x})\,,\q \ww \to \ww_0\,,
\end{equation*} 
which can be further reformulated as 
\begin{equation*}
    E_\infty^s(\ww,\h{x}) \backsimeq \sum_{m = -1}^1 \frac{1}{\ww} \h{\gamma}_{m}\h{x}\t{\bf Y}_m\,, \q \ww \to \ww_0\,,
\end{equation*}
by the vector spherical harmonics in \eqref{def:vsh} and the vectors ${\bf Y}_j$ in  \eqref{def:gahhp}. Comparing it with the formula \eqref{eq:ppsa_led}, we obtain the approximate resonant magnetic dipole for a spherical nanoparticle:
\begin{equation*}
    \w{\md} = - \sum_{j=-1}^1\frac{4\pi}{\ww^3} \h{\gamma}_j {\bf Y}_j\,,
\end{equation*}
which, along with \eqref{eq:hgmspn}, gives
\begin{equation} \label{eq:expappmd}
  \w{\md} = \d^3 \frac{8\pi^2}{3} \dd
    \frac{2j_1(\d \wt)-\jj_1(\d \ww_\tau) }{\jj_1(\d \ww_\tau) + j_1(\d \ww_\tau)} \sum_{j=-1}^1{\bf Y}_j \otimes \overline{{\bf Y}}_j (\di \t \ei)\,,
\end{equation}
where we have used 
\begin{equation*}
    \overline{V_1^m}(\di)\dd \ei = - \frac{1}{\sqrt{2}} \overline{ {\bf Y}}_m \dd( \di \t \dd \ei)\,. 
\end{equation*}
The explicit representation for the magnetic dipole \eqref{eq:expappmd} motives us to define a scattering function:
\begin{equation} \label{def:scafun_exp}
    \w{s}_m(\ww) = \frac{8\pi^2}{3} \dd
    \frac{2j_1(\d \wt)-\jj_1(\d \ww_\tau) }{\jj_1(\d \ww_\tau) + j_1(\d \ww_\tau)}\,.
\end{equation}

We shall numerically compare $\w{s}_m(\ww)$ with the one derived from the general formula \eqref{eq:ppsa_led}.
By generalizing Corollary \ref{thm:mainmultimom_2} to the case of a multidimensional eigenspace, we obtain 
the following representation for the approximate magnetic dipole:
\begin{equation} \label{check:appmd_0}
    \h{\md} = \ww^2 \d^3 \frac{-\ww_0}{2(\ww-\ww_0)} \sum_{j=-1}^1\Big(\int_{B(0,1)} \vp_j \otimes \int_{B(0,1)}\overline{\vp_j} \Big) \di \t \ei \,,
\end{equation}
where $\ww_0 = \pi$ and $\vp_j$ are the potentials chosen from $X_N^0$ such that $\curl \vp_j = \w{E}^{TE}_{1,j}/\|\w{E}^{TE}_{1,j}\|_{{\bf L}^2(B(0,1))}$.
To explicitly find $\vp_j$, recalling the definition of $\w{E}^{TE}_{1,j}$ in  \eqref{eq:enete},
we have
\begin{equation*}
    \int_{B(0,1)}|\w{E}^{TE}_{1,j}(\pi ,x)|^2 dx = \int_0^1  2 j_1(\pi r)j_1(\pi r) r^2 dr = \frac{1}{\pi^2}\,,
\end{equation*}
by the Lommel's integral \cite{watson1995treatise}:
$
    \int_0^1 j_n(a r)j_n(a r)r^2 dr = 
    \frac{1}{2}(j_n^2(a) - j_{n+1}(a)j_{n-1}(a))\,,
$
and then introduce $\vp_j \in X_N^0$ by
\begin{equation*}
    \vp_j(x) = \pi x j_1(\pi |x|) Y_1^{j}(\h{x}) + \na p_j\,,
\end{equation*}
where $p_j \in H_0^1(B(0,1))$ is constructed such that $\ddiv \vp_j = 0$, by solving a Laplacian Dirichlet problem. It further follows that, by \eqref{def:sh1-1}--\eqref{def:sh10}
and \eqref{eq:avgob}, 
\begin{align*}
    \int_{B(0,1)} \vp_j dx 
    = \pi \int_0^1 r^3  j_1(\pi r) dr \int_\S\h{x} Y_1^j(\h{x}) d\sigma(\h{x}) = \frac{4}{\pi} {\bf Y}_j\,,
\end{align*}
due to $\int_{B(0,1)} \na p_j dx = 0$. We substitute the above formula into \eqref{check:appmd_0} and find
\begin{equation*}  
    \h{\md} = \d^3 \frac{-\ww_0 \ww^2}{\ww-\ww_0} \frac{8}{\pi^2} \sum_{j=-1}^1{\bf Y}_j \otimes \overline{{\bf Y}}_j  \di \t \ei \,,
\end{equation*}
which gives us another scattering function:
\begin{equation} \label{def:scafun_gen}
    \h{s}_m(\ww) = - \frac{8}{\pi^2} \frac{\ww^2\ww_0}{\ww-\ww_0}\,.
\end{equation}
We plot the scattering functions $\w{s}_m$ and $\h{s}_m$ in Figure \ref{fig:scafunction} with $\tau = \d^{-2}$  and $\d = 0.15$, from which we see that they agree very well with each other.


\begin{figure}[!htbp]
 \centering
\includegraphics[clip,width=0.53\textwidth]{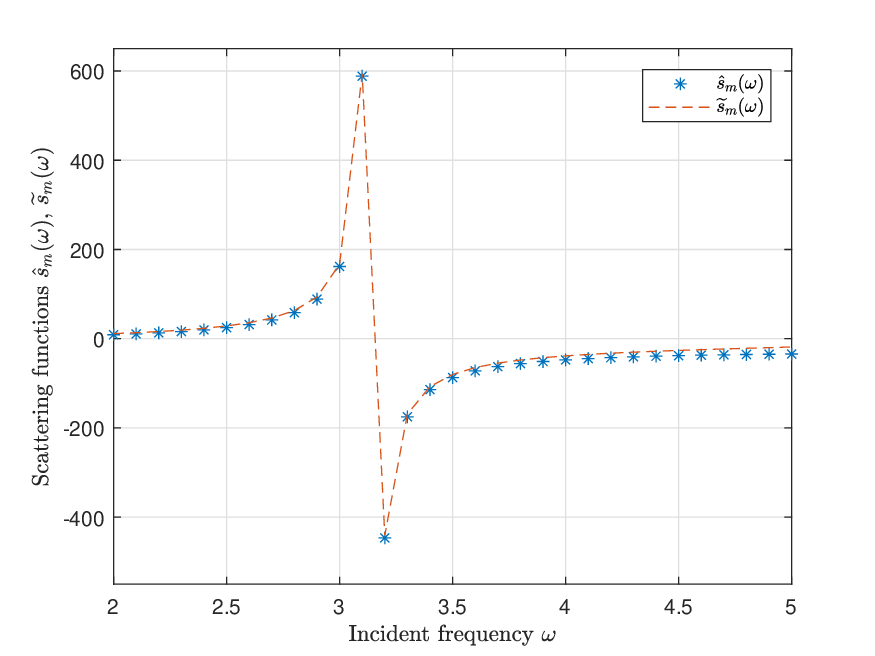}  
\caption{Scattering functions $\w{s}_m(\ww)$ and $\h{s}_m(\ww)$ for incident frequencies near $\ww_0 = \pi$ with $\d = 0.15$ and $\tau = 0.15^{-2}$.}\label{fig:scafunction}
\end{figure}

\section{Concluding remarks and discussions} \label{sec:conanddis}

In this work, we have comprehensively investigated the EM scattering by nanoparticles with high refractive indices. The geometry of the particles can be of very general shape, e.g., a torus with a hole. We have mathematically proved the existence of the dielectric subwavelength resonances and derived their asymptotic expansions when the contrast is large enough. A priori estimates 
for resonances and the associated resonant modes have also been provided. To explore the behavior of the scattering amplitude, we have proposed a systematic  approach to derive the EM multipole moment expansion based on the Helmholtz decomposition. We have theoretically found that near the resonant frequencies, the scattered field by nanoparticles with high refractive indices is dominated by the resonant magnetic dipole radiation. We have also estimated the enhancements of the scattering and extinction cross sections.



Because of the low losses of dielectric nanostructures ($\im \tau \ll 1$; see Figure \ref{fig:parasili}), the dielectric subwavelength resonances $\ww(\d,\tau) \approx (\lad_i \d^2 \tau)^{-1/2}$ (cf.\,\eqref{eq:existence_resonance})  can have a very high quality factor $Q = |\re \ww/ \im \ww|$, which is a desirable property in many potential applications and gives rise to the superior performance of the all--dielectric metamaterials over the lossy plasmonic devices. 
In a forthcoming paper, we plan to combine the results obtained in this paper on the subwavelength resonances in dielectric nanoparticles with high refractive index together with the effective medium approaches in \cite{add5,add6,add4} for the design of all--dielectric, electromagnetic metamaterials. 

\bigskip

\mb{{\bf Acknowledgements}. The authors 
would like to thank two anonymous referees for their very careful reading of the manuscript and their numerous insightful and constructive comments and suggestions, which have helped us improve the presentation and results of this work essentially.}

\titleformat{\section}{\bfseries}{\appendixname~\thesection .}{0.5em}{}
\titleformat{\subsection}{\normalfont\itshape}{\thesubsection.}{0.5em}{}
\appendices

\section{Justification of the asymptotic regime and the assumptions} \label{app:S_1}
\setcounter{figure}{0}
\renewcommand\thefigure{A.\arabic{figure}}

To verify the quasi--static regime \eqref{eq:subwave} and the high contrast assumption \eqref{asp:highcont}, we give the practical values of the physical parameters in the experiments \cite{kuznetsov2016optically,evlyukhin2012demonstration,garcia2011strong}, where the dielectric resonances and the strong magnetic responses are observed. We consider the scattering by a silicon nanoparticle in the free space. We have 
\begin{enumerate}[\textbullet]
    \item for the background medium,\small
    \begin{equation*}
        \mu_{{\rm m}} \approx 12 \times 10^{-7} {\rm H\dd m^{-1}}\,,\q \ep_{{\rm m}} \approx  9 \times 10^{-12} {\rm F\dd m^{-1}}\,,\q \text{(speed of light)}\ c \approx 3 \t 10^8 {\rm m \dd s^{-1}};
    \end{equation*}\normalsize
    \item for the silicon nanoparticle,\small
    \begin{align*}
    \mu_r \approx 1\,,\q \ep_r \approx [10, 50]\,,\q \d \approx [50,200] \t 10^{-9} {\rm m};    
    \end{align*}\normalsize
    \item for the incident wave, \small
    \begin{equation*}
        \ww \approx 2\pi \t [350, 650]\t 10^{12} {\rm Hz}\,,\q \text{(wavelength)}\ \lad = \frac{2 \pi c}{\ww} \approx [461,850]\t 10^{-9} {\rm m}.
    \end{equation*}\normalsize
\end{enumerate}
It is clear that $\d \ll \lad$ and $\ep_r \gg 1$, i.e., \eqref{eq:subwave} and \eqref{asp:highcont} hold. \mb{We also see from the above configuration (where
 the resonance occurs) that the wavelength inside the particle is comparable with the particle size, namely, $\d \sim \lad \ep_r^{\sss -1/2}$. To verify the assumptions $0 < \d \ll 1$ and
 \eqref{assp:taudelta} (or \eqref{asp:tauexpd}), it suffices to consider $10^{-6} {\rm m}$ as the unit length. Then we have $\lad \sim 1$ and $\d \ll 1$, and \eqref{assp:taudelta} readily follows from the relation $\d/\lad \sim \ep_r^{\sss -1/2}$.}  

We plot in Figure \ref{fig:parasili} the experimentally measured values of the relative electric permittivity, as well as the refractive index, of silicon as a function of the incident wavelength (the data is from \cite{aspnes1983dielectric}). We see that in the visible region, the refractive index $n$ is approximately $4$ with a very small imaginary part.
\begin{figure}[!htbp]
    \centering 
    \subfigure[Relative electric permittivity]{\label{fig:siliconep}
       \includegraphics[clip,width=0.5\textwidth]{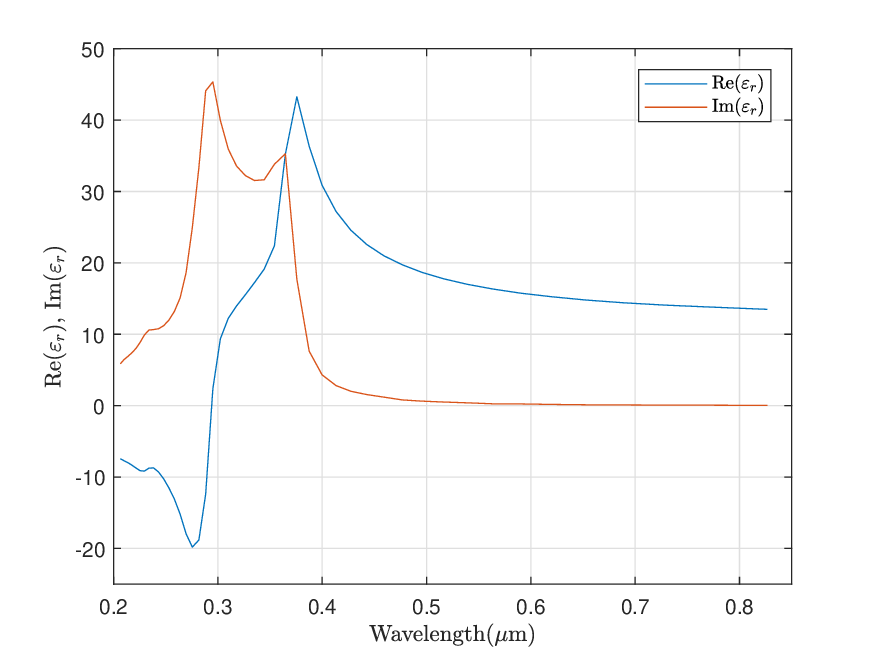}}  \hskip -0.8 cm
        \subfigure[Refractive index]{\label{fig:siliconreind}
       \includegraphics[clip,width=0.5\textwidth]{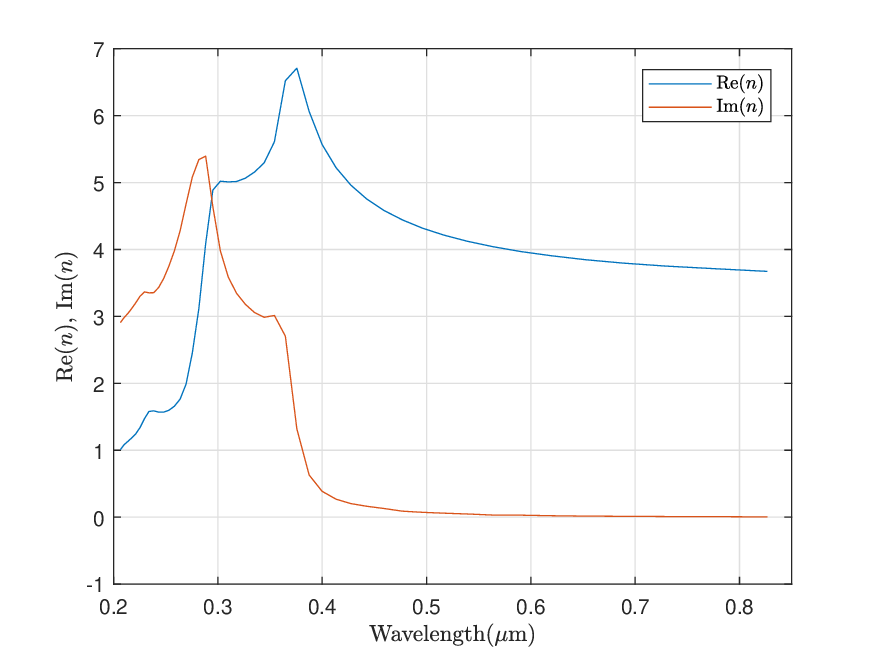}} 
      \caption{Relative electric permittivity and refractive index for silicon \cite{aspnes1983dielectric}.}
    \label{fig:parasili} 
\end{figure}

\if \commentflag = \ct
\section{Trace of the finite rank operator}\label{app:D}
We provide a framework to calculate the trace of a finite rank operator $T$ on a Banach space $X$. Suppose that the range of $T: X\to X$ is spanned by the basis $\{v_k\}_{k=1}^n \subset X$, with its dual set in $X^*$ given by $\{v_k^*\}_{k=1}^n$, i.e., 
$\l v_i^*, v_j \r_X = \d^j_{i}$. Then, the following representation of $T$ readily follows,
\begin{equation*}
    Tx = \sum_{k=1}^n \l v_k^*, T x\r_X v_k\,.
\end{equation*} 
With the help of the above representation, the trace of $T$ is given by 
\begin{align} \label{def:trafinite}
    \tr(T) = \sum_{k=1}^n \l v_k^*, Tv_k \r_X.
\end{align}
By some fundamental algebraic calculations, we can check that the definition in \eqref{def:trafinite} is independent of the choice of the basis $\{v_k\}$. If $X$ is a Hilbert space, we may replace the dual pairing $\l\dd,\dd\r_X$ in the above formulas with the inner product $(\dd,\dd)_X$. In this case, $v_i^*$ is given by $v_i$ if $v_i$ is normalized, i.e., $\norm{v_i}_X = 1$. 
\fi

\section{Some definitions, asymptotics and estimates}\label{app:A}
\begin{enumerate}[\textbullet]
    \item Let $\ynm(\h{x}),\  n =0, 1,2,\cdots,\ m=-n,\cdots,n,$ be the spherical harmonics on the unit sphere $\S$, which are the eigenfunctions of the Laplace--Beltrami operator:
  $
        \Delta_\S \ynm(\h{x}) + n(n+1)\ynm(\h{x}) = 0\,.
$    
    We define the vector spherical harmonics as follows:
    \small
    \begin{equation} \label{def:vsh}
        \unm = \frac{1}{\sqrt{n(n+1)}} \na_\S \ynm\,,\q \vnm = \h{x} \t \unm\,, \ n = 1,2,\cdots,\ m=-n,\cdots,n\,.
    \end{equation}
    \normalsize
  For $n = 1$, $m = -1, 0, 1$, the spherical harmonics $Y_n^m(\h{x})$ can be explicitly represented as:\small
    \begin{align}
        &  Y_{1}^{-1}(\h{x}) =  \frac{1}{2}\sqrt{\frac{3}{2\pi}} \sin \theta e^{-i\vp} = \frac{1}{2}\sqrt{\frac{3}{2\pi}}  (\h{x}_1 - i \h{x}_2)\,, \label{def:sh1-1}\\
        &    Y_{1}^{1}(\h{x}) = - \frac{1}{2}\sqrt{\frac{3}{2\pi}} \sin \theta e^{i\vp}=- \frac{1}{2}\sqrt{\frac{3}{2\pi}} (\h{x}_1 + i \h{x}_2)\,,\label{def:sh11} \\
        &   Y_1^0(\h{x}) = \frac{1}{2}\sqrt{\frac{3}{\pi}} \cos \theta = \frac{1}{2}\sqrt{\frac{3}{\pi}} \h{x}_3\,,\label{def:sh10}
        \end{align}
         \normalsize
        with the azimuthal angle $0 \le \vp \le 2\pi$ and the polar angle $ 0 \le \theta \le \pi$. By \eqref{def:sh1-1}--\eqref{def:sh10}, 
 we have the explicit formulas for
        the gradients of homogeneous harmonic polynomials of degree one ${\bf Y}_j: = \na (|x|Y_{1}^{j}(\h{x}))$, $i = -1,0,1$:
        \small
      \begin{align}  \label{def:gahhp}
              {\bf Y}_{-1} = \frac{1}{2}\sqrt{\frac{3}{2\pi}}  \mm 1 \\ -i \\ 0 \nn, \quad   {\bf Y}_1 =- \frac{1}{2}\sqrt{\frac{3}{2\pi}} \mm 1\\i \\0 \nn,
            \quad   {\bf Y}_0 = \frac{1}{2}\sqrt{\frac{3}{2\pi}}  \mm 0\\0\\ \sqrt{2} \nn.
       \end{align}
        \normalsize
    \item Define the entire electric multipole fields $\w{E}^{TE}_{n,m}(k,x)$ and $\w{E}^{TM}_{n,m}(k,x)$ on $\R^3$ with the wave number $k$ for $n = 1,2,\cdots$ and $m=-n,\cdots,n$:
    \small
    \begin{align}
       \wete(k,x) & = \curl \{x j_n(k |x|)\ynm(\h{x})\} =  - \sqrt{n(n+1)} j_n(k|x|)\vnm(\h{x}), \label{eq:enete} \\
       \wetm(k,x)  & = - \frac{1}{i k } \curl \wete(k,x)\notag\\ &  = -\frac{\sqrt{n(n+1)}}{ik|x|}\jj_n(k|x|)\unm(\h{x}) - \frac{n(n+1)}{ik|x|}j_n(k|x|)\ynm(\h{x})\h{x}\,, \label{eq:enetm}
    \end{align}
     \normalsize
    where $j_n(t)$ is the spherical Bessel function of order $n$, and $\jj_n(t)$ is given by $\jj_n(t) := j_n(t) + tj'_n(t)$.    
    The radiating electric multipole fields $\ete(k,x)$ and $\etm(k,x)$ on $\R^3\backslash\{0\}$ can be introduced similarly with $j_n(t)$ and $\jj_n(t)$ in \eqref{eq:enete} and \eqref{eq:enetm} replaced by the spherical Hankel function $h_n^{\sss (1)}(t)$ and $\hh_n(t) := h_n^{\sss (1)}(t) + t(h_n^{\sss (1)})'(t)$. 
    \item When $t \to 0$, the following asymptotics for spherical Bessel functions and Hankel functions hold  \cite{jackson1999classical,abramowitz1964handbook}:
    \small
    \begin{align}
        j_n(t) & = \frac{t^n}{(2n+1)!!} \big(1 - \frac{1}{2(2n+3)}t^2 + O(t^4)\big)\,, \label{eq:asyzero_j}\\
        h_n^{(1)}(t) & = -i \frac{(2n-1)!!}{t^{n+1}}\big(1 + \frac{1}{2(2n-1)}t^2 + O(t^4)\big)\,,\label{eq:asyzero_h}\\
        \jj_n(t) &= \frac{t^n}{(2n+1)!!} \big(n + 1 - \frac{n+3}{2(2n+3)}t^2  + O(t^4) \big)\,,\label{eq:asyzero_jj}\\
        \hh_n(t) &= -i \frac{(2n-1)!!}{t^{n+1}}\big(  - n - \frac{n - 2}{2(2n-1)}t^2 + O(t^4) \big)\,.\label{eq:asyzero_hh}
   \end{align}
    \normalsize
\item Wave functions  $\ete$ and $\etm$ have the following far--field behaviors: when $|x| \to \infty$,
\small
\begin{align} 
    &\ete(k,x) = -\sqrt{n(n+1)}\frac{e^{ik|x|}}{k|x|} e^{-i\frac{n+1}{2}\pi} \vnm(\h{x}) +O\Big(\frac{1}{|x|^2}\Big)\,, \label{eq:ffpete} \\
&\etm(k,x) = - \sqrt{n(n+1)}\frac{e^{ik|x|}}{k|x|} e^{-i\frac{n+1}{2}\pi}\unm(\h{x}) + O\Big(\frac{1}{|x|^2}\Big)\,. \label{eq:ffpetm}
\end{align}
 \normalsize
\item By the asymptotic forms in  \eqref{eq:asyzero_j}--\eqref{eq:asyzero_hh}, for $\tau \sim \d^{-2}$, $\ww$ near $\ww_0 = \pi$ and $\d = o(|\ww - \ww_0|^{1/2})$,
we have the following estimates used in \eqref{eq:estco_gamma} and \eqref{eq:estco_eta}: 
as $\d \to 0$, \small
 \begin{align}
 \label{eq:estforcoeff_1}
    &\frac{- j_n(\d \ww_\tau)\jj_n(\d \ww) + \jj_n(\d \ww_\tau) j_n(\d \ww)}{h_n^{(1)}(\d \ww) \jj_n(\d \ww_\tau) -  j_n(\d \ww_\tau) \hh_n(\d\ww)}\notag\\
    =& \frac{- j_n(\d \ww_\tau)(\jj_n(\d \ww)/\hh_n(\d\ww)) + \jj_n(\d \ww_\tau) (j_n(\d \ww)/\hh_n(\d\ww))}{(h_n^{(1)}(\d \ww)/\hh_n(\d\ww)) \jj_n(\d \ww_\tau) -  j_n(\d \ww_\tau) }\notag\\
   = & \frac{- j_n(\d \ww_\tau)O(\d^{2n+1}) + \jj_n(\d \ww_\tau) O(\d^{2n+1})} 
   { - \frac{1}{n}\jj_n(\d \wt) -  j_n(\d \ww_\tau) + O(\d^2)} = O\Big(\frac{\d^{2n+1}}{- \frac{1}{n}\jj_n(\d \wt) -  j_n(\d \ww_\tau)}\Big)\,,\q 
   \end{align} 
   and 
   \begin{align} \label{eq:estforcoeff_2}
         &\frac{ \frac{1}{ 1+\tau}   \jj_n(\d \ww_\tau) j_n(\d \ww) -  j_n(\d \ww_\tau ) \jj_n(\d \ww) }{ \frac{1}{ 1+\tau} h_n^{(1)}(\d \ww) \jj_n(\d \ww_\tau) -  j_n(\d \ww_\tau) \hh_n (\d \ww)}\notag\\= &\frac{ \frac{1}{ 1+\tau}   \jj_n(\d \ww_\tau) (j_n(\d \ww) /\hh_n (\d \ww)) -  j_n(\d \ww_\tau ) (\jj_n(\d \ww) /\hh_n (\d \ww)) }{ \frac{1}{ 1+\tau}( h_n^{(1)}(\d \ww) /\hh_n (\d \ww)) \jj_n(\d \ww_\tau) -  j_n(\d \ww_\tau)} \notag\\
    = &\frac{ \jj_n(\d \ww_\tau) O(\d^{2n+3}) -  j_n(\d \ww_\tau ) (\jj_n(\d \ww) /\hh_n (\d \ww))  }{ O(\d^{2})  -  j_n(\d \ww_\tau)} = \frac{\jj_n(\d \ww)}{\hh_n(\d\ww)} + O\Big(\frac{\d^{2n+3}}{j_n(\d \wt)}\Big)\,.
   \end{align}
\end{enumerate}
 \normalsize
In particular, for $n = 1$, we have a sharper estimate than \eqref{eq:estforcoeff_1}:
\small
\begin{align} \label{eq:estforcoeff_3}
    &\frac{- j_1(\d \ww_\tau)\jj_1(\d \ww) + \jj_1(\d \ww_\tau) j_1(\d \ww)}{h_1^{(1)}(\d \ww) \jj_1(\d \ww_\tau) -  j_1(\d \ww_\tau) \hh_1(\d\ww)} = \frac{i}{3}(\d \ww)^3 \frac{\jj_1(\d \ww_\tau) - 2 j_1(\d \ww_\tau) }{\jj_1(\d\wt) + j_1(\d \wt)} + O\Big(\frac{\d^5}{\jj_1(\d\wt) + j_1(\d \wt)}\Big)\,.
   \end{align} 
 \normalsize
\section{Multipole expansion for the Laplacian vector fields}\label{app:B}

We consider a three--dimensional vector field $E$ satisfying \small
\begin{equation} \label{eq:veclap}
    \left\{
    \begin{aligned}
        &\Delta E = 0\,,\ \ddiv E = 0\,, & |x|& > 1\,, \\
        &E = O(|x|^{-2})\,, & |x|& \to \infty\,.
    \end{aligned}    
    \right.
\end{equation}
 \normalsize
By the elliptic regularity theory, $E$ is analytic on $|x| > 1$. We shall find the multipole expansion for $E$ via the Debye potentials. Since many calculations involved are similar to those for the case of Maxwell's equations and can be found in \cite[pp.\,235-236]{monk2003finite}, we will not give all the details but only outline the main steps.

Given a smooth scalar function $u$, the associated Debye potential is defined by 
$ 
    E_u = \curl (u x)\,.
$ 
It is easy to check, by a direct calculation, that if $u$ is harmonic, i.e., $\Delta u = 0$ , then the Debye potential $E_u$ satisfies \small
\begin{equation*} 
    \curl \curl E_u  = 0\,.
\end{equation*}
 \normalsize
This, along with $\ddiv E_u = 0$, 
yields  $E_u$ is harmonic.

Moreover, for any scalar harmonic function $u$ in a neighborhood of infinity, we can expand it by $$ r^{-(n+1)}\ynm(\h{x})\,,\q n =0, 1,2,\cdots\,, \ m=-n,\cdots,n\,, $$ with $r = |x|$. The corresponding Debye potentials and their curls are given by \small
\begin{equation*}
    \eth(x) = \curl \Big(\frac{1}{r^{n+1}}\ynm(\h{x})x\Big)\,, \quad  \curl \eth(x) = \curl \curl \Big(\frac{1}{r^{n+1}}\ynm(\h{x})x\Big)\,,
\end{equation*}  \normalsize
for $n = 1,2,\cdots\,, \ m=-n,\cdots,n$,
which can be calculated explicitly:\small
\begin{align}
&\eth (x) = - \sqrt{n(n+1)} \frac{1}{r^{n+1}} \vnm(\h{x})\,, \label{def:dpeth}\\
&\curl \eth (x) =\frac{n(n+1)}{r^{n+2}}\ynm(\h{x}) \h{x} - \sqrt{n(n+1)} \frac{n}{r^{n+2}}\unm(\h{x})\,; \label{def:cdpeth}
\end{align}
 \normalsize
see the end of this section for the details. To summarize, for any harmonic function $u$, both the Debye potential $E_u$ and its curl solve the equation \eqref{eq:veclap}, and they can be spanned by $\eth$ in \eqref{def:dpeth} and $\curl \eth$ in \eqref{def:cdpeth}, respectively. We remark that $\curl \eth$ are nothing else but, up to factors $-n$, the electrostatic multipole fields
$
    \na \left(r^{-(n+1)} Y_n^m(\h{x})\right)\,.
$
\begin{proposition} \label{prop:expout}
Any solution $E$ to the equation \eqref{eq:veclap} has the following series expansion: \small
\begin{equation*}
    E(x) = \sum_{n=1}^\infty\sum_{m=-n}^n\gnm \eth(x) + \enm \curl \eth(x) + d_0 \frac{\h{x}}{r^2}\,,
\end{equation*}
 \normalsize
where $\gnm$, $\enm$ and $d_0$ are complex constants with  \small
\begin{equation} \label{def:doo}
    d_0 = \frac{1}{4 \pi} \int_\S \h{x}\dd E(\h{x}) d\sigma(\h{x})\,.
\end{equation}
 \normalsize
\end{proposition}

\begin{proof}
Since $\{\vnm\}$ and $\{\unm\}$ form an orthogonal basis for ${\bf L}^2_{\rm T}(\S)$, the 
tangential trace of a solution $E$ on $\S$ admits the following expansion: \small
\begin{equation} \label{eq:traceout}
    \h{x} \t E(x)|_{\S} = \sum_{n=1}^\infty\sqrt{n(n+1)}\sum_{m=-n}^n \big(\gnm \unm(\h{x}) -  n\enm\vnm(\h{x})\big) \,.
\end{equation}
 \normalsize
We define the constant $d_0$ by \eqref{def:doo}. We next show that $E$ equals to
\small
\begin{equation} \label{eq:traceout_2}
    \w{E}:= \sum_{n=1}^\infty\sum_{m=-n}^n\gnm \eth(x) + \enm \curl \eth(x) + d_0 \frac{\h{x}}{r^2}\,,
\end{equation}
 \normalsize
with $\gnm$ and $\enm$ being given in \eqref{eq:traceout}. This will complete the proof.

For this, by Green's formula, it is easy to check 
\small
\begin{align} \label{eq:auxforuniq}
    \int_{1\le |x| \le R} |\curl (E-\w{E})|^2 dx = \int_{\S_R} \n \t (\overline{E - \w{E}})\dd \curl (E - \w{E}) d\sigma\,.
\end{align}
 \normalsize
To proceed, since each Cartesian component of $E$ is harmonic at infinity, we know from \cite[Proposition 2.75]{folland1995introduction} that 
$ 
    \p_j E_i = O(|x|^{-2})\,. 
$ 
We also note from \eqref{eq:veclap} and \eqref{eq:traceout_2} that $E, \w{E}, \curl \w{E} = O(|x|^{-2})$. It then follows from \eqref{eq:auxforuniq} that \small
\begin{align} 
    \int_{1\le |x| \le R} |\curl (E-\w{E})|^2 dx = O\Big(\frac{1}{R^2}\Big) \q \text{as}\ R \to \infty\,,
\end{align}
 \normalsize
which implies $\curl (E - \w{E}) = 0$. Recalling \cite[Theorem 3.37]{monk2003finite} and $\h{x} \t (E-\w{E})|_\S = 0$, we have \small
\begin{equation*}
    E - \w{E} = \na \phi \q \text{for some}\ \phi \in H^1_{\rm loc}(\R^3\backslash \overline{B(0,1)})\ \text{with}\ \phi|_\S = \text{constant}\,.
\end{equation*}
 \normalsize
We now prove that $\phi$ must be zero. Note that $\phi$ solves the exterior Dirichlet problem for the Laplacian. It follows that $\phi$ has the form: \small
\begin{equation} \label{auxeqphi0}
    \phi =  c \frac{1}{r} Y_0^0(\h{x})\,, \q c \in \C\,.
\end{equation}
 \normalsize
By \eqref{def:doo}, \eqref{eq:traceout_2} and \eqref{auxeqphi0},  we derive \small
\begin{equation*}
   0 = \int_\S \h{x} \dd (E-\w{E}) d\sigma(\h{x}) = - c\int_{\S} \frac{1}{r^2}Y_0^0(\h{x}) d\sigma(\h{x})\,, 
\end{equation*}
 \normalsize
which gives $c = 0$. The proof is complete. 
\end{proof}

We end this section with the detailed calculations for \eqref{def:dpeth} and \eqref{def:cdpeth}, which mainly rely on the following fundamental formula for the $\ccurl$ operator \cite[p.187]{nedelec2001acoustic}: \small
\begin{equation*}
    \curl E = \na_{\S_r} \dd (E \t \h{x})\h{x} + \vec{{\rm curl}}_{\S_r}(E \dd \h{x}) - \frac{\p}{\p r}(E \t \h{x}) - \frac{1}{r}(E \t \h{x})\,.
\end{equation*}
 \normalsize
For $\eth$, we have \small
\begin{align*}
    \eth(x) = &\curl \Big( \frac{1}{r^{n+1}} \ynm(\h{x}) x \Big) = \vec{{\rm curl}}_{\S_r} \Big(\frac{1}{r^{n+1}}\ynm(\h{x}) r \Big) \\
  = & \frac{1}{r^n}\vec{{\rm curl}}_{\S_r} \ynm(\h{x}) = \frac{1}{r^{n+1}} \na_\S \ynm(\h{x}) \t \h{x}
  =  - \sqrt{n(n+1)} \frac{1}{r^{n+1}} \vnm(\h{x})\,.
\end{align*}
 \normalsize
 Similarly, it holds for $\curl \eth$ that \small
\begin{align*}
 -\frac{1}{\sqrt{n(n+1)}} \curl \eth(x) &=  \curl \Big(\frac{1}{r^{n+1}} \vnm(\h{x})\Big) =\frac{1}{r^{n+2}} \na_\S \dd \unm(\h{x}) \h{x} 
      - \frac{\p}{\p r} \Big(\frac{1}{r^{n+1}} \unm(\h{x})\Big) - \frac{1}{r}\frac{1}{r^{n+1}} \unm(\h{x})\\
    = &  -\sqrt{n(n+1)}\frac{1}{r^{n+2}} \ynm(\h{x})\h{x} + \frac{n}{r^{n+2}} \unm(\h{x})\,.
\end{align*}
 \normalsize
\section{Some proofs}

\subsection{Proof of Proposition \ref{prop:ledsolmain}}  \label{app:prfapp}

We first approximate $\pd \w{E}$ up to the second order in terms of $\d$, for which, thanks to \eqref{est:neum_1} and \eqref{auxeq:targetsysrhs}, we only need to consider $\AA_0^{-1}$. By \eqref{eq:invera0} with $(f,g) = (\pd \w{E}^i, \tau^{-1}\pw\w{E}^i )$, we derive 
\small
    \begin{align} \label{eq:ppeladpw_0}
         \AA_0(\ww)^{-1}\mm \pd \w{E}^i \\ \tau^{-1}\pw\w{E}^i \nn
          =& \mm \frac{-\ww_0}{2(\ww-\ww_0)}\poo + \rr(\ww) & 0 \\ 0 & I \nn   \mm \pd \w{E}^i - \tau^{-1} \ww^2 \kbdw \tbb^{-1} \pw\w{E}^i \\ - \tau^{-1}\tbb^{-1}\pw\w{E}^i \nn \notag\\ 
          = & \mm \frac{-\ww_0}{2(\ww-\ww_0)}\poo \w{E}^i + \frac{\ww_0  \tau^{-1} \ww^2 }{2(\ww-\ww_0)} \poo \kbdw \tbb^{-1} \pw\w{E}^i + \d \rr(\ww) \\ - \tau^{-1}\tbb^{-1} \pw \w{E}^i 
           \nn \,,
    \end{align}
\normalsize
where we have used $\pd \w{E}^i = O(\d)$ and $\tau = \d^{-2}$.

We next consider $\AA_0^{-1}\AA_2\AA_0^{-1}$ to find the third--order approximation for $\pw \w{E}$. We calculate, by \eqref{eq:ppeladpw_0},
\small
    \begin{align} \label{eq:ppepw_2mid} 
       \left(\AA_0^{-1}\AA_2\AA_0^{-1} \mm \pd \w{E}^i \\ \tau^{-1} \pw \w{E}^i \nn\right)_2  = &- (\tbb|_W)^{-1}\left(\AA_2\AA_0^{-1} \mm \pd \w{E}^i \\ \tau^{-1} \pw \w{E}^i \nn \right)_{2}\notag \\ 
         = &- (\tbb|_W)^{-1}\left\{\left(\AA_2\right)_{2,1}\left(\AA_0^{-1} \mm \pd \w{E}^i \\ \tau^{-1} \pw \w{E}^i \nn \right)_{1} + \left(\AA_2\right)_{2,2}\left(\AA_0^{-1} \mm \pd \w{E}^i \\ \tau^{-1} \pw \w{E}^i \nn\right)_{2}\right\}\notag\\
         = &- (\tbb|_W)^{-1}\left\{\left(\AA_2\right)_{2,1}\left(\frac{-\ww_0}{2(\ww-\ww_0)}\poo \w{E}^i + \d \rr(\ww)  \right) + O\Big(\frac{\d^2}{\ww - \ww_0}\Big)  \right\}\,. 
    \end{align}
\normalsize    
It is easy to compute 
 $ 
    \left(\AA_2(\ww)\right)_{2,1} = - \ww^2 \kbwd\,.
$ 
    Then it follows  from \eqref{eq:ppepw_2mid} that \small
    \begin{equation} \label{eq:ppepw_2}
       \left(\AA_0^{-1}\AA_2\AA_0^{-1} \mm \pd \w{E}^i \\ \tau^{-1} \pw \w{E}^i \nn \right)_2 =  - \frac{\ww^2 \ww_0}{2(\ww-\ww_0)} (\tbb|_W)^{-1}\kbwd \P_{\lad_0} \w{E}^i  + \d \rr(\ww) +  O\Big(\frac{\d^2}{\ww - \ww_0}\Big)  \,.
    \end{equation}
    \normalsize
The proof is complete by the following error estimate:    
\small
\begin{align} \label{eq:neuexpad_rigo}
\Big \|\Big(\big(\AA(\d,\ww)^{-1} - \AA_0^{-1} + \d^2 \AA_0^{-1}\AA_2\AA_0^{-1}\big)\mm \pd \w{E}^i \\ \tau^{-1} \pw \w{E}^i  \nn \Big)_2 \Big\|_{{\bf L}^2(D)} \lesssim \frac{\d^4}{|\ww - \ww_0|} + \frac{\d^5}{|\ww - \ww_0|^2}\,,
\end{align}    
which is derived from \eqref{est:neum_2} and \eqref{auxeqq_neu_1}. 
\normalsize

\subsection{Proof of Theorem \ref{thm:eigvalspher}} \label{app:eigvalspher} 

By the series representations  \eqref{eq:solin_1}, \eqref{eq:solin_2} and \eqref{eq:solout}, we first calculate the tangential traces of $u$ on $\S$ from inside and outside $B(0,1)$:\small
\begin{align*}
   \h{x} \t u(x)|_- = &\sum_{n=1}^\infty \sqrt{n(n+1)} \sum_{m=-n}^n \big(\frac{c_{n,m}}{n}  - \frac{\bnm}{ik} \jj_n(k)\big)\vnm(\h{x}) + \anm j_n(k)\unm(\h{x})\,,
\end{align*}
 \normalsize
and  \small
\begin{equation*}
    \h{x} \t u(x)|_+ = \sum_{n=1}^\infty \sqrt{n(n+1)}\sum_{m=-n}^n \gnm \unm(\h{x}) - n\enm  \vnm(\h{x})\,,
\end{equation*}
 \normalsize
respectively.
Similarly, we have the following tangential traces of $\curl u$ on $\S$:\small
\begin{equation*}
    \h{x} \t \curl u(x)|_- = \sum_{n=1}^\infty \sqrt{n(n+1)}\sum_{m=-n}^n \jj_n(k)\anm \vnm(\h{x}) +  ik j_n(k) \bnm\unm(\h{x})\,,
\end{equation*}
 \normalsize
and \small
\begin{align*}
    \h{x} \t \curl E(x)|_+ = - \sum_{n=1}^\infty  \sqrt{n(n+1)}\sum_{m=-n}^n n \gnm \vnm(\h{x})\,.
\end{align*}
 \normalsize
Matching the tangential traces of $u$ and $\curl u$ from inside and outside $B(0,1)$ gives \small
\begin{equation} \label{eq:coeff1}
    \mm 
  j_n(k) & -1 \\ 
 \jj_n(k) & n
 \nn 
 \mm \anm \\ \gnm \nn = \mm 0 \\0\nn,
\end{equation}
 \normalsize
and 
\small
\begin{equation} \label{eq:coeff2}
    \mm 
    j_n(k) & 0 \\ 
  (ik)^{-1} \jj_n(k) & -n
   \nn 
   \mm \bnm \\ \enm \nn = \mm 0 \\ n^{-1} c_{n,m}  \nn .
\end{equation}
 \normalsize
To complement the second linear system which involves three variables, we match the normal traces of $u$ inside and outside the domain, which yields
$  c_{n,m} = \enm n(n+1)
$. Then, by substituting it into \eqref{eq:coeff2}, we obtain
\small
\begin{equation} \label{eq:coeff3}
    \mm 
    j_n(k) & 0 \\ 
  (ik)^{-1}\jj_n(k) & -2n-1
   \nn 
   \mm \bnm \\ \enm \nn = \mm 0\\ 0\nn\,.
\end{equation}
 \normalsize



It is easy to observe that the equation \eqref{eq:eigsphere} with $D = B(0,1)$ has nontrivial solutions if and only if one of the linear systems \eqref{eq:coeff1} and \eqref{eq:coeff3} is singular for some $n,m$, equivalently, one of the determinants of \eqref{eq:coeff1} and \eqref{eq:coeff3} vanishes:\small
\begin{equation} \label{eq:det_1}
   n j_n(k) + \jj_n(k) = 0\,, 
\end{equation}
 \normalsize
or  
\small
\begin{equation} \label{eq:det_2}
    j_n(k) = 0\,.
\end{equation}
 \normalsize
If \eqref{eq:det_1} holds for some $n$ and $k$, then the system \eqref{eq:coeff1}, as well as \eqref{eq:eigsphere}, has nontrivial solutions, and hence the eigenspace of $\np_{\sss B(0,1)}^{\sss \rm d,d}$ associated with the eigenvalue $\lad = k^{\sss -2}$ contains a linear space spanned by ${\scriptstyle \{\wete\}_{m=-n}^n}$. Similarly, in the case where \eqref{eq:det_2} holds, we can conclude \eqref{eq:coeff3} is singular and $ {\scriptstyle\{\wetm\}_{m = -n}^n }$ are the linearly independent eigenfunctions of $\np_{\sss B(0,1)}^{\sss \rm d,d}$ for the eigenvalue $k^{\sss -2}$. We proceed to claim that \eqref{eq:det_1} is equivalent to  $j_{n-1}(k) = 0$. Indeed, by the recurrence relation of Bessel functions:
$
    j_n'(k) = j_{n-1}(k) - \frac{n+1}{k}j_n(k)\,,
$
we can rewrite \eqref{eq:det_1} as 
$
    0 = (n+1)j_n(k) + k j_{n-1}(k)-(n+1)j_n(k) =  k j_{n-1}(k).
$

To conclude, we have shown that if $k \neq 0$ is a zero of some Bessel function $j_n$ ($n \ge 0$) which must be real, simple, and symmetric with respect to the origin \cite{abramowitz1964handbook}, then $k^{\sss -2} > 0$ is an eigenvalue of $\np_{\sss B(0,1)}^{\sss \rm d,d}$ with the eigenspace containing linearly independent eigenfunctions:
\small
\begin{equation*}
    \Big(\bigcup_{m=-n-1}^{n+1}\w{E}^{TE}_{n+1,m}\Big)\bigcup \Big(\bigcup_{m=-n}^{n}\w{E}^{TM}_{n,m}\Big)\q \text{for}\ n \ge 1\,,
\end{equation*}
\normalsize
and 
\small
\begin{equation*}
    \bigcup_{m=-1}^{1}\w{E}^{TE}_{1,m} \q \text{for}\ n = 0\,.
\end{equation*}
\normalsize
To prove that these eigenfunctions exactly span the corresponding eigenspace, we need the Bourget’s hypothesis \cite{watson1995treatise,petropoulou2003common}, which states that the Bessel functions of the first kind: $J_\mu(z)$ and $J_{\mu + m}(z)$, with $\mu$ being rational and $m$ being a positive integer, can not have common zeros except $z = 0$. 
It readily follows that $j_n(z)$, $n \ge 0$, have no common positive zeros, which completes the proof.


\if \commentflag = \ct

\section{Algebraic equations and functions} \label{app:D}

In this appendix, we briefly recall the fundamentals of algebraic functions, following \cite{ahlfors1966complex,baumgartel1985analytic,suwa2007introduction}. Let 
$\mathcal{O}_n$ be the set of analytic functions in some neighborhood of $0 \in \C^n$. We consider the ring of polynomials
in a complex variable $w$ with the coefficients in $\mathcal{O}_{n-1}$: \begin{align*}
     \mathcal{O}_{n-1}[w]:= \big\{f(z,w) = \sum_{i = 0}^k a_i(z) w^i \,; \ a_i(z) \in \mathcal{O}_{n-1}\,,\ k \in \mathbb{N}\cup \{0\}\big\} \,,
\end{align*}
and define the Weierstrass polynomials in $w$ of degree $k$:
\begin{align} \label{def:weipoly}
    f(z,w) = a_0(z) + a_1(z) w  + \cdots + a_{k - 1}(z) w^{k-1} + w^k \in \mathcal{O}_{n-1}[w]  \q \text{with}\ a_i(0) = 0\,.
\end{align}




We are interested in the zero set of an irreducible monic polynomial $f \in \mathcal{O}_{1}[w]$ of degree $k$. The main result is that there is a one-to-one correspondence between the irreducible monic polynomials and the $k$--valued algebraic functions; see Theorem 1 in \cite[p.\,403]{baumgartel1985analytic}, Theorem 2 in \cite[p.\,404]{baumgartel1985analytic}, and Theorem 4 in \cite[p.\,306]{ahlfors1966complex}. The following special case \cite[p.\,406]{baumgartel1985analytic} is enough for our use. 

\begin{lemma} \label{prop:corres}
Let $p$ be an irreducible Weierstrass polynomial in $w$ of degree $k$:
\begin{equation*}
    p(z,w) = a_0(z) + a_1(z) w + \cdots + a_{k - 1}(z) w^{k-1} + w^k\,, 
\end{equation*}
where $a_i(z) \in \mathcal{O}_1$ are analytic functions on $|z| < \ep$ with $a_i(0) = 0$. We assume that $p(z,w)$ has only simple roots with respect to $w$ for $0 < |z | < \ep$. Then $f(z) := \{w\,;\ p(z,w) = 0\}$ is a $k$--valued algebraic function on the disk $|z| < \ep$ with $0$ being its only (algebraic) singularity. 
In particular, the Puiseux series expansion holds: 
\begin{align} \label{eq;puiseries}
      f(z) = \sum_{j = 0}^\infty c_j z^{j/k}\,, \q c_0 = 0\,,  \q |z| < \ep\,. 
\end{align}
 Moreover, substituting the $k$ different roots of $z$ into \eqref{eq;puiseries} yields exactly $k$ (different) values of $f$ at $z$. 
\end{lemma}


We end this appendix with the concept of algebraic functions, which is necessary for understanding the above lemma. 
Adopting the terminology in \cite{ahlfors1966complex}, given a domain $D \subset \C$ and an analytic function $f$ at $a \in D$, we say that the set of all the analytic function elements obtained by the analytic continuation of $f$ in $D$ is a \emph{global analytic function}, which may be multivalued; see \cite[p.\,285]{ahlfors1966complex}  for the rigorous statement and also \cite[p.\,390]{baumgartel1985analytic} for an equivalent definition. Then we can define the algebraic functions as follows; see \cite[p.\,394]{baumgartel1985analytic} and \cite[p.\,306]{ahlfors1966complex}. 
\begin{definition}
    A $p$--valued global analytic function $f$ on $D$ is called algebraic if its singular points are at most algebraic, which means 
    \begin{enumerate}
    \item the set of algebraic singularities $S$ must be isolated in $D$. Hence, $D \backslash S$ is a subdomain of $D$, and there are exactly $p$ values of $f$ at each point of $D \backslash S$.  
    \item for a singular point $b$ of $f$ in $D$ and a $q$--valued branch of $f$ in a neighborhood of $b$, the point $b$ is a algebraic branch point, that is, near $b$, the branch $f$ can be represented as
    \begin{align*}
        f(z) = h((z - b)^{1/q})\,, \q |z - b| < r\,,  
    \end{align*}
    where $h$ is a single--valued analytic function. Then the power series expansion of $h$ at $b$ gives us the Puiseux series expansion of $f$.
\end{enumerate}
\end{definition}
\fi

\end{document}